\documentclass[11pt]{article}
\usepackage[letterpaper, margin=1in]{geometry}

\usepackage{makecell}

\usepackage{amsmath,amssymb,bbm,amsthm}
\usepackage{thm-restate,color,xcolor,xspace}
\usepackage{graphicx}
\usepackage[normalem]{ulem}  
\usepackage{color,xcolor,xspace}
\usepackage[hypertexnames=false]{hyperref}
\usepackage{cleveref}
\usepackage[linesnumbered,ruled,vlined]{algorithm2e}

\usepackage{tikz}
\usetikzlibrary{decorations.pathreplacing}

\usepackage{enumitem}
\usepackage{cite}

\usepackage{array}

\newtheorem{theorem}{Theorem}[section]
\newtheorem{lemma}[theorem]{Lemma}
\newtheorem{proposition}[theorem]{Proposition}

\newtheorem{definition}[theorem]{Definition}
\newtheorem{corollary}[theorem]{Corollary}
\newtheorem{remark}[theorem]{Remark}
\newtheorem{observation}[theorem]{Observation}
\newtheorem{claim}[theorem]{Claim}

\DeclareMathOperator\poly{poly}
\newcommand{\abs}[1]{\left \lvert #1 \right \rvert}
\newcommand{\floor}[1]{\left\lfloor #1 \right\rfloor}
\newcommand{\ceil}[1]{\left\lceil #1 \right\rceil}
\newcommand{\E}{\mathbb{E}}
\newcommand{\M}{\mathcal{M}}
\newcommand{\bbP}{\mathbb{P}}

\SetKwInput{KwData}{Input}
\SetKwInput{KwResult}{Output}

\newcommand{\maximummatching}{\textsc{MaximumMatching}}
\newcommand{\gadget}[1]{g^{#1}}
\newcommand{\subgadget}[1]{h^{#1}}
\newcommand{\mc}{\mathcal}

\newcommand{\Dup}{D_+}

\newcommand{\rr}{Regular-Graph Preservation}
\newcommand{\rrl}{Regular-Graph Preservation Lemma}

\title{On the Randomized Locality of Matching Problems in Regular Graphs}

\author{
  Seri Khoury\thanks{INSAIT, Sofia University “St. Kliment Ohridski”. \texttt{seri.khoury@insait.ai}} 
  \and
  Manish Purohit\thanks{Google Research. \texttt{purohitmanish@gmail.com}}
  \and
  Aaron Schild\thanks{Google Research. \texttt{aaron.schild@gmail.com}}
  \and
  Joshua Wang\thanks{Google Research. \texttt{josh.w0@gmail.com}}
}

\date{}

\begin{document}

\maketitle
\begin{abstract}

The main goal in distributed symmetry-breaking is to understand the locality of problems; i.e., the radius of the neighborhood that a node needs to explore in order to arrive at its part of a global solution. In this work, we study the locality of matching problems in the family of regular graphs, which is one of the main benchmarks for establishing lower bounds on the locality of symmetry-breaking problems, as well as for obtaining classification results. Our main results are summarized as follows:

 \begin{enumerate}
     \item \textbf{Approximate matching:} We develop randomized algorithms to show that \((1 + \epsilon)\)-approximate matching in regular graphs is truly local; i.e., the locality depends only on \(\epsilon\) and is independent of all other graph parameters. Furthermore, as long as the degree \(\Delta\) is not very small (namely, as long as \(\Delta \geq \operatorname{poly}(1/\epsilon)\)), this dependence is only logarithmic in $1/\epsilon$. This stands in sharp contrast to maximal matching in regular graphs which requires some dependence on the number of nodes \(n\) or the degree \(\Delta\). We show matching lower bounds for both results.
     \item \textbf{Maximal matching:} Our techniques further allow us to establish a strong separation between the node-averaged complexity and worst-case complexity of maximal matching in regular graphs, by showing that the former is only \(O(1)\).
 \end{enumerate}

Central to our main technical contribution is a novel martingale-based analysis for the \(\approx 40\)-year-old algorithm by Luby. In particular, our analysis shows that applying one round of Luby’s algorithm on the line graph of a \(\Delta\)-regular graph results in an almost \(\Delta/2\)-regular graph.
\end{abstract}

\newpage

\section{Introduction}

Matching problems, such as maximal or maximum matching, have garnered significant attention across several computational models, including the classical sequential model~\cite{BrandLNPSS0W20, Madry13, HopcroftK73, GoelKK13, Yuster13}; dynamic networks~\cite{BhattacharyaKSW23, ArarCCSW18, 0001LSSS19, Solomon16, GuptaP13, BhattacharyaKS23}; streaming algorithms~\cite{FischerMU22, PazS19, AssadiS23, AssadiBKL23, BuryGMMSVZ19, McGregor05}; online algorithms~\cite{BuchbinderNW23, GuptaGPW19, KarpVV90, KarandeMT11, Mehta13}; and distributed computing~\cite{GhaffariHK18, GhaffariKMU18, LotkerPP15, AhmadiK20, AhmadiKO18, Fischer20, Bar-YehudaCGS17, BenbasatKS19, LotkerPR09, Harris19, CzygrinowH03, EvenMR15, IzumiKY24, 0002MNSU25}.

In the classical LOCAL model of distributed computing, there is a network of $n$ nodes that can communicate via synchronized communication rounds. In each round, a node can send an unbounded-size message to each of its neighbors. The goal for these nodes is to solve a graph problem (e.g., find a large matching) while minimizing the number of communication rounds. We refer
to the number of rounds required to solve a problem in the LOCAL model as the \emph{locality} of that problem.
Since each node can only interact with the nodes in its $r$-hop neighborhood in $r$ rounds, the locality of a
problem is the radius of the neighborhood that each node needs to explore to determine its part in the global
solution (e.g., whether it is matched and if so to which neighbor).

Typically, the locality of matching problems depends on global graph parameters, such as the number of nodes \(n\) or the maximum degree \(\Delta\). For instance, a simple algorithm by Israeli and Itai~\cite{IsraelI86} finds a maximal matching in \(O(\log n)\) rounds in the LOCAL model, which constitutes a 2-approximate maximum matching in unweighted graphs (or approximate matching, for short). Later, Barenboim et~al.\ \cite{BarenboimEPS12}  showed that this logarithmic dependence on \(n\) can be substituted with a logarithmic dependence on \(\Delta\) by providing an \(O(\log \Delta + \mathrm{poly} \log \log n)\)-round algorithm. This upper bound also applies to finding a \((1+\epsilon)\)-approximate matching (while incurring a multiplicative \(\mathrm{poly}(1/\epsilon)\) factor in the number of rounds), as shown by Harris~\cite{Harris19}.

All of the above results are for randomized algorithms that succeed with high probability.\footnote{Throughout the paper, we say that an algorithm succeeds with high probability if it succeeds with probability $1-1/n^c$ for some arbitrarily large constant $c>0$.} If instead one is satisfied with a matching that is only large in expectation, then we can combine recent algorithms for fractional matching, rounding, and amplification to shave a $\log\log\Delta$ factor from the round complexity~\cite{Bar-YehudaCGS17, Bar-YehudaCS16, FischerMU22, GGKMR18}. For deterministic algorithms, the primary question is whether we can approach the state-of-the-art logarithmic bound that applies to randomized ones.\footnote{The current state-of-the-art deterministic algorithms take $\tilde{O}(\log^{5/3}n)$ rounds for maximal matching~\cite{ecomposition}, and $\tilde{O}(\log^{4/3}n\cdot \poly(1/\epsilon))$ rounds for $(1+\epsilon)$-approximate matching~\cite{G23,FischerMU22}, where $\tilde{O}(x)$ hides $\poly\log x$ factors.} For a more detailed discussion on deterministic algorithms, we refer the reader to the excellent surveys by Suomela~\cite{Suomela13} and Rozhon~\cite{abs-2406-19430}. 

Whether there exist $o(\log n)$-round algorithms that succeed with high probability for maximal or \((1+\epsilon)\)-approximate matching in general graphs remains one of the main open questions in the field. On the other hand, the best known lower bound as a function of $n$ is $\Omega(\sqrt{\log n/\log\log n})$ rounds, which was shown by Kuhn, et~al.\ \cite{KuhnMW16}, and it applies all the way up to $\poly(\log n)$-approximation. It is worth noting that for maximal matching, the best known lower bound has recently been improved to $\Omega(\sqrt{\log n})$ as a function of $n$ \cite{KS25}, but this lower bound does not extend to $(1+\epsilon)$-approximate matching.

\subparagraph{Regular Graphs} We say that a graph is regular if all the nodes have the same degree. The family of regular graphs has received much attention as a natural benchmark for studying the complexity of various fundamental problems (e.g.,~\cite{Babai14,BabaiCSTW13,Babai80,Spielman96,Kucera87,CarrollGT09,AlonM21,AlonP11,AlonBK10,AlonCHKRS10,Zamir23,Golowich23,AlonFK84,AlonFK84a}). Regular graphs admit nearly perfect matchings (see, for instance,~\cite{FlaxmanH07}), and the problem of finding an approximate matching in regular graphs has been widely studied in many different computational models~\cite{GoelKK13, CohenW18, AggarwalMSZ03, MotwaniR95, Alon03, KahnK98, DahlhausK92, GoldwasserG17, ColeH82, Yuster13}. In the LOCAL model, regular graphs serve as the standard benchmark for implementing the \emph{Round Elimination} technique for proving lower bounds~\cite{BalliuBHORS21,Bo20,BalliuB0O24,Balliu0KO23,BalliuBO22,Balliu0KO21,Balliu0KO22,Brandt19,BrandtFHKLRSU16}, which will be discussed shortly, and also for obtaining classifications of locality results (e.g., in paths and cycles~\cite{BrandtHKLOPRSU17,Balliu0CORS19}, grids~\cite{NaorS95,BrandtHKLOPRSU17,ChangKP19,GrunauR022}, and regular trees~\cite{BalliuBCOSST23,Balliu0COSS22}).\footnote{For more on classification of locality results, as well as ones in general graphs, we refer the reader to the survey by Rozhon~\cite{abs-2406-19430}.}

The locality of some problems in regular graphs still depends on graph parameters such as the number of nodes $n$ or the degree $\Delta$, while for others, it is independent of any such parameter. Interestingly, as we discuss next, maximal matching falls into the former category, whereas $(1 + \epsilon)$-approximate matching is in the latter.

For maximal matching, the approximately 40-year-old $O(\log n)$ upper bound by Israeli and Itai~\cite{IsraelI86} remains the best known even for regular graphs. On the hardness side, Khoury and Schild \cite{KS25} showed that maximal matching requires $\Omega(\min\{\log\Delta, \sqrt{\log n}\})$ rounds on $\Delta$-regular graphs. For low-degree regular graphs, Balliu et~al.\ \cite{BalliuBHORS21} showed that maximal matching in regular graphs requires $\Omega(\min\{\Delta, \log\log n / \log\log\log n\})$ rounds for randomized algorithms and $\Omega(\min\{\Delta, \log n / \log\log n\})$ rounds for deterministic ones. Even for very low-degree regular graphs (e.g. cycles), maximal matching still requires $\Omega(\log^{*} n)$ rounds~\cite{Linial92,Naor91}, even when using randomness.

On the other hand, approximate matching in regular graphs has received far less attention. The hard instances for approximate matching developed by Kuhn et~al.\ \cite{KuhnMW16} are very far from being regular. In fact, these instances are also hard for finding an approximate fractional matching, a problem that admits a trivial zero-round solution in regular graphs by simply setting the fractional value for each edge to be $1/\Delta$. This fractional matching can be rounded to an $O(1)$-approximate integral matching in one round by using a simple sampling technique~\cite{GGKMR18}.

This raises the question: where does $(1+\epsilon)$-approximate matching fall on the spectrum of locality in regular graphs? Is it closer to approximate fractional matching, or does it require some  dependence on $n$ or $\Delta$ similar to maximal matching?\footnote{Note that the amplification technique of~\cite{FischerMU22} cannot be used to amplify the constant-approximation factor to $(1+\epsilon)$ in regular graphs while using $\poly(1/\epsilon)$ phases of amplification. This is because the technique is designed for general graphs, and applying it to a regular graph can result in a non-regular graph after the first step. Consequently, the amplification algorithm would need to use an algorithm for constant-approximate matching in \emph{general graphs}, which requires some dependence on $n$ or $\Delta$.} 


\subsection{Our Results} 

\subparagraph{Approximate Matching} In this work, our first result is a simple algorithm that finds a $(1+\epsilon)$-approximate matching in regular graphs without any dependence on graph parameters.\footnote{All our upper bounds apply also to almost regular graphs, where the degrees of all the nodes are within a $(1\pm o(1))$-multiplicative factor from each other. Moreover, the case where $\epsilon \le n^{-1/20}$ can be handled in $\poly(1/\epsilon)$ rounds by transmitting the entire graph to all nodes.} For constant values of $\epsilon$, the algorithm also works in the more restricted CONGEST model, where the size of the messages is bounded by $O(\log n)$ bits. 

\begin{restatable}{theorem}{thmpolywarmup}\label{thm:warmup}
Let $n$ be a positive integer, and let $\epsilon \in (n^{-1/20}, 1/2)$ be an accuracy parameter. There is an $O(\epsilon^{-5} \log (1/\epsilon))$-round LOCAL algorithm that finds a $(1 + \epsilon)$-approximate maximum matching in $n$-node regular graphs, with high probability. The algorithm works in the CONGEST model for constant values of $\epsilon$.
\end{restatable}

While Theorem~\ref{thm:warmup} advances our understanding of $(1+\epsilon)$-approximate matching in regular graphs, the polynomial dependence on $1/\epsilon$ is far from desirable. Comparing with the $\tilde{O}(\frac{\log \Delta}{\epsilon^3} + \text{polylog}(1/\epsilon, \log\log n))$ round algorithm  of Harris~\cite{Harris19}, the algorithm from Theorem~\ref{thm:warmup} is better only in the regime $\epsilon > 1/\sqrt{\log \Delta}$.
In our next result, we present an exponentially faster algorithm as long as $\epsilon > 1/\Delta^{c'}$ for some constant $c' > 0$.
 In other words, we present an exponentially faster $(1+\epsilon)$-approximation algorithm for graphs that are not extremely sparse, i.e., when $\Delta > \text{poly}(1/\epsilon)$. We note that the constant $c$ in Theorem~\ref{thm:UpperBoundMain} has not been optimized and can be improved substantially with a more careful analysis. 

\begin{restatable}[Main Result I]{theorem}{thmmain}
\label{thm:UpperBoundMain}
Let $c = 10^5$ and let $\epsilon \in (0, 1)$ be an accuracy parameter. For $\Delta$-regular graphs with $\Delta > (1/\epsilon)^c$ (equivalently; $\epsilon > \Delta^{-1/c}$), there is an $O(\log(1/\epsilon))$-round CONGEST algorithm that finds a $(1+\epsilon)$-approximate maximum matching with high probability.
\end{restatable}

 One may wonder whether the restriction in the theorem about the graph being dense enough is a mere technicality. The following (simple) lower bound presented in Theorem~\ref{thm:LowerInformal} shows that this is not the case, allowing us to establish a separation between dense and sparse regular graphs in which dense ones are easier. Note that, for the same problem, the \emph{opposite} separation holds in general graphs (due to the $\Omega(\min\{\log\Delta/\log\log\Delta,\sqrt{\log n/\log\log n}\})$ lower bound by~\cite{KuhnMW16}, and the upper bounds by~\cite{Harris19}). 
 
\begin{theorem}[Informal version of \Cref{thm:lower-general}]\label{thm:LowerInformal}
  For any degree $\Delta \ge 2$ and error $\epsilon = O(\Delta^{-1})$, any (randomized) LOCAL algorithm that computes a $(1 + \epsilon)$-approximate maximum matching in bipartite $\Delta$-regular graphs with at least $n \ge \Omega(\Delta^{-1}\epsilon^{-1})$ nodes requires $\Omega(\Delta^{-1}\epsilon^{-1})$ rounds.
\end{theorem}

One intuition behind this separation is that in $\Delta$-regular graphs, the number of nearly optimal solutions to maximum matching scales with $\Delta$ (see for instance~\cite{AsratianK99}). This abundance of nearly optimal solutions intuitively facilitates the rapid identification of one using randomness.

Nonetheless, there is a limit to how much the abundance of nearly optimal solutions in dense graphs helps. In fact, Theorem \ref{thm:UpperBoundMain} is essentially optimal thanks to a slight extension of \cite{KS25} that we prove in Appendix \ref{app:LowerDense}:

\begin{restatable}[Main Result I Is Essentially Optimal]{theorem}{thmmainoptimal}\label{thm:LowerDense}
Let $c = 10^5$. For any even degree $\Delta\ge 2$ and error $\epsilon > \max(\Delta^{-1/c}, n^{-1/(c\log\Delta)})$, any (randomized) LOCAL algorithm that takes less than $(\log(1/\epsilon))/c$ rounds produces a $(1 + \epsilon)$-approximate maximum matching in $\Delta$-regular graphs with probability at most $\exp(-\sqrt{n})$.
\end{restatable}

At the heart of the proof of our first main result (Theorem~\ref{thm:UpperBoundMain}) is a novel martingales-based analysis of the classical algorithm by Luby~\cite{Luby86}. In particular, we use this analysis to prove our \emph{\rrl{}} (Lemma~\ref{lem:informal-rrl}), which shows that running a single round of Luby's algorithm on the line graph of a sufficiently dense $\Delta$-regular graph and removing the matched nodes together with their incident edges yields an almost $\approx \Delta/2$-regular graph. Interestingly, this analysis also allows us to prove the following implication for maximal matching.

\subparagraph{Node-Averaged Complexity of Maximal Matching}
When performing a distributed computation, the algorithm may arrive at some parts of its final output before the rest. That is, some nodes know their local result before the algorithm fully terminates on the last few stragglers. The notion of node-averaged complexity \cite{BalliuGKO23, BarenboimT18, ChatterjeeGP20, Feuilloley17} captures this nuance by reasoning about the average over the times at
which the nodes finish their computation and commit to their outputs (we formally define the node-averaged complexity in Section~\ref{sec:model}). Node-averaged complexity captures important applications such as optimizing the total energy consumption of the network (see, for instance, \cite{ChatterjeeGP20} and references within).


In sharp contrast to the worst-case complexity of maximal matching in regular graphs, our martingale-based analysis of Luby's algorithm yields a \emph{constant} node-averaged complexity. This also contrasts with a lower bound of Balliu et~al.\ \cite{BalliuGKO23} who showed that in general graphs, the node-averaged complexity of maximal matching is $\Omega\left(\min\left\{\frac{\log \Delta}{\log \log \Delta},\sqrt{\frac{\log n}{\log\log n}}\right\}\right)$.\footnote{We note that the corresponding edge-averaged complexity of maximal matching is known to be $O(1)$, even in general graphs, due to the classical algorithms of~\cite{Luby86,alon1986fast} (when applied on the line graph) that delete a constant-fraction of the edges in each round~\cite{BalliuGKO23}. However, when an edge arrives at its part of the output (i.e., whether it is matched or not), it does not necessarily help an incident node to arrive at its part of the output, as several edges are not matched. This nuance is fundamental due to the lower bound by~\cite{BalliuGKO23} discussed above for the node-averaged complexity of maximal matching in general graphs.}

\begin{restatable}[Main Result II]{theorem}{thmnodeaveraged}
\label{thm:NodeAveragedMain}
  The node-averaged complexity of finding a maximal matching in $\Delta$-regular graphs is $O(1)$.
\end{restatable}

Theorem~\ref{thm:NodeAveragedMain} implies that, with respect to node-average complexity, maximal matching is easier in regular graphs than in general ones. On the other hand, as discussed earlier, for worst-case complexity, the $O(\log n)$ bound by~\cite{IsraelI86} is still the best known for both regular and general graphs. Whether maximal matching is any easier in regular graphs with respect to worst-case complexity remains an interesting open question.

\subparagraph{Outline of the paper} In Section~\ref{sec:model} we provide some basic definitions and notation. In Section~\ref{sec:tec} we provide a brief technical overview of our results. 
\Cref{sec:reg} and \Cref{sec:localrrl} contain the technical details of our main results (Theorems~\ref{thm:UpperBoundMain} and~\ref{thm:NodeAveragedMain}). \Cref{sec:warmup} provides a $\poly(1/\epsilon)$-round algorithm for approximate matching in general regular graphs.  Finally, \Cref{sec:lower} contains the details of our lower bound construction. We defer some basic definitions and concentration inequalities to \Cref{app:prelims}, along with a new concentration inequality for sums of random variables using shifted martingale analysis.

\subsection{Model and Basic Definitions}\label{sec:model}

\subparagraph{The LOCAL and CONGEST model} In this work we are interested in the LOCAL and CONGEST models of distributed computing~\cite{peleg,Linial92}. In both models, there is a synchronized communication network of $n$ computationally unbounded
nodes that can communicate via communication rounds; this network both defines the communication pattern and serves as the input graph we want to solve our problem on. In each round, each node can send an unbounded-size (in the LOCAL model) or $O(\log n)$-bit (in the CONGEST model) message to each of its neighbors. The goal is to perform a task (e.g., find 
a large matching) while minimizing the number of communication rounds. At the end of all communication rounds, each node needs to know which of its neighbors it is matched with, if any.

\subparagraph{Basic Graph Notations} For a graph $G$, we denote by $V(G)$ the set of nodes in $G$ and by $E(G)$ the set of edges. Given a node $u\in V(G)$, we denote by $N_{G}(u)$ the set of neighbors of $u$ in $G$, and by $N^d_{G}(u)$ the set of nodes at distance exactly $d$ from $u$. When $G$ is clear from the context, we omit the letter $G$ from the notation and use $V,E$ and $N(u)$ for brevity. Throughout the paper, $n$ denotes the number of nodes in the graph. In $\Delta$-regular graphs, all the nodes have the same degree $\Delta$. In this work, we are interested in unweighted and undirected graphs.

\subparagraph{Maximum and maximal matching} A matching $\mathcal{M}$ in a graph $G$ is a set of edges in $E(G)$, where no two edges in $\mathcal{M}$ share a node. A maximum matching is a matching of maximum possible size. A $(1+\epsilon)$-approximate matching in $G$ is a matching $\mathcal{M}$ satisfying $OPT\leq (1+\epsilon)|\mathcal{M}|$, where $OPT$ is the size of a maximum matching. A maximal matching is a matching that is not a strict subset of any other matching.

\subparagraph{Node-averaged Complexity (\cite{BalliuGKO23})} The node-averaged complexity of an algorithm $A$, \(\text{AVG}_V(A)\), is defined as follows.
\[
\text{AVG}_V(A) :=  \max_{G \in \mathcal{G}} \frac{1}{|V(G)|} \cdot \sum_{v \in V(G)} \mathbb{E} \left[ T_v^G(A) \right]
\]
where $T_v^G(A)$ is the time it takes for $v$ to reach its part of the output when running $A$. The node-average complexity of a graph problem is defined as the node-average complexity of the best algorithm $A^*$ that minimizes $\text{AVG}_V(A^*)$.
\section{Technical Overview}
\label{sec:tec}

\subsection{Warmup: A \texorpdfstring{$\poly(1/\epsilon)$}{poly(1/eps)}-Round Algorithm for Regular Graphs}
\label{sec:TechWarmup}

To prove \Cref{thm:warmup}, we use a two stage algorithm. We begin with a $\Delta$-regular graph on $n$ vertices with target error parameter $\epsilon > n^{-1/20}$. The first stage (sampling) involves uniformly and independently sampling edges from the graph with the goal of reducing the degree from $\Delta$ to $\poly(1/\epsilon)$, plus some post-processing. After this stage, we have an (irregular) graph on at most $n$ vertices with degree at most $d = \poly(1/\epsilon)$ and have used up a constant fraction of our error parameter $\epsilon$ (i.e. this restricted graph still has an almost-perfect matching). Our second stage (matching) involves finding a matching in this restricted graph, and its runtime only depends on the error parameter $\epsilon$ and the max degree $d$.

For the sampling stage, uniformly sampling to degree approximately $\epsilon^{-2} \log n$ would result in a graph that retains a near-perfect matching with high probability (via a Chernoff bound plus a union bound), which then when combined with Harris' algorithm~\cite{Harris19} can already give us an $O(\poly(1/\epsilon)\log\log n)$-round algorithm. However, to avoid the dependence on $n$, we need to find a way to reduce this degree even further. Instead, we sample down to degree $\Theta(\epsilon^{-4})$. The resulting subgraph can be very irregular, but we manage to tease out enough structure to make our argument go through. In particular, some small fraction of vertices may have degree exceeding our target $\Theta(\epsilon^{-4})$ by more than a factor two. We use Chernoff with bounded dependence and the matching polytope to argue that stripping out these problematic high-degree vertices still leaves an almost-perfect matching.

For the matching stage, we find a matching in the constructed low-degree subgraph from the sampling stage. The state-of-the-art algorithms of Harris~\cite{Harris19} fit our task, but they have a small runtime dependence on $n$. Instead, we combine the hypergraph framework from~\cite{Harris19} (which was also used in\cite{FischerMU22,DBLP:conf/focs/FischerGK17,Bar-YehudaCGS17}) with some ideas from~\cite{GGKMR18}, as follows. In  $\poly(1/\epsilon)$ phases, we increase the size of the matching in each phase by $\poly(\epsilon)\cdot  n$ edges. The algorithm for each phase finds a large set of disjoint augmenting paths. This is done by first constructing a hypergraph $H$ with the same set of nodes as in the low-degree subgraph, where each hyperedge corresponds to a $1/\epsilon$-length augmenting paths. Then, we find an $O(1/\epsilon)$-approximate fractional matching in the hypergraph by using the algorithms of~\cite{KuhnMW06,Harris19,BenBasatEKS23}, and we round this fractional matching by sampling each hyperedge with probability proportional to its fractional value. By using a similar McDiarmid-type argument as in~\cite{GGKMR18}, we can show that this rounding produces an integral $O(1/\epsilon)$-approximate matching in the hypergraph $H$ with high probability. Furthermore, observe that the maximum degree of a node in the hypergraph $H$ is $\exp(\poly(1/\epsilon))$. Therefore, the algorithms of~\cite{KuhnMW06,Harris19,BenBasatEKS23} for finding a fractional matching in this hypergraph take $O(\poly(1/\epsilon))$ rounds, as desired.

\subsection{Exponentially Faster Algorithm}
\label{sec:OverviewDense}

In this section, we give a brief technical overview for our main result. On bipartite graphs, our algorithm for proving \Cref{thm:UpperBoundMain} simply runs $O(\log(1/\epsilon))$ rounds of Luby's algorithm~\cite{Luby86} that finds a large matching in each round. On general graphs, we first run a color coding step to find a large bipartite almost regular subgraph. While the algorithm is very simple, the main challenge is its analysis. In this work, we provide a new martingale-based analysis for Luby's algorithm.

Since our analysis relies on martingale concentration inequalities, it is more convenient to work with a sequential view of Luby's algorithm. Our key lemma (Lemma~\ref{lem:informal-rrl}) shows that after applying one round of Luby's algorithm (and removing the matched nodes along with incident edges), the remaining graph remains almost regular, i.e., almost all nodes have very similar degrees.
Even with the sequential view, classical martingale concentration inequalities are not sufficient to provide the high probability bounds that we require. We use two techniques, a \emph{shifted martingale trick} and \emph{scaled martingale trick}, in order to provide the requisite bounds. Roughly speaking, we show that the number of matched neighbors of a node behaves similarly to a martingale. Finally, we show that a constant fraction of the nodes are matched in each iteration of Luby's algorithm when the graph is almost regular. Combined with our key lemma, this implies that after $O(\log(1/\epsilon))$ rounds at most $\epsilon$ fraction of nodes remain unmatched. We now present a deeper overview of each of the above components in the subsections below.

\subsubsection{Sequential view of Luby's algorithm}
In the traditional distributed implementation of one round of Luby's algorithm, each edge $f$ picks a uniformly random number $r_f$ and some edge $e$ is chosen into the matching if and only if $r_e < r_{e'}$ for all neighboring edges $e'$. 

We consider the following sequential view - the edges of the graph arrive sequentially in a uniformly random order and an edge $e$ is chosen into the matching if and only if it arrives before any of its neighboring edges.

Assuming that there are no collisions (i.e. each edge chooses a different random number from its neighbors), it can be readily seen that the two algorithms above produce exactly the same distribution over matchings. For CONGEST algorithms, we restrict the range of the random numbers to be integers in $\{1, 2, \ldots, M\}$ for some polynomially large $M$. In this case, \cite{KawarabayashiKS20} showed that the two algorithms are identical up to a vanishingly small failure probability.

\subsubsection{Martingale Techniques}
\label{sec:martingale-overview}

Our analysis of a single round of Luby's algorithm relies on analyzing certain associated martingales and using martingale concentration inequalities.

\subparagraph{Shifted Martingale} Consider a collection $X_1, X_2, \ldots, X_t$ of boolean random variables and let $\E[X_i \mid X_1, \ldots, X_{i-1}] = p_i$. We are interested in obtaining high probability bounds on the sum $S_t = \sum_{i=1}^t X_i$. For example, consider $X_i$ to be the indicator random variable for the event that the $i$th edge is chosen into the matching. Now clearly, the random variables $\{X_i\}$ are not independent, so we cannot use standard Chernoff bounds. If we let $S_i = \sum_{j=1}^i X_j$, then one could hope to use martingale inequalities to get a concentration result for $S_i$. The challenge here is that the sequence $S_1, \ldots, S_t$ is not necessarily a martingale. Nevertheless, when the $p_i$ are bounded, then we can still utilize martingale concentration inequalities to show that $S_t$ does not deviate too much from its mean by considering the following \emph{shifted random variables}: 
\[Y_i = S_i - \mathbb{E}[S_i \mid Y_0, \ldots, Y_{i-1}] + Y_{i-1}\]
Since the sequence $Y_1, \ldots, Y_t$ is a martingale and further has bounded variance, we can use bounded variance martingale concentration bounds on $Y_t$ to give good concentration on $S_t$ as well. This shifting idea is inspired by~\cite{KawarabayashiKS20}, in which it is used in the context of approximate-maximum independent set. In this work, we generalize this shifting idea to broader scenarios in Theorems~\ref{thm:ShiftedMartingaleUpperBound} and~\ref{thm:ShiftedMartingaleLowerBound}, that are deferred to the Appendix (our analysis uses these theorems as black-boxes). 

\subparagraph{Scaled Martingale}
In some parts of our analysis, the shifted martingale trick doesn't suffice for our purposes. Consider a sequence of random variables $S_1, \ldots, S_t$ such that $\mathbb{E}[S_i \mid S_1, \ldots, S_{i-1}] = (1 - p)S_{i-1}$ for some fixed $0 < p < 1$. In this scenario, we expect the difference $S_i - S_{i-1}$ to decrease as $i$ increases. It is challenging for the shifted martingale trick to exploit such dynamics. The main reason is that in order for the shifted martingale to give a concentration result, we need the expected value of $S_i - S_{i-1}$ to be bounded by some fixed number, which wouldn't exploit the property that the difference is decreasing over time. To get a concentration result in such scenarios, we use another trick which we refer to as the \emph{scaling trick}. We can obtain a concentration bound for $S_t$ by considering the following scaled random variables:
\[F_i = \frac{S_i}{(1-p)^{i-1}}\]
Observe that $F_i$ is a martingale. This is because $\mathbb{E}[F_i\mid F_1,\cdots, F_{i-1}] = \frac{1}{(1-p)^{i-1}}\mathbb{E}[S_i\mid F_1,\cdots, F_{i-1}] = \frac{1}{(1-p)^{i-1}}\mathbb{E}[S_i\mid S_1,\cdots, S_{i-1}] = \frac{S_{i-1}}{(1-p)^{i-2}} = F_{i-1}$. Therefore, we can get a concentration result for $S_t$ by getting a concentration result for $F_t$. The main intuition behind the scaling trick is that it exploits the decreasing difference between $S_i$ and $S_{i-1}$ over time. This is exactly the reason for dividing $S_i$ by $(1-p)^i$. For instance, since $F_1=S_1$, if we get that $F_t$ doesn't deviate too far from $F_1$, it would imply that $S_i = (1-p)^i F_i \approx (1-p)^i F_1 = (1-p)^i S_1$, which is exactly where we're expecting $S_i$ to be at step $i$.

\subsubsection{Local \rrl{}}
We analyze a single round of Luby's algorithm using the above martingale based techniques. First, we show a lemma with the following local guarantee. 

\begin{lemma}[Local \rrl{} - Informal]
\label{lem:informal-localrrl}
    Let $G$ be a bipartite $\Delta$-regular graph and suppose we run one round of Luby's algorithm on $G$ and let $u$ be an arbitrary node in $G$. Then, with probability at least $1 - \exp(-\poly(\Delta))$, $\Delta/2 \pm o(\Delta)$ neighbors of $u$ get matched.
\end{lemma}


Recall that $N(u)$ is the set of neighbors of node $u$ and $N^2(u)$ is the set of nodes at distance exactly $2$ from $u$. Let $A_u$ be the set of edges between $N(u)$ and $N^2(u)$. To prove \Cref{lem:informal-localrrl}, we show that roughly $\Delta/2$ edges from $A_u$ are chosen into the matching with probability at least $1 - \exp(-\poly(\Delta))$. While the claim holds trivially in expectation, it is challenging to obtain high probability bounds due to the dependencies between $u$'s neighbors. 

To simplify exposition, let $E_u$ denote the set of edges that are at most 3 hops away from $u$, i.e. $E_u = E \cap \left\{(\{u\} \times N(u)) \cup (N(u) \times N^2(u)) \cup (N^2(u) \times N^3(u))\right\}$. Clearly edges not in $E_u$ do not affect any of the edges in $A_u$ and hence can be ignored, so we restrict the analysis to assume that only edges from $E_u$ arrive in the sequential view of Luby's algorithm.

Let $\M_u \subset A_u$ be the set of matching edges chosen from $A_u$ and our goal is to get high probability upper and lower bounds on $|\M_u|$. Let $X_i \in \{0, 1\}$ be an indicator random variable for the event that the $i$th arriving edge belongs to $\M_u$. Let $q_i = \mathbb{E}[X_i \mid X_1, \ldots, X_{i-1}]$ be the probability that the $i$th edge is from $A_u$ and none of its neighboring edges have already arrived. One could try now to utilize the shifted martingale trick described above to obtain concentration bounds on $|\M_u| = \sum_{i} X_i$. However, the main challenge here is that $q_i$ itself is a random variable. Our goal is to analyze $q_i$ using martingales analysis. Roughly speaking, we use the scaled martingale trick discussed above to get concentration results for $q_i$ for all $i$, which enables us to apply the shifted martingale trick (i.e., Theorems~\ref{thm:ShiftedMartingaleUpperBound} and~\ref{thm:ShiftedMartingaleLowerBound}) to get the desired bounds on $|\mathcal{M}_u|$.

\subparagraph{Analyzing $q_i$:} To analyze $q_i$, we need to understand how many edges \emph{survive} after the first $i-1$ edges have already arrived. Intuitively, an edge is still surviving in iteration $i$ if neither it nor any of its neighbors was not sampled in the first $i-1$ iterations. Let $E_i$ be the set of surviving edges in $A_u$ at the beginning of the $i$th iteration. Then we have, $q_i = \frac{|E_i|}{|E_u| - (i-1)} = \frac{|E_i|}{\Delta(|N^2(u)| + 1) - (i-1)}$ since a sampled surviving edge is always added to the matching.

\subparagraph{The final key property towards the proof:} To get high probability bounds on $|E_i|$, we first prove that $\mathbb{E}[|E_i|\mid E_{i-1}] \approx (1-2/k)|E_{i-1}|$, where $k = |N^2(u)|$. This is exactly the setting where the scaled martingale trick can help us to get a concentration result for $|E_i|$, which paves the way for proving Lemma~\ref{lem:informal-localrrl}.

\subsubsection{\rrl{}}
We use \Cref{lem:informal-localrrl} to show that applying a single round of Luby's algorithm on a regular graph and removing the matched nodes along with incident edges results in a graph that is still almost regular.

\begin{lemma}[\rrl{} - Informal]
\label{lem:informal-rrl}
    Let $G$ be a bipartite $\Delta$-regular graph and let $G'$ be the graph obtained by running one round of Luby's algorithm on $G$ and deleting matched nodes and their incident edges. Then all but $o(1)$ fraction of nodes in $G'$ have their degree in the range $\Delta/2 \pm o(\Delta)$ with high probability.
\end{lemma}

When $\Delta$ is large enough ($\Delta \geq \poly \log n$), \Cref{lem:informal-localrrl} followed by a union bound suffices to show that all nodes have their degree in the range $\Delta/2 \pm o(\Delta)$ with high probability. On the other hand, when $\Delta$ is small, then by \Cref{lem:informal-localrrl}, the expected number of nodes that do \emph{not} have the requisite degree is only $n \cdot \exp(-\poly(\Delta))$. Further, the degree of a node after a round of Luby's algorithm only depends on $O(\poly(\Delta))$ other nodes. So, we can use a Chernoff-Hoeffding with bounded dependence inequality to argue that all but a $\exp(-\poly(\Delta))$-fraction of nodes have the required degree.

\subsubsection{Proof Sketch of \texorpdfstring{\Cref{thm:UpperBoundMain}}{mainthm}}
We first argue that running one round of Luby's algorithm on an almost regular graph matches a constant fraction of the nodes with high probability. We note that while this claim is easy to see in expectation, obtaining a high probability bound requires the use of our shifted martingale technique, particularly when the graph becomes only almost regular (instead of fully regular, as in the first iteration). Combined with \Cref{lem:informal-rrl} that states that the resulting graph remains almost regular, we get that after $O(\log 1/\epsilon)$ rounds, at most $\epsilon$-fraction of nodes remain unmatched.

\subsubsection{Extension to Node-Averaged Complexity}

We now sketch our result for the node-average complexity of MM given the Regular-Graph Preservation Lemma. We can combine \Cref{lem:informal-rrl} with the observation that in an (almost) regular graph, every node has a constant probability of being matched (and hence is only expected to survive a constant number of rounds under Luby). Since node-averaged complexity cares about the expected time to match a node, our initial rounds of Luby are consistent with our $O(1)$ node-averaged complexity goal. However, we cannot continue this trick until all nodes are matched, since eventually the graph will become too irregular for us to proceed. Instead, we only use Luby until a $O(1 / \log \Delta)$ fraction of nodes remain; after so few nodes remain, it is safe for us to finish by repeatedly calling a previous black box~\cite{Bar-YehudaCGS17} that takes $O(\log \Delta/\log\log\Delta)$ rounds to match a constant fraction of nodes but which can handle general (i.e. irregular) graphs. Similar to our previous argument, nodes are only expected to survive a constant number of black box invocations, but since there are so few nodes left our node-averaged complexity is still $O(1)$. Unlike the Luby phase, it is safe to continue this phase until all nodes are matched since it no longer matters how regular the graph is.

\subsection{Overview of Lower Bounds}\label{sec:TechLB}

To prove lower bounds against LOCAL algorithms, we use some ideas from a lower bound construction of Ben-Basat, Kawarabayashi, and Schwartzman~\cite{BenbasatKS19} that gave $\Omega(1 / \epsilon)$ lower bounds in the LOCAL model for $(1 + \epsilon)$ maximum matching as well as other approximate graph problems. The critical idea in that proof is that when an $r$-round LOCAL algorithm is deciding what to do with a node $v$ (e.g. who to match it with), it can only use the local structure of the graph around $v$; in particular, it can only see the $r$-hop neighborhood around $v$. This means that we could cut out this $r$-hop neighborhood from the graph, put it back in differently, and the algorithm would have to make the same decision on $v$ (or for randomized algorithms, the same distribution on decisions). The proof revolves around designing these $r$-hop neighborhoods as (symmetrical) gadgets, then showing that a constant number of gadgets will (with constant probability) induce an unmatched node between them. The BKS proof uses a simple path as their gadget which involves $O(r)$ nodes. Hence the algorithm can be shown to have an overall error rate of one error per $O(r)$ nodes, so the critical round threshold to allow for $(1 + \epsilon)$-multiplicative approximations is $\Omega(1 / \epsilon)$ rounds.

Relative to their result, the main upgrade we want to make is that the counterexample graph(s) should be $\Delta$-regular. It is relatively straightforward to take BKS path gadgets and stick them into a large cycle, recovering their $\Omega(1 / \epsilon)$ round lower bound for $2$-regular bipartite graphs. The main technical hurdle we overcome is generalizing to higher degree. We know from our upper bounds that there must be some degradation as the degree $\Delta$ increases, so the main question is, how much efficiency do we need to lose to burn our excess degree? We design gadgets with $O(\Delta r)$ nodes and asymptotically maintain the original error rate of one error per constant number of gadgets (the cycle proof argues about the outcome of two gadgets inducing a mistake, but for the general case we reason about the outcome of five gadgets), yielding $\Omega(1/(\Delta \epsilon))$ round lower bounds.
\section{An $O(\log 1/\epsilon)$-round Algorithm}
\label{sec:reg}

In this section, we show that there is an $O(\log(1/\epsilon))$-round algorithm to find a $(1+\epsilon)$-approximate maximum matching in $\Delta$-regular graphs when $\Delta \geq (\frac{1}{\epsilon})^c$ for a large constant $c$. We note that the proof makes no attempt to optimize the constant $c$ and we expect that it can be reduced significantly by a more careful analysis.

\thmmain*

\subparagraph{A Roadmap for the Proof of Theorem~\ref{thm:UpperBoundMain}:} The algorithm for proving Theorem~\ref{thm:UpperBoundMain} first runs a simple color coding step to find a bipartite almost regular subgraph, and then runs Luby's algorithm for $O(\log(1/\epsilon))$ rounds. Since our analysis for Luby's algorithm uses several martingale inequalities, it is cleaner to work with the sequential view of Luby that we present in \Cref{sec:sequential-luby}. In Section~\ref{sec:constantApprox} we show that a single round of Luby's algorithm in an almost regular graph matches a constant fraction of the nodes, with high probability.\footnote{In fact, we prove a more general claim where it suffices that a constant fraction of the edges are balanced.} In Section~\ref{sec:PuttingEverythingTogether}, we state our key \rrl{} that shows that running one round of Luby's algorithm on an almost regular graph and deleting the matched nodes yields an almost regular graph. A local version of that lemma, which bounds the probability that the degree of a particular node $u$ almost exactly halves after each round is the most technical part of the proof and we devote \Cref{sec:localrrl} for its proof. Finally, in \Cref{sec:putting-it-together}, we put these components together to finish the proof.

\subsection{Sequential View of Luby's Algorithm}
\label{sec:sequential-luby}

In the traditional distributed implementation of one round of Luby's algorithm, each edge $e$ picks a uniformly random integer $r_e$ in the set $\{1, 2, \ldots, M\}$ for an appropriately chosen polynomially large $M$. 
An edge $e$ is chosen to be in the matching if $r_e < r_{e'}$ for all neighboring edges $e'$. This distributed view is formally presented in Algorithm \ref{alg:L}.

\begin{algorithm}[H]
\SetAlgoLined
\DontPrintSemicolon
\KwData{An unweighted graph $G=(V,E)$, where $|E|=m$ and a constant $c'> 0$ }
\KwResult{A matching $\mathcal{M}$}	
$I\gets \emptyset$\;	
\For{each edge $e\in E$}{		
    $r_e\gets $ uniformly random number in $\{1,2,\hdots,100m^{c'+2}\}$\;		
}
\For{each edge $e\in E$}{
    Add $e$ to $\mathcal{M}$ if $r_e < r_{e'}$ for all neighboring edges $e'$ of $e$ in $G$ (i.e., for edges $e'$ where $e\cap e'\neq \emptyset$)\;		
}
\Return $\mathcal{M}$\;	
\caption{Distributed One Round Luby}
\label{alg:L}
\end{algorithm}

\begin{algorithm}[H]
\SetAlgoLined
\DontPrintSemicolon
\KwData{an unweighted graph $G=(V,E)$}
\KwResult{a matching $\mathcal{M}$}
$\mathcal{M}\gets \emptyset$\;
$U\gets E$\;
\While{$U\ne\emptyset$}{		
    $e\gets $ uniformly random element of the set $U$ 
    $U\gets U\setminus \{e\}$\;		
    \If{$\{e'\in E\mid e'\cap e\neq \emptyset\} \subseteq U$}{
        $\mathcal{M} \gets \mathcal{M} \cup \{e\}$\;
    }		
}
\Return $\mathcal{M}$\;
\caption{SeqLuby - Sequential View of Distributed One Round Luby}
\label{alg:SeqL}
\end{algorithm}

However, it is convenient for our analysis to work with a sequential view of Luby's algorithm. In the sequential view, the edges are sequentially sampled independently without replacement. When an edge $e$ is sampled, it is added to the matching only if none of its neighboring edges had been sampled earlier. Crucially, unlike the greedy matching algorithm, a sampled edge $e$ is \emph{blocked} by a neighboring edge $e'$ that was previously sampled even if $e'$ itself is not in the matching. This sequential view is formally presented in Algorithm \ref{alg:SeqL}. \cite{KawarabayashiKS20}
showed that the two algorithms are equivalent by showing that they produce the same distribution over matchings with high probability\footnote{They stated the claim in the context of independent sets.}. We formalize this in the following proposition.

\begin{proposition}[Proposition 3 in~\cite{KawarabayashiKS20}]
\label{prop:indep-equiv}
	For any unweighted graph $G$ with $m$ edges and constant $c' > 0$, let $\mathcal{D}_0$ and $\mathcal{D}_1$ be the distributions over matchings produced by  Algorithm~\ref{alg:L} and Algorithm~\ref{alg:SeqL}, respectively. The total variation distance between $\mathcal{D}_0$ and $\mathcal{D}_1$ is at most $1/m^{c'}$ from the distribution produced by Algorithm~\ref{alg:SeqL}; in particular
	
	$$\sum_{\text{matchings } \mathcal{M}_0} |\mathbb{P}_{\mathcal{M}\sim \mathcal{D}_0}[\mathcal{M} = \mathcal{M}_0] - \mathbb{P}_{\mathcal{M}\sim \mathcal{D}_1}[\mathcal{M} = \mathcal{M}_0]|\le 1/m^{c'}$$
\end{proposition}

Following \Cref{prop:indep-equiv}, in this paper we focus on analyzing the sequential view of Luby whenever we want to make claims about a single round of Luby's algorithm. This incurs an additional tiny $1/\poly(n)$ failure probability, which we can tolerate.

\subsection{Analyzing One Round of Luby's Algorithm on Almost Regular Graphs}
\label{sec:constantApprox}

In this section we show that a single round of Luby's algorithm on almost regular graphs matches a constant fraction of the nodes with high probability. In fact, we prove this claim for a more general family of graphs in the following theorem. 

\begin{theorem}\label{thm:MatchConstantFraction}
 Let $G$ be an undirected graph with $n$ nodes and $m$ edges, and let $d = 2m/n$ be the average degree. Let $E^{low} = \{(u, v) \in E(G) \mid deg(u) \leq 2d, deg(v) \leq 2d\}$ be the set of edges induced by nodes with degree at most $2d$. If $|E^{low}| \geq m/2$, then one round of Luby's algorithm (Algorithm \ref{alg:L}) finds a matching in $G$ of size at least $n / 288$ with high probability.
\end{theorem}


\begin{proof}
For analysis, we'll focus on the sequential view of Luby's algorithm. By Proposition \ref{prop:indep-equiv}, the result holds for the distributed version as well.

Consider the first $t = n/24$ iterations of the while loop in Algorithm \ref{alg:SeqL}. We first claim that in each of these $t = n/24$ iterations, we add the sampled edge to the matching with probability at least $1/6$ (irrespective of previous random outcomes). Formally, for each $i \in [t]$, let $X_i \in \{0, 1\}$ be a random variable indicating whether we add an edge to the matching in the $i$th iteration, and let $p_i = \mathbb{E}[X_i \mid X_1, \ldots X_{t-1}]$. We now show that $p_i \geq 1/6, \forall i \in [t]$.

We say that an edge $e$ is \emph{blocked} if during any of the first $t$ iterations we sample some neighboring edge $e'$. By definition, if we sample some edge $e$ in an iteration $i \in [t]$ and it isn't blocked, then $e$ is always added to the matching. Next, we show that the number of blocked edges in $E^{low}$ is at most $nd/6$. Assume for contradiction, that there is an iteration where edge $e = \{u, v\}$ is sampled that blocks more than $4d$ edges in $E^{low}$. This implies that either $u$ or $v$ is incident on more than $2d$ edges from $E^{low}$. But that's a contradiction since no edge from $E^{low}$ can be incident on a node with degree $>2d$. Thus, in each iteration, at most $4d$ edges from $E^{low}$ can be blocked. Hence, in the first $t = n/24$ iterations, at most $(n/24)\cdot(4d) = nd/6$ edges are blocked in total.

Therefore, even at the end of the first $t$ iterations, at least $|E^{low}| - nd/6 \geq nd/4 - nd/6 = nd/12$ edges from $E^{low}$ remain unblocked. In any fixed iteration $i \in [t]$, if we sample any of these edges, it will be added to the matching. Hence, the probability that we add an edge to the matching in iteration $i$ is at least:
\[p_i \geq \frac{nd/12}{nd/2 - n/24} \geq \frac{nd/12}{nd/2} = 1/6\]

Finally, to prove the theorem we use the shifted martingale  trick (that was discussed in the preliminaries, and is formally presented in \Cref{sec:ShiftedMartingale}) in a black-box fashion, as follows. Let $S_t = \sum_{i=1}^t X_i$, and let $p_i = \mathbb{E}[X_i\mid X_1,\cdots, X_{i-1}]$. Since $p_i \geq 1/6$ for all $i\in [t]$, then by plugging in $P_i^\ell=1/6$ for all $i$ and $P^{\ell}=\sum_{i=1}^t P_i = t/6$ and $P^h = t$ in Theorem~\ref{thm:ShiftedMartingaleLowerBound}, we get that:

    $$\mathbb{P}[S_t<n/288]\leq \exp{\left(-\frac{(t/6-n/288)^2}{16t}\right)} \leq \exp{\left(-\frac{(n/144-n/288)^2}{2n/3}\right)} \leq e^{-n/96}$$

\end{proof}


\subsection{\rr{}}
\label{sec:PuttingEverythingTogether}

We first define some notation to facilitate the rest of the discussion. Intuitively, we say a node $(\alpha, \Delta)$-regular if all nodes in its two hop neighborhood have the same degree. 
\begin{definition}[$(\alpha, \Delta)$-regular node]
\label{def:alpha-regular}
    Let $\Delta$ be an integer and $\alpha\in [0,1]$ be a real number. Given a graph $G$, we say that a node $u$ is $(\alpha,\Delta)$-regular in $G$ if for any $v\in \{u\}\cup N(u)\cup N^2(u)$, we have $\Delta(1-\alpha) \leq deg(v) \leq \Delta(1+\alpha)$.
\end{definition}

We can now present our main technical lemma that states that if $u$ is a $(\alpha, \Delta)$-regular node, then its degree becomes almost exactly $\Delta / 2$ after running one round of Luby's algorithm, with failure probability exponentially small in $\Delta$. The proof of Lemma \ref{lem:RecursiveRegularity} is the most technical part of the paper and is deferred to \Cref{sec:localrrl}.

\begin{restatable}[Local \rrl{}]{lemma}{lemmaRR}
\label{lem:RecursiveRegularity}
    Let $\Delta\geq 2^{10}$ be an integer, and $\alpha\in [0,1/10]$ be a real number.  Let $G$ be a bipartite graph and $u$ be an $(\alpha,\Delta)$-regular node in it. Let $deg'(u)$ be the number of unmatched neighbors of $u$ after running SeqLuby (Algorithm~\ref{alg:SeqL}) on $G$. With probability at least $1-\exp(-\Delta^{1/16})$, it holds that: 
    
    $$\frac{\Delta}{2}\left(1-(10\alpha+\Delta^{-1/600})\right)\leq deg'(u)\leq \frac{\Delta}{2}\left(1+(10\alpha+\Delta^{-1/600})\right)$$
\end{restatable}

As long as $\Delta$ is large enough ($\approx \log n$), \Cref{lem:RecursiveRegularity} followed by a union bound suffices to show that an almost regular graph remains almost regular with their degrees halved after one round of Luby's algorithm. However, for smaller $\Delta$, we can no longer rely on a simple union bound. In the following lemma, we use the Chernoff-Hoeffding concentration inequality with bounded dependence to show that even for small $\Delta$, \emph{almost all} nodes halve their degree.

\begin{lemma}[\rrl{}]
\label{lem:SparseRecursiveRegularity} 
Let $n \geq K\geq\Delta\geq 2^{10}$ be integers, and let  $\alpha\in [0,1/10]$, $\delta\in [0,1/100]$ be real numbers. Let $G$ be a bipartite graph with $n$ nodes and max degree $K$ and let $G'$ be the graph obtained by running $SeqLuby$ (Algorithm~\ref{alg:SeqL}) on $G$ and removing all matched nodes together with their incident edges.

If at least $(1-\delta)$-fraction of nodes in $G$ are $(\alpha, \Delta)$-regular, then at least $(1 - \delta')$-fraction of nodes in $G'$ are $(\alpha', \Delta/2)$-regular with probability at least $1-\exp{\left(-n/(K^{10}\exp{(\Delta^{1/99})})\right)}$,  where $\delta' = K^2\cdot (\delta + 2\exp{(-\Delta^{1/100})})$ and $\alpha' = 10\alpha+\Delta^{-1/600}$.

\end{lemma}


\begin{proof}
We say that a node $v$ is \emph{good} if its degree in $G'$ is in the range $\frac{\Delta}{2}(1\pm \alpha')$. Any node that's not good is called \emph{bad}. Let $R$ be the set of nodes that are $(\alpha,\Delta)$-regular in $G$. By Lemma~\ref{lem:RecursiveRegularity}, 
each $(\alpha,\Delta)$-regular node in $G$ that remains unmatched is good with probability at least $1-\exp{(-\Delta^{1/16})}$. Hence, the expected number of nodes from $R$ that are bad is at most $|R|\cdot \exp{(-\Delta^{1/16})} \leq n$. To obtain a high probability bound, we use Chernoff-Hoeffding inequality with bounded dependence as follows.

For each node $u \in R$, let $B_u \in \{0, 1\}$ be an indicator random variable that indicates whether $u$ is bad. We have $\sum_{u \in R} \mathbb{E}[B_u] \leq |R|\cdot \exp{(-\Delta^{1/16})} \leq n\cdot \exp{(-\Delta^{1/16})}$. Since the maximum degree in $G$ is $K$, each $B_u$ depends on fewer than $K^{10}$ nodes (with $\lambda = |R|\cdot \exp{(-\Delta^{1/16})})$ to get:
\begin{align*}
    \bbP\left[\sum_{u \in R} B_u > 2 n \exp(-\Delta^{1/100})\right] &\leq \bbP\left[\sum_{u \in R} B_u > 2 |R| \exp(-\Delta^{1/16})\right]\\
    &\leq \bbP\left[\sum_{u \in R} B_u > \sum_{u \in R} \E[B_u] + |R| \exp(-\Delta^{1/16})\right]\\
    &\leq \exp{\left(-\frac{2 (|R| \cdot \exp(-\Delta^{1/16}))^2}{|R| \cdot K^{10}}\right)}\\
    &= \exp{\left(-\frac{2|R|}{K^{10}\exp{(2 \Delta^{1/16})}}\right)}\\
    &\leq \exp{\left(-\frac{2|R|}{K^{10}\exp{(\Delta^{1/99})}}\right)}
    \leq \exp{\left(-\frac{n}{K^{10}\exp{(\Delta^{1/99})}}\right)}
\end{align*}
where the last line used $\Delta > 2^{10}$ and $|R| \geq (1-\delta)n \geq n/2$. We say that a node $v$ poisons a node $u$ if $v$ prevents $u$ from being $(\alpha', \Delta/2)$-regular, i.e., $v$ poisons $u$ if $v$ is \emph{bad} and is at distance at most $2$ from $u$.
 We note that each bad node poisons at most $K^2$ nodes. Let $B$ be the set of bad nodes in $R$, i.e., we have $|B| = \sum_{u \in R} B_u$ and let $\tilde{B} = V(G) \setminus R$ be the set of nodes that are not $(\alpha, \Delta)$ regular in $G$. For nodes in $\tilde B$, we have no guarantee on the probability of whether they become good, so we always assume that they poison $K^2$ nodes. In total, the number of nodes that are not $(\alpha',\Delta/2)$-regular in $G'$ is at most $(|\tilde B| + |B|)K^2$ which is at most:
    \begin{align*}
    (|\tilde B| + 2n\cdot \exp{(-\Delta^{1/100})})\cdot K^2 \leq (\delta n + 2n\cdot \exp{(-\Delta^{1/100})})\cdot K^2 = n\cdot K^2\cdot (\delta + 2\exp{(-\Delta^{1/100})})
    \end{align*}
with probability at least $1-\exp{\left(-n/(K^{10}\exp{(\Delta^{1/99})})\right)}$, as desired.
\end{proof}

\Cref{lem:SparseRecursiveRegularity} shows that after each round of Luby's algorithm, the remaining graph still satisfies the requirement that most vertices remain almost regular, albeit with worsening parameters. The following technical claim shows that these parameters remain small enough after $O(\log(1/\epsilon))$ rounds. The proof uses basic arithmetic and is deferred to the appendix.

\begin{claim}\label{claim:BoundingParameters}
    Let $\Delta$ be an integer, $c=1/10^5$, and $\epsilon$ be a real number satisfying $1> \epsilon\geq 1/\Delta^{c}$. Furthermore let:
    
    \begin{enumerate}
        \item $\alpha_0 = \Delta^{-1/600}$ and $\alpha_i = 10\alpha_{i-1} + \Delta^{-1/600}$ for $i\geq 1$, 
        \item $\delta_0=\exp(-\Delta^{1/200})$, and $\delta_i = \Delta^2(\delta_{i-1}+2\exp(-(\Delta/2^i)^{1/100}))$ for $i\geq 1$.
    \end{enumerate} 
    
    It holds that $\alpha_i\leq 1/10$ and 
        $\delta_i\leq \exp(-\Delta^{1/300})$ for $1\leq i\leq 10\log(1/\epsilon)$.
    \end{claim}

We now extend Lemma~\ref{lem:SparseRecursiveRegularity} to a multi-round argument, showing that most of nodes' degrees after running Luby's algorithm for $i$ rounds become $\approx \Delta/2^i$.

\begin{lemma}[Multi-Round \rr{}]
\label{lem:FinalSparseRegularity}
     Let $(\Delta, \epsilon, (\alpha_i), (\delta_i))$ be as defined in \Cref{claim:BoundingParameters}. Let $G$ be a bipartite $n$-node graph, where all but $\delta=\delta_0 =\exp{(-\Delta^{1/200})}$-fraction of the nodes are $(\alpha_0,\Delta)$-regular. For $i\leq 10\log(1/\epsilon)$, let $G_i$ be the graph obtained after running Luby's algorithm (Algorithm~\ref{alg:L}) for $i$ rounds on $G$, where in each round we remove the matched nodes together with their incident edges. If $|V(G_i)|\geq \epsilon n$, then with high probability, at least $(1 - \delta_i)$-fraction of the nodes in $G_i$ are $(\alpha_i,\Delta/2^i)$-regular.
\end{lemma}

\begin{proof}
    The idea is to use a recursive argument where we apply Lemma~\ref{lem:RecursiveRegularity} with a simple union bound in the dense case, and Lemma~\ref{lem:SparseRecursiveRegularity} in the sparse case. For convenience of notation, let $\Delta_i=\Delta/2^i$. Since we apply Lemmas~\ref{lem:RecursiveRegularity} and ~\ref{lem:SparseRecursiveRegularity} recursively for $i$ rounds, we need to ensure that $\alpha_i\leq 1/10$ and $\delta_i\leq 1/100$, which follows from Claim~\ref{claim:BoundingParameters}. The proof is split into two cases, a dense case where $\Delta\geq \log^{50}n$, and a sparse case. We start with the dense case. 

    \subparagraph{Dense case where $\Delta\geq \log^{50} n$:} 
    By Lemma~\ref{lem:RecursiveRegularity} and a simple union bound argument, in the case where $\Delta\geq \log^{50} n$, \emph{all} the nodes are $(\alpha_i,\Delta/2^i)$-regular after running Luby's algorithm (Algorithm \ref{alg:L}) for $i$ rounds with high probability. This is because the failure probability of Lemma~\ref{lem:RecursiveRegularity} is at most $\exp{(-\Delta_i^{1/16})}$ at round $i$, and $\Delta_i\geq \Delta^{0.99}$ for any $i$ (since $\epsilon\geq 1/\Delta^{1/c}$ and $i\leq 10\log(1/\epsilon)$).

    \subparagraph{Sparse case where $\Delta\leq \log^{50}n$:} The idea is to apply Lemma~\ref{lem:SparseRecursiveRegularity} recursively for $i$ rounds and the proof follows by induction. 
    
    \subparagraph{Induction base i=1:} By Proposition~\ref{prop:indep-equiv} Algorithms~\ref{alg:L} and~\ref{alg:SeqL} produce the same distributions over matchings in the graph $G$, up to $1/\poly(n)$ total variation distance. Hence, for $i=1$ the claim follows directly from Lemma~\ref{lem:SparseRecursiveRegularity}, since the failure probability of Lemma~\ref{lem:SparseRecursiveRegularity} is only $\exp{\left(-n/\Delta^{10}\exp{(\Delta^{1/99})}\right)}<1/n^{1000}$, for $\Delta\leq \log^{50} n$.

    \subparagraph{Induction step:} Assume that the claim is true for $i$, i.e., assume that all but $|V(G_i)|\cdot \delta_{i}$ nodes in $G_i$ are $(\alpha_i,\Delta/2^i)$-regular with high probability. We show that the claim is true for $i+1$. Again, by Proposition~\ref{prop:indep-equiv}, Algorithms~\ref{alg:L} and~\ref{alg:SeqL} produce the same distributions over matchings in the graph $G_i$, up to a $1/|E(G_i)|^c\leq 1/(|V(G_i)|)^c\leq 1/(\epsilon n)^c\leq (1/n^{0.99})^c$ total variation distance, for an arbitrarily large constant $c$. Since the max degree in $G_i$ is at most $\Delta$, by Lemma~\ref{lem:SparseRecursiveRegularity} it holds that all but $|V(G_{i+1})|\cdot \delta_{i+1}$ nodes in $G_{i+1}$ are $(\alpha_{i+1},\Delta/2^{i+1})$-regular in $G_{i+1}$ with probability at least 
    
    $$1-\exp{\left(-\frac{|V(G_i)|}{\Delta^{10}\exp{(\Delta_i^{1/99})}}\right)}\geq 1-\exp{\left(-\frac{\epsilon n}{\Delta^{10}\exp{(\Delta^{1/99})}}\right)}\geq 1-\frac{1}{n^{1000}}$$

    \noindent where the first inequality holds since $\Delta_i\leq \Delta$ and $|V(G_i)|\geq \epsilon n$, and the second inequality holds since $\epsilon\geq 1/\Delta^{1/100}\geq 1/n^{1/100}$ and $\Delta\leq \log^{50}n$. Hence, by a union bound on all $i$'s, we get that as long as $V(G_i)$ has at least $\epsilon n$ nodes, it holds that all but $|V(G_i)|\cdot \delta_i$ nodes are $(\alpha_i,\Delta_i)$-regular in $G_i$ with probability at least $1-1/n^{100}$.
\end{proof}

\subsection{Putting it together}
\label{sec:putting-it-together}

\Cref{lem:FinalSparseRegularity} shows that the graph remains almost regular after $i \approx \log(1/\epsilon)$ repeated applications of Luby's algorithm, and \Cref{thm:MatchConstantFraction} showed that each round of Luby's  on an almost regular graph matches a constant fraction of nodes. We can now combine these two results to show that as long as the remaining graph is large enough, each application of Luby's algorithm matches a constant fraction of nodes. 

We first present a corollary of \Cref{thm:MatchConstantFraction} to formalize that the almost regular graphs implied by \Cref{lem:FinalSparseRegularity} do indeed satisfy the requirements of \Cref{thm:MatchConstantFraction}. We defer the proof to the appendix for brevity.

\begin{corollary}
\label{cor:ConstantMatchingRegularGraphs}
    Let $\Delta \geq C$, where $C$ is a large enough constant, and let
    $G'$ be a graph with $n'$ nodes and max degree $\Delta$. Suppose at least $(1 - \exp(-\Delta^{1/300}))$-fraction of the nodes are $(\alpha',\Delta')$-regular for $\alpha'\leq 1/10$. It holds that Luby's algorithm matches $n'/288$ nodes in $G'$ with high probability. 
\end{corollary}


We are now ready to prove Theorem~\ref{thm:UpperBoundMain}.

\begin{proof}[\textbf{Proof of Theorem~\ref{thm:UpperBoundMain}}]

We first provide a proof that assumes that the graph in bipartite, and then we lift this assumption by a simple sampling argument. We simply run Luby's algorithm in $G$ for $i=10\log(1/\epsilon)$ rounds. By Lemma~\ref{lem:FinalSparseRegularity}, all but an $\exp{(-\Delta^{1/300})}$-fraction of the nodes in the graph are $(\alpha_i,\Delta_i)$-regular at round $i$ with high probability. Hence, by ~\Cref{cor:ConstantMatchingRegularGraphs}, in each of these rounds we match at least a $(1/288)$-fraction of the nodes. Therefore, after $O(\log(1/\epsilon))$ rounds, all but an $\epsilon$-fraction of the nodes are matched, as desired.

\subparagraph{Overcoming bipartiteness:} Lemma~\ref{lem:FinalSparseRegularity} assumes that the input graph is bipartite. To overcome this, we apply a simple color-coding trick. Before running Luby's algorithm, each node picks a uniformly random color independently in $\{0,1\}$. This naturally defines a bipartite subgraph graph $G'$ by ignoring all monocolored edges. Observe that for a given node $u\in G$, the degree of $u$ in $G'$ is $\frac{\Delta}{2}(1\pm \Delta^{-0.4})$ with probability at least $1-\exp{(-\Delta^{0.2}/6)}$ by a standard Chernoff-Hoeffding bound (Theorem~\ref{thm:Chernoff}). Hence, if $\Delta\geq \log^{50} n$, then all the nodes degrees in $G'$ are $\frac{\Delta}{2}(1\pm \Delta^{-0.4})$ with high probability. Therefore, all the nodes are also $(\Delta^{-0.4},\frac{\Delta}{2})$-regular in that case. Otherwise, if $\Delta<\log^{50} n$, then we can use a bounded dependence Chernoff-Hoeffding (Theorem~\ref{thm:CorrelatedChernoff}) type argument, as follows. By a union bound, the probability that a node $u$ isn't $(\Delta^{-0.4},\frac{\Delta}{2})$-regular is at most $\Delta^2\cdot \exp{(-\Delta^{0.2}/6)}\leq \exp{(-\Delta^{0.1})}$. Hence, the expected number of nodes that aren't $(\Delta^{-0.4},\frac{\Delta}{2})$-regular is at most $n\cdot \exp{(-\Delta^{0.1})}$. Furthermore, the event of whether a node is $(\Delta^{-0.4},\frac{\Delta}{2})$-regular depends only on $<\Delta^{10}$ other nodes. Hence, by applying Theorem~\ref{thm:CorrelatedChernoff} with $\lambda = n\cdot \Delta^{-0.1}$, we get that with high probability, all but $n\cdot \exp{(-\Delta^{1/200})}$ nodes are $(\Delta^{-0.4},\frac{\Delta}{2})$-regular in $G'$. Hence, since $\Delta^{-0.4}\leq \Delta^{-1/600}$, it suffices to apply Lemma~\ref{lem:FinalSparseRegularity} on $G'$. This concludes the proof. 
\end{proof}

\subsection{Extension: Node-Averaged Complexity}

We next utilize the techniques developed in this section to derive the following main result concerning node-averaged complexity.

\thmnodeaveraged*

\begin{proof}
  The algorithmic plan is similar to that of \Cref{thm:UpperBoundMain}. We begin by handling regular graphs whose degree $\Delta$ is a sufficiently large constant. We first apply our color-coding trick and then run Luby for $i = 200 \log \log \Delta$ rounds. We then repeatedly invoke the ``modified nearly-maximal independent set'' algorithm of \cite{Bar-YehudaCGS17} (in particular, we will use their Theorem 3.1) on the line graph (to translate from matching to independent set) of our remaining graph until our matching is maximal.

  Our color-coding trick guarantees that all but a $\exp{(-1^{1/200})}$-fraction of nodes are $(\Delta^{-0.4}, \Delta/2)$-regular in the bipartite subgraph $G'$. We then choose $\epsilon = 1/\log^{20} \Delta$; note that this is at least $1 / \Delta^{1/10^5}$ when $\Delta$ is sufficiently large (which we assumed for this part of the proof). This lets us invoke \Cref{lem:FinalSparseRegularity} to show that, with high probability, it is safe to run Luby's algorithm for $i \le 10 \log (1 / \epsilon) = 10 \log (\log^{20} \Delta) = 200 \log \log \Delta$ rounds because approximate degree factor is bounded ($\alpha_i \le 1/10$) and the number of irregular nodes is bounded ($\delta_i \le \exp{(-\Delta^{1/300})}$). We again invoke \Cref{cor:ConstantMatchingRegularGraphs} to say that with high probability, during these $200 \log \log \Delta$ rounds of Luby, we match at least a $1/288$-fraction of nodes in each Luby round. Hence with high probability our Luby rounds leave at most a $(287/288)^{200 \log \log \Delta} \le (1/2)^{\log \log \Delta} = 1 / \log \Delta$ fraction of nodes unmatched. Furthermore, since we were able to match a constant fraction of nodes per round, the contribution of this algorithmic phase to the overall node-averaged complexity is $O(1)$.

  Next, we use an algorithm of \cite{Bar-YehudaCGS17} to handle the remaining (irregular) graph. By running it on the line graph of our remaining graph, their theorem has the following guarantee for matchings: there is a constant $\beta$ such that over $\beta(\log \Delta / \log K + K^2 \log 1 / \delta)$ rounds, their algorithm will choose matching edges so that the probability that an edge survives is at most $\delta$, where $\Delta$ is a bound on the degree of any node, $K$ is an algorithm parameter (which they choose to be $\log^{0.1} \Delta$, and $\delta$ is another algorithm parameter (which they choose based on their error parameter $\epsilon$). Because we have so few surviving nodes going into this algorithm, we can choose parameters less aggressively and instead set $K = 2$ and $\delta = 1 / (\Delta \log^2 \Delta)$ to get a $O(\log \Delta)$ round algorithm which only lets edges survive with probability at most $1 / (\Delta \log^2 \Delta)$; we also know that the max degree is bounded by the degree of our original regular graph (which we have been using $\Delta$ to denote anyways). By a union bound, we know that nodes survive this algorithm with probability at most $1 / (\log^2 \Delta)$ (a node cannot survive if none of its edges do. Since this survival probability is smaller than the number of rounds this algorithm takes, each node only survives for $O(1)$ invocations of this algorithm, in expectation. We plan to simply invoke this algorithm forever until we get a maximal matching, which is safe because our node-averaged goal only considers the expected (not worst-case) survival time.

  In total, we have that with high probability, Luby handles at least a $(1 - 1/\log \Delta)$-fraction with $O(1)$ node-averaged rounds, and then the remaining at most $(1 / \log \Delta)$-fraction are handled by \cite{Bar-YehudaCGS17} in $O(\log \Delta)$ node-averaged rounds, for a total node-averaged round count of $O(1)$. If our high probability claim regarding Luby rounds fails, then we instead observe that the failure probability is at most $O(1 / \log \Delta)$ and instead handle all nodes using \cite{Bar-YehudaCGS17}, which still yields a node-averaged round count of $O(1)$.
  
  Finally, we need to consider what happens if $\Delta$ is not sufficiently large for this argument. In that case, $\Delta$ is bounded by a constant, and we can directly repeatedly invoke \cite{Bar-YehudaCGS17} with the same parameters above, which handles a constant fraction of nodes in $O(\log \Delta)$ rounds. Each node still survives a constant number of invocations in expectation and since $\Delta$ is now a constant, the node-averaged round count is $O(1)$ as well. This concludes the proof.
\end{proof}
\section{Local \rr{}}
\label{sec:localrrl}

We devote this section to the proof of \Cref{lem:RecursiveRegularity} that we restate here for convenience.

\lemmaRR*

\subsection{Notation and Preliminaries}
For the entirety of this section, let $u$ be some fixed $(\alpha, \Delta)$-regular node in $G$. Recall that $N(u)$ is the set of neighbors of $u$ and $N^d(u)$ is the set of nodes at distance exactly $d$ from $u$. Let $k = |N^2(u)|$.  Let $E_u$ denote the set of edges that are at most 3 hops away from $u$, i.e, $E_u = E \cap \left\{(\{u\} \times N(u)) \cup (N(u) \times N^2(u)) \cup (N^2(u) \times N^3(u))\right\}$ and let $A_u = E \cap \{N(u) \times N^2(u)\}$. The lemma states that almost exactly $\Delta/2$ edges from $A_u$ are matched with high probability.

\begin{figure}[tb]
\begin{center}
\includegraphics[scale=0.5]{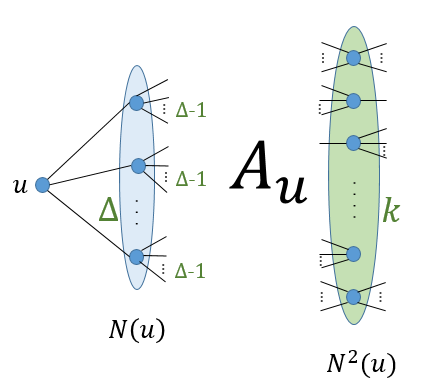}
\end{center}
\caption{We focus on a node $u$ that has $\Delta$ neighbors. The set of neighbors of $u$ is denoted by $N(u)$, and the set of nodes at distance 2 from $u$ is denoted by $N^2(u)$. We denote the size of $N^2(u)$ by $k$. The set of edges between $N(u)$ and $N^2(u)$ is denoted by $A$. Observe that each $v\in N(u)$ has exactly $\Delta-1$ edges incident to it in $A_u$. However, this is not necessarily the case for the nodes in $N^2(u)$, that can have anywhere between $1$ to $\Delta$ incident edges in $A_u$.}
\label{fig:setup}
\end{figure}

It is easy to observe that only edges from $E_u$ can affect the insertion of edges in $A_u$ in the matching, and hence in this section, we restrict our attention to only edges in the set $E_u$. In other words, instead of SeqLuby (Algorithm~\ref{alg:SeqL}), it suffices to show the lemma for Algorithm \ref{Alg: SeqLuby_u}. We formalize this in the following observation.

\begin{observation}\label{ob:SameMatchingSeqLuby_u}
    Let $D_0$ and $D_1$ be the distributions over the matching edges picked from $A_u = N(u)\times N^2(u)$ by Algorithms~\ref{alg:SeqL} and~\ref{Alg: SeqLuby_u}, respectively. It holds that for any matching $\mathcal{M}\subseteq A_u$:

    $$\mathbb{P}_{\mathcal{M}' \sim D_0}[\mathcal{M}' = \mathcal{M}] = \mathbb{P}_{\mathcal{M}' \sim D_1}[\mathcal{M}' = \mathcal{M}]$$
\end{observation}

Throughout this section, by iteration $i$ we mean the $i$th iteration of the while loop in Algorithm \ref{Alg: SeqLuby_u}.
We say an edge $e \in A_u$ \emph{survives} in iteration $i$ if neither it nor any of its neighbors have been sampled in the first $i-1$ iterations. Intuitively, a surviving edge that gets sampled gets added to the resulting matching.

\begin{algorithm}[H]
	\SetAlgoLined
	\DontPrintSemicolon
	\KwData{an unweighted graph $G=(V,E)$ and a node $u\in V$}
	\KwResult{a matching $\mathcal{M}_u\subseteq N(u)\times N^2(u)$}	
	$\M'_u\gets \emptyset$\;	
	$E_u\gets \{(u\times N(u))\cup (N(u)\times N^2(u))\cup (N^2(u)\times N^3(u))\}\cap E$\;
    $U\gets E_u$\;
	\While{$U\ne\emptyset$ }{		
		$e\gets $ uniformly random element of the set $U$ 	\;	
		$U\gets U\setminus \{e\}$\;		
		\If{$\{e'\mid e'\cap e\neq \emptyset\} \subseteq U$}{
			$\mathcal{M}'_u \gets \mathcal{M}'_u \cup \{e\}$\;
		}		
	}
        $\M_u = \mathcal{M}_u'\cap (N(u)\times N^2(u))$\;     
	\Return $\M_u$\;	
	\caption{SeqLuby$_u$ - SeqLuby in the neighborhood of $u$}\label{Alg: SeqLuby_u}
\end{algorithm}

\begin{definition}[The set of surviving edges $E_i$]
\label{def:surviving edges}
    For $i\geq 1$, we say that an edge $e\in A_u$ is still surviving at the beginning of the $i$th iteration of the while loop in Algorithm~\ref{Alg: SeqLuby_u}, if it holds that
    \[\{e\}\cup \{e'\mid e'\cap e\neq \emptyset\}\subseteq U\]
    where $U$ is the set of edges that haven't been sampled so far. 
    Let $E_i$ be the set of surviving edges at the beginning of the $i$th iteration.
\end{definition}

\subparagraph{The Stopping Time:} Since our proofs involve martingales analysis, one important aspect of it is \emph{the stopping time $t$} of our martingales. In several applications of martingale inequalities, we need to stop the martingale earlier than its final termination point. This concept is standard in martingales analysis, and in most cases, without stopping the martingale early, one wouldn't be able to apply the black-box martingale inequalities conveniently. In our context, $t$ is the number of iterations of the while loop in Algorithm~\ref{Alg: SeqLuby_u} that we consider for our analysis. We set $t=k\log\Delta/100$ (recall that $k=|N^2(u)|$). For the lower bound on $|\mathcal{M}_u|$, we show that the size of $|\mathcal{M}_u|$ is already large enough after $t$ iterations of the while loop. For the upper bound, we show that $|\mathcal{M}_u|$ doesn't change substantially after the first $t$ iterations (this is formally proven in Lemma~\ref{lem:FewUnlabeledNodes}), which allows us to ``stop" the analysis after $t$ iterations. 

For easy reference, we include a table of notations that we use throughout the proof in \Cref{table:notations}.

\begin{table}[ht]
\centering
\caption{Table of Notations}\label{table:Notations}
\begin{tabular}{|c|l|}
\hline
Notation & Meaning \\
\hline
\(V\) & The set of nodes in the input graph.\\
\hline
\(E\) & The set of edges in the input graph.\\
\hline
\(u\) & The fixed node throughout ~\Cref{sec:localrrl}.\\
\hline
\( N(u) \) & The set of neighbors of \( u \). \\
\hline
\( N^d(u) \) & The set of nodes at distance d from \( u \). \\
\hline
\(k\) & The size of \(N^2(u)\), i.e., \(k = |N^2(u)|\).\\
\hline
\(\Delta\) & The degree of the original input regular graph. In the Local \rrl{}\\& all the nodes in $\{u\}\cup N(u)\cup N^2(u)$ have degree in the range $\Delta(1\pm \alpha)$\\
\hline
\(\alpha\) & A parameter in $[0,1/200]$ used to lower and upper bound the degrees of the nodes in\\& $\{u\}\cup N(u)\cup N^2(u)$ by $\Delta(1\pm \alpha)$.\\
\hline
\(E_u\) & The set of edges in \((u\times N(u))\cup (N(u)\times N^2(u))\cup (N^2(u)\times N^3(u))\).\\
\hline
\( A_u \) & The set of edges in \( N(u) \times N^2(u) \).  \\
\hline 
\(\mathcal{M}_u\) & The returned set of matching edges in $A_u$, i.e., $|\mathcal{M}_u|$ is the number\\ & of matched neighbors of $u$.\\
\hline
$U$ & The set of edges that weren't sampled so far in Algorithm~\ref{Alg: SeqLuby_u}.\\
\hline
\(Z_i\) & A boolean random variable indicating whether we add an edge\\ &from $A_u$ to the matching in the $i$th iteration.\\
\hline
\( q_i \) & The probability that $Z_i = 1$ given $Z_1,\cdots, Z_{i-1}$.\\ &Since the $Z_i$'s are boolean, $q_i = \mathbb{E}[Z_i\mid Z_1,\cdots, Z_{i-1}]$. \\
\hline
\(E_i\) & The set of surviving edges in $A_u$ at the beginning of the $i$th iteration.\\
\hline
$t$ & The stopping time $t = k\log\Delta/100$.\\
\hline
\end{tabular}
\label{table:notations}
\end{table}

\subsection{Analyzing Number of Surviving Edges \texorpdfstring{$|E_i|$}{Ei}}
\label{sec:analyzingei}

Our first goal is to obtain tight bounds on the number of surviving edges at the beginning of the $i$th iteration. Towards this goal, we first bound the total number of edges adjacent to any node that are sampled by the first $t$ iterations. We setup some additional notation to facilitate this discussion.

\begin{definition}[Labeled Edges]
\label{def:labeledEdges}
    We say that an edge $e$ is labeled at iteration $i\in [t]$ if $e$ has been sampled in one of the first $i$ iterations, i.e., if $e \notin U$ at the end of $i$th iteration of the while loop in  Algorithm~\ref{Alg: SeqLuby_u}. If $e$ is labeled at iteration $i$ then it also labeled at any iteration $j>i$.
\end{definition}

\begin{claim}
    \label{cor:EveryNodeFewLabeled}
    With probability at least $1-\exp{\left(-\Delta^{1/11}\right)}$, every node in $\{u\}\cup N(u)\cup N^2(u)$ has at most $\Delta^{1/9}$ labeled edges incident to it at any iteration $i\in [t=k\log\Delta/100]$.
\end{claim}

\begin{proof}
    Fix a node $v\in \{u\}\cup N(u)\cup N^2(u)$. We prove the claim by applying Theorem~\ref{thm:ShiftedMartingaleUpperBound}. For this, let $X_i\in \{0,1\}$ be a random variable indicating whether we sample an edge incident to $v$ at iteration $i$, and let $p_i(v) = \mathbb{E}[X_i\mid X_1,\cdots,X_{i-1}] = \mathbb{P}[X_i = 1\mid X_1,\cdots,X_{i-1}]$.  Furthermore, let $S_j = \sum_{i=1}^j X_i$ be the number of labeled edges incident to $v$ at iteration $j$. Since $S_j\geq S_{i}$ for any $j>i$, it suffices to prove the claim for $j=t$. Observe that since the degree of any node in $\{u\}\cup N(u)\cup N^2(u)$ is at least $(1-\alpha)\Delta$, it holds that $|E_u|\geq (1-\alpha)\Delta(k+1)$ (since each edge in $E_u$ is either incident with $u$ or some node in $N^2(u)$). Hence, we have that:

    $$p_i(v)\leq \frac{(1+\alpha)\Delta}{(1-\alpha)\Delta(k+1) - (i-1)}$$

    \noindent This is because $v$ has at most $(1+\alpha)\Delta$ edges incident to it, and at the $i$'th iteration we're sampling an edge out of $|E_u| - (i-1) \geq (1-\alpha)\Delta(k+1)-(i-1)$ remaining edges. Therefore, we have that:

    $$p_i(v)\leq \frac{(1+\alpha)\Delta}{(1-\alpha)\Delta(k+1)-(i-1)}\leq \frac{(1+\alpha)\Delta}{(1-\alpha)\Delta(k+1) - t +1} = \frac{(1+\alpha)\Delta}{(1-\alpha)\Delta(k+1) - k\log\Delta/100 +1}\leq \frac{8}{k+1}$$
    
    \noindent Let $P_i = 8/(k+1)$, and let $P=\sum_{i=1}^t P_i = 8t/(k+1)\leq \log\Delta$. By plugging these values into Theorem~\ref{thm:ShiftedMartingaleUpperBound} with $\lambda = \Delta^{1/9}$ and $M=1$, we get that:

    $$\mathbb{P}[S_t\geq \Delta^{1/9}]\leq \exp{\left(-\frac{(\Delta^{1/9}-\log\Delta)^2}{8\log\Delta + 2(\Delta^{1/9} - \log\Delta)/3}\right)}\leq \exp{\left(-\Delta^{1/10}\right)}$$
    as desired.

    The claim now follows by a union bound over all vertices in $\{u\} \cup N(u) \cup N^2(u)$.
\end{proof}

We say that node a $v$ is labeled at iteration $i$ if the edge sampled during the $i$th iteration is adjacent to  $v$. We now utilize \Cref{cor:EveryNodeFewLabeled} to argue that any iteration, every node $v$ is equally likely to be labeled with probability $\approx 1/k$.

\begin{lemma}\label{lem:ProbNodeIsLabeled}
    Recall that $p_i(v)$ is the probability that we label a node $v$ in the $i$th iteration. With probability at least $1-\exp{\left(-\Delta^{1/11}\right)}$, for any node $v\in \{u\}\cup N(u)\cup N^2(u)$ and for any iteration $i\in [t=k\log\Delta/100]$, it holds that the probability that $v$ is labeled at iteration $i$ is:

    $$\frac{(1-\alpha)}{(1+\alpha)\cdot k}\cdot (1-\Delta^{-7/8})\leq p_i(v)\leq \frac{(1+\alpha)}{(1-\alpha)\cdot k}\cdot(1+\frac{\log\Delta}{25\Delta})$$
\end{lemma}

\begin{proof}
        We start with the upper bound. The probability that we label an edge incident to $v$ at any iteration $i$ is at most:
        
        \begin{align*}
            p_i(v)&\leq \frac{(1+\alpha)\Delta}{(1-\alpha)\Delta(k+1) - (i-1)}\leq \frac{(1+\alpha)\Delta}{(1-\alpha)\Delta(k+1) - t +1}\\
            &\leq \frac{(1+\alpha)\Delta}{(1-\alpha)\Delta k - k\log\Delta/100}  = \frac{(1+\alpha)\Delta}{(1-\alpha)\Delta k(1-\log\Delta/(100(1-\alpha)\Delta))}\leq \frac{(1+\alpha)}{(1-\alpha)\cdot k}\cdot  (1+\frac{\log\Delta}{25\Delta})
        \end{align*}
        
        \noindent where the last inequality follows since:
        
        \begin{align*}
            \frac{1}{1-\log\Delta/(100(1-\alpha)\Delta)} &= \frac{1-\log\Delta/(100(1-\alpha)\Delta) + \log\Delta/(100(1-\alpha)\Delta)}{1-\log\Delta/(100(1-\alpha)\Delta)}\\
            &\leq 1 + \frac{\log\Delta/(100(1-\alpha)\Delta)}{1/2}\leq 1+\frac{\log\Delta}{25\Delta}
        \end{align*}
        
        \noindent For the lower bound, we apply Claim~\ref{cor:EveryNodeFewLabeled}, which says that with probability at least $1-\exp{\left(-\Delta/11\right)}$, every node $v\in \{u\}\cup N(u)\cup N^2(u)$ has at most $\Delta^{1/9}$ labeled edges incident to it at any of the first $t$ iterations. Hence, any node $v$ has at least $(1-\alpha)\Delta - \Delta^{1/9}$ unlabeled edges incident to it at any iteration $i\in [t]$. The probability that we label $v$ at iteration $i$ is exactly the probability that we pick one of these unlabeled edges out of the $|E_u|-(i-1)$ remaining edges. Hence:

        \setcounter{equation}{0}
        \begin{align}
            p_i(v)&\geq \frac{(1-\alpha)\Delta - \Delta^{1/9}}{|E_u| - (i-1)}\\
            &\geq \frac{(1-\alpha)\Delta - \Delta^{1/9}}{(1+\alpha)\Delta(k+1) - (i-1)}\\
            &\geq \frac{(1-\alpha)\Delta\cdot (1-1/((1-\alpha)\Delta^{8/9}))}{(1+\alpha)\Delta(k+1)}\\
            &\geq \frac{(1-\alpha)\Delta\cdot (1-2/\Delta^{8/9})}{(1+\alpha)\Delta(k+1)}
            \\&\geq \frac{(1-\alpha)}{(1+\alpha)\cdot k}\cdot(1-\Delta^{-7/8})
        \end{align}

        \noindent where (1) holds since the degree of $v$ is at least $(1-\alpha)\Delta-\Delta^{1/9}$, (2) holds since $|E_u|\leq (1+\alpha)\Delta(k+1)$, (4) holds since $\alpha\leq 1/10$, which implies that $(1-1/((1-\alpha)\Delta^{8/9}))\geq (1-2\Delta^{-8/9})$, and (5) holds since $1/(k+1)\geq (1/k)(1-1/k)$, and since $k\geq \Delta/7$. The latter holds since $u$ has at least $(1-\alpha)\Delta$ neighbors, and each of them has at least $(1-\alpha)\Delta-1$ edges incident to it in $|A_u|$. Hence:
        
        $$k=|N^2(u)|\geq \frac{|A_u|}{(1+\alpha)\Delta}\geq \frac{(1-\alpha)^2\Delta^2-(1-\alpha)\Delta}{(1+\alpha)\Delta}\geq \Delta/7$$
        
        \noindent where the last inequality holds since $\alpha\leq 1/10$.
    \end{proof}

The above bounds on $p_i(v)$ allow us to show that the expected number of surviving edges drops by a constant factor in each iteration. Recall that an edge $e \in A_u$ is said to be surviving at iteration $i$ if neither it or any of its neighbors have been sampled before that iteration. Informally, we show that $\mathbb{E}[|E_i| \mid E_{i-1}] \approx (1-2/k) |E_{i-1}|$.

\begin{lemma}\label{lem:newlyKilledEdges}
With probability at least $1-\exp{\left(-\Delta^{1/11}\right)}$, for any $1<i\leq t=k\log\Delta/100$, it holds that:

\begin{enumerate}
    \item $\mathbb{E}[|E_i|\mid E_1, \cdots, E_{i-1}]\geq \Big(1-\frac{2}{ k}\frac{(1+\alpha)}{(1-\alpha)}\left(1+\Delta^{-6/7}\right)\Big)\cdot |E_{i-1}|$.
    \item $\mathbb{E}[|E_i|\mid E_1,\cdots,E_{i-1}]\leq \Big(1-\frac{2}{ k}\frac{(1-\alpha)}{(1+\alpha)}\left(1-\Delta^{-6/7}\right)\Big)\cdot |E_{i-1}|$.
\end{enumerate}

\end{lemma}

\begin{proof}
    We say that an edge $e'$ is killed at the $(i-1)$'th iteration if $e'\in E_{i-1}$, and either $e'$ or a neighboring edge is sampled in the $(i-1)$'th iteration. That is, $i-1$ is the first iteration in which $e'$ is no longer surviving. We denote the number of killed edges at the $(i-1)$'th iteration by $K_{i-1}$. Observe that $|E_{i}| = |E_{i-1}| - K_{i-1}$. We analyze $\mathbb{E}[K_{i-1}\mid E_{i-1}]$ to derive $\mathbb{E}[|E_{i}|\mid E_{i-1}]$. 
    
    \subparagraph{Some notation:} For an edge $e=\{v,w\}$, let $p_{i-1}(e)$ be the probability that $e$ is sampled at the $(i-1)$th iteration. Similarly, for a set of edges $S$, let $p_{i-1}(S)$ be the probability that some edge in $S$ is sampled in the $(i-1)$th iteration. Recall that $p_{i-1}(v)$ is the probability that a node $v$ is labeled in the $(i-1)$th iteration; i.e. the probability that some incident edge is labeled in that iteration. In particular, $p_{i-1}(v) = \sum_{e=\{v,w\}} p_{i-1}(e)$. Since the graph is bipartite, we assume without loss of generality that $v\in L$ and $w\in R$ when we refer to an edge $\{v,w\}$, where $L$ and $R$ are the left and right sides of the graph, respectively. Furthermore, we assume without loss of generality that $N(u)\subseteq L$ and $N^2(u)\subseteq R$. Observe that when an edge $e$ is sampled, it kills the set of edges $\{e'\in E_{i-1}\mid e\cap e'\neq \emptyset\}$. Let $Y_e\in \{0,1\}$ be $1$ if and only if $e\in E_{i-1}$, and let $d_{i-1}(v)$ be the number of surviving edges incident to a node $v$ at the $(i-1)$th iteration. Observe that:

    \setcounter{equation}{0}
    \begin{align}
        \mathbb{E}[K_{i-1}\mid E_{i-1}] &= \sum_{e = \{v,w\}\in E_u} p_{i-1}(e)\cdot |\{e'\in E_{i-1}\mid e\cap e'\neq \emptyset\}|\\
        & =\sum_{e=\{v,w\}\in E_u} p_{i-1}(e)\cdot\big( 
        d_{i-1}(v)
        +d_{i-1}(w) 
        -Y_e\big)\\
        & = \left(\sum_{v\in L} p_{i-1}(v) \cdot d_{i-1}(v)\right) + \left(\sum_{w\in R} p_{i-1}(w) \cdot d_{i-1}(w)\right)- p_{i-1}(E_{i-1})\\
        &= \left(\sum_{v\in N(u)} p_{i-1}(v)\cdot d_{i-1}(v)\right) + \left(\sum_{w\in N^2(u)} p_{i-1}(w)\cdot d_{i-1}(w)\right) - p_{i-1}(E_{i-1})
    \end{align}

    \noindent (2) follows by a simple inclusion/exclusion argument, since the number of edges that are killed by $e=\{v,w\}$ is the number of edges killed that are incident to $v$, plus the number of edges killed that are incident to $w$, minus the number of edges killed that are incident to both (which could only be $e$ itself if happens to be surviving at the beginning of the $(i-1)$th iteration). (3) holds because $p_i(x)$ is the sum of the incident $p_i(e)$s. (4) follows since the only nodes in $L$ that are incident to edges in $E_{i-1}$ are the nodes in $N(u)$ and the only nodes in $R$ that are incident to edges in $E_{i-1}$ are the nodes in $N^2(u)$. Now we are ready to derive an upper and lower bounds on $\mathbb{E}[K_{i-1}\mid E_{i-1}]$. 
    
    \subparagraph{Upper Bound on $\mathbb{E}[K_{i-1}\mid E_{i-1}]$:} By plugging the upper bound on $p_{i-1}(v)$ and $p_{i-1}(w)$ from Lemma~\ref{lem:ProbNodeIsLabeled} into equation (4) above, we get that:

    \begin{align*}
        \mathbb{E}[K_{i-1}\mid E_{i-1}]\leq \frac{2(1+\alpha)}{(1-\alpha)\cdot k}\cdot |E_{i-1}|\cdot (1+\frac{\log\Delta}{25\Delta})\leq \frac{2(1+\alpha)}{(1-\alpha)\cdot k}\cdot |E_{i-1}|\cdot (1+\Delta^{-6/7})
    \end{align*}

    \subparagraph{Lower Bound on $\mathbb{E}[K_{i-1}\mid E_{i-1}]$:} First, observe that since $i\leq t=k\log\Delta/100$, we have that:
    
    $$p_{i-1}(E_{i-1}) \leq \frac{|E_{i-1}|}{(1-\alpha)\Delta(k+1) - (i-2)}\leq \frac{|E_{i-1}|}{(1-\alpha)(\Delta k) - t}\leq \frac{|E_{i-1}|}{(1-\alpha)(\Delta k)/2}\leq \frac{4|E_{i-1}|}{\Delta k}$$
    
    \noindent where the last inequality holds since $\alpha\geq 1/2$. Hence, by plugging the lower bound on $p_{i-1}(v)$ and $p_{i-1}(w)$ from Lemma~\ref{lem:ProbNodeIsLabeled} into equation (4) above, we get that:

    \begin{align*}
        \mathbb{E}[K_{i-1}\mid E_{i-1}]\geq \frac{2(1-\alpha)}{(1+\alpha)\cdot k}\cdot |E_{i-1}|\cdot (1-\Delta^{-7/8}) - \frac{4|E_{i-1}|}{\Delta k}\geq \frac{2(1-\alpha)}{(1+\alpha)\cdot k}\cdot |E_{i-1}|\cdot (1-\Delta^{-6/7})
    \end{align*}
    
    \subparagraph{Deriving Upper and Lower Bounds on $\mathbb{E}[|E_{i}|\mid E_{i-1}]$:} Since  $|E_i| = E_{i-1} - K_{i-1}$, we get that:

    \begin{align*}
        \left(1-\frac{2(1+\alpha)}{(1-\alpha)\cdot k}\left(1+\Delta^{-6/7}\right)\right)\cdot |E_{i-1}|\leq \mathbb{E}[|E_i|\mid E_{i-1}]\leq \left(1-\frac{2(1-\alpha)}{(1+\alpha)\cdot k}\left(1-\Delta^{-6/7}\right)\right)\cdot |E_{i-1}|
    \end{align*}
as desired. Finally, the lemma follows since $\mathbb{E}[|E_i| \mid E_1, \cdots, E_{i-1}] = \mathbb{E}[|E_i| \mid E_{i-1}]$.
\end{proof}

Now we are ready to give concentration results for $|E_i|$ for any $i\in [t=k\log\Delta/100]$. We start with showing a lower for $|E_i|$ in Lemma~\ref{lem:E_iConcentratesLowerBound} and then we show an upper bound on $|E_i|$ in Lemma~\ref{lem:E_iConcentratesUpperBound}. 

\begin{lemma}\label{lem:E_iConcentratesLowerBound}
    Let $\gamma_{\ell} =  \Big(1-\frac{2}{k}\frac{(1+\alpha)}{(1-\alpha)}\left(1+\Delta^{-6/7}\right)\Big)$. With probability at least $1-\exp{(-\Delta^{1/13})}$, it holds that for every $i\in [t=k\log\Delta/100]$. 

    $$|E_i|\geq\gamma_{\ell}^{i-1}\cdot ((1-\alpha)\Delta)^2\cdot (1-\Delta^{-1/5})$$
\end{lemma}
\begin{proof}
    The idea of the proof is to use the scaled martingale trick that was briefly discussed in Section~\ref{sec:martingale-overview}. Let $\mathcal{B}$ be the event in which the upper and lower bounds on $\mathbb{E}[|E_i|\mid E_{i-1}]$ from Lemma~\ref{lem:newlyKilledEdges} hold. By Lemma~\ref{lem:newlyKilledEdges}, the event $\mathcal{B}$ happens with probability at least $1-\exp{(-\Delta^{1/11})}$. Our analysis next is conditioned on the event $\mathcal{B}$ happening. For $i\geq 1$, we define the following random variable $F_i$ as follows. 

    $$F_i = \frac{|E_{i}|}{\gamma_{\ell}^{i-1}}$$

    \noindent Observe that $F_i$ is a supermartingale (see also Definition~\ref{def:martingale}) with a starting point $F_1 = |E_1|=|A_u|$ (recall that $A_u = (N(u)\times N^2(u))\cap E$ is the set of edges between $u$'s neighbors and their neighbors). This is because:

    $$\mathbb{E}[F_i\mid F_1,\cdots, F_{i-1} ]= \frac{1}{\gamma_{\ell}^{i-1}}\mathbb{E}[|E_{i}|\mid E_1,\cdots, E_{i-1}]\geq \frac{|E_{i-1}|}{\gamma_{\ell}^{i-2}} = F_{i-1}$$

    \noindent where the first equality follows from Observation~\ref{obs:ConditioningOnFunctionOfRV}, and the inequality follows from ~\Cref{lem:newlyKilledEdges}. Hence, to get a lower bound on $|E_i|$, we would like to apply the supermartingale inequality from Theorem~\ref{thm:SuperMartingaleConcentrationVariance}. For this, we analyze $Var[F_i\mid F_{i-1}]$ and $\mathbb{E}[F_i\mid F_{i-1}] - F_i$ for $i\geq 2$. First, observe that for any $i\in [t]$, we have that:

    $$\gamma_{\ell}^i \geq \gamma_{\ell}^t = \Big(1-\frac{2}{k}\frac{(1+\alpha)}{(1-\alpha)}\left(1+\Delta^{-6/7}\right)\Big)^t \geq \Big(1-\frac{8}{k}\Big)^{k\log\Delta/100}\geq e^{-8\log\Delta/50}\geq \Delta^{-1/10}$$
    
    \noindent where the second to last inequality follows from the fact that $1-x\geq e^{-2x}$ for any $0\leq x\leq 1/2$. Next, for $i\geq 2$, we analyze $Var[F_i\mid F_{i-1}]$ and then $\mathbb{E}[F_i\mid F_{i-1}] - F_i$. Observe that:

    \setcounter{equation}{0}
    \begin{align}
        Var[F_i\mid F_{i-1}] &= Var\left[F_i - \frac{F_{i-1}}{\gamma_{\ell}}\mid F_{i-1}\right]\\
        & = Var \left[\frac{|E_i|}{\gamma_{\ell}^{i-1}} - \frac{|E_{i-1}|}{\gamma_{\ell}^{i-1}}\mid E_{i-1} \right]\\
        &\leq \frac{1}{\gamma_{\ell}^{i-1}}\mathbb{E}\left[(|E_i| - |E_{i-1}|)^2\mid E_{i-1}\right]\\
        &\leq \frac{2\Delta}{\gamma_{\ell}^{i-1}}\mathbb{E}\left[\big||E_i| - |E_{i-1}|\big|\mid E_{i-1}\right]\\
        & = \frac{2\Delta}{\gamma_{\ell}^{i-1}}\big(\mathbb{E}[|E_{i-1}|\mid E_{i-1}] - \mathbb{E}[|E_i|\mid E_{i-1}]\big)\\
        &\leq \frac{2\Delta}{\gamma_{\ell}^{i-1}}\big(|E_{i-1}| - |E_{i-1}| + 12|E_{i-1}|/k\big)\\
        &\leq \frac{2\Delta}{\gamma_{\ell}^{i-1}}\cdot \frac{12|E_{i-1}|}{k}\\
        &\leq \frac{96\Delta^{3.1}}{k}
    \end{align}
    
    \noindent where (1) follows from Observation~\ref{obs:conditionalVariance}, (3) follows by the definition of variance, (4) follows since the maximum value of $|E_i - E_{i-1}|$ is $(1+\alpha)\Delta\leq 2\Delta$, (5) follows since $|E_{i-1}|\geq |E_i|$, (6) follows since $\mathbb{E}[|E_i|\mid E_{i-1}]\geq (1-12/k)|E_{i-1}|$ by Lemma~\ref{lem:newlyKilledEdges}, and (8) follows since $\gamma_{\ell}^{i-1}\geq \Delta^{-1/10}$ and $|E_{i-1}|\leq |E_1|\leq (1+\alpha)^2\Delta^2\leq 4\Delta^2$. 
    
    \noindent Next, we analyze $\mathbb{E}[F_i\mid F_{i-1}] - F_i$ for any $i\geq 2$. Observe that:
    \setcounter{equation}{0}
    \begin{align}
        \mathbb{E}[F_i\mid F_{i-1}] - F_i &=
        \frac{1}{\gamma_{\ell}^{i-1}}(\mathbb{E}[|E_i|\mid E_{i-1}] - |E_{i}|)\\
        &\leq \frac{1}{\gamma_{\ell}^{i-1}}\left(|E_{i-1}| - |E_{i}|\right)\\
        &\leq 4 \Delta^{1.1}
    \end{align}
    \noindent where (2) holds since $|E_i|\leq |E_{i-1}|$ (as there are more surviving edges at iteration $i-1$ compared to iteration $i$), and (3) holds since $|E_{i-1}| - |E_i|\leq 2(1+\alpha)\Delta\leq 4\Delta$ (as we kill at most $2(1+\alpha)\Delta-1$ edges in each iteration), and $\gamma_{\ell}^{i-1}\geq \Delta^{-1/10}$. Furthermore, observe that $F_1=|E_1| = |A_u|\geq ((1-\alpha)\Delta)^2-(1-\alpha)\Delta$, as $u$ has at least $(1-\alpha)\Delta$ neighbors, and each of them has at least $(1-\alpha)\Delta -1$ edges incident to it in $A_u$. Hence, we can plug our bounds on $Var[F_i\mid F_{i-1}]$ and $\mathbb{E}[F_i\mid F_{i-1}] - F_i$ into Theorem~\ref{thm:SuperMartingaleConcentrationVariance} to get that for a given $i\in [t]$, conditioned on the event $\mathcal{B}$ happening, we have that: 

    \setcounter{equation}{0}
    \begin{align}\mathbb{P}[F_i\leq ((1-\alpha)\Delta)^2- (1-\alpha)\Delta -\Delta^{7/4}\mid \mathcal{B}]&\leq \mathbb{P}[F_i\leq F_1-\Delta^{7/4}\mid \mathcal{B})\\
    &\leq \exp{\left(-\frac{\Delta^{3.5}}{2\left((\sum_{j=1}^i 96\Delta^{3.1}/k) + 4\Delta^{1.1}\cdot \Delta^{7/4}/3\right)}\right)}\\
    &\leq \exp{\left(-\frac{\Delta^{3.5}}{2\Delta^{3.1}\log\Delta + 8\Delta^{2.85}}\right)}\\
    &\leq \exp{\left(-\Delta^{3/10}\right)}\end{align}
    
    \noindent where (3) follows since $i\leq t=k\log\Delta/100$, and therefore $\sum_{j=1}^i96\Delta^{3.1}/k\leq t\cdot 96\Delta^{3.1}/k \leq \Delta^{3.1}\log\Delta$. Furthermore, observe that $\mathbb{P}[F_i\leq ((1-\alpha)\Delta)^2(1-\Delta^{-1/5})]\leq \mathbb{P}[F_i\leq ((1-\alpha)\Delta)^2- (1-\alpha)\Delta -\Delta^{7/4}]$. This is because $((1-\alpha)\Delta)^2(1-\Delta^{-1/5})\leq ((1-\alpha)\Delta)^2- (1-\alpha)\Delta -\Delta^{7/4}$. Hence, to get rid of the conditioning on the event $\mathcal{B}$ happening, we just use a union bound argument to get that for a given $i\in [t]$:

    \begin{align*}
        \mathbb{P}[F_i\leq ((1-\alpha)\Delta)^2\cdot (1-\Delta^{-1/5})]&\leq \mathbb{P}[F_i\leq ((1-\alpha)\Delta)^2- (1-\alpha)\Delta -\Delta^{7/4}\mid \mathcal{B}]\\
        &\leq \exp{\left(-\Delta^{3/10}\right)} + \exp{\left(-\Delta^{1/11}\right)}\\
        &\leq \exp{\left(-\Delta^{1/12}\right)}
    \end{align*}

    \noindent By another simple union bound argument, we get that with probability at least $1-\exp{(-\Delta^{1/13})}$, \emph{for every $i\in [t]$}, $F_i\geq ((1-\alpha)\Delta)^2\cdot (1-\Delta^{-1/5})$. Hence, the lemma follows since $F_i = |E_i|/\gamma_{\ell}^{i-1}$
\end{proof}

\begin{lemma}\label{lem:E_iConcentratesUpperBound}
    Let $\gamma_h =  \Big(1-\frac{2}{k}\frac{(1-\alpha)}{(1+\alpha)}\left(1-\Delta^{-6/7}\right)\Big)$. With probability at least $1-\exp{(-\Delta^{1/13})}$, it holds that for every $i\in [t=k\log\Delta/100]$. 

    $$|E_i|\leq\gamma_h^{i-1}\cdot ((1+\alpha)\Delta)^2(1+\Delta^{-1/5})$$
\end{lemma}
\begin{proof}
    The proof is very similar to the proof of Lemma~\ref{lem:E_iConcentratesLowerBound}. As in the proof on Lemma~\ref{lem:E_iConcentratesLowerBound}, we define the random variable:

    $$F_i = \frac{|E_i|}{\gamma^{i-1}_h}$$
    \noindent Observe that $F_i$ is a submartingale (see also Definition~\ref{def:martingale}) with starting point $F_1=|E_1| = |A_u|$. This is because for any $i\geq 2$, we have that:

    $$\mathbb{E}[F_i\mid F_1,\cdots, F_{i-1}] = \frac{1}{\gamma_h^{i-1}}\mathbb{E}[|E_i|\mid E_1,\cdots, E_{i-1}]\leq \frac{|E_{i-1}|}{\gamma^{i-2}_h} = F_{i-1}$$

    \noindent Hence, we would like to use Theorem~\ref{thm:SubMartingaleConcentrationVariance} to get a concentration result for $F_i$. For this, we need to analyze $Var[F_i\mid F_{i-1}]$ and $F_i-\mathbb{E}[F_i\mid F_{i-1}]$. Observe that by a similar calculation to the one that we provided for $\gamma^{i}_{\ell}$ in the proof of Lemma~\ref{lem:E_iConcentratesLowerBound}, it holds that $\gamma^{i}_h\geq \Delta^{-1/10}$, for any $i\in [t]$. Hence, by the exact same calculation of $Var[F_i\mid F_{i-1}]$ in the proof of Lemma~\ref{lem:E_iConcentratesLowerBound}, it holds that $Var[F_i\mid F_{i-1}]\leq 96\Delta^{3.1}/k$. 

    \noindent It remains to analyze $F_i - \mathbb{E}[F_i\mid F_{i-1}]$. Observe that:
    \setcounter{equation}{0}
    \begin{align}
        F_i - \mathbb{E}[F_i\mid F_{i-1}] & = \frac{1}{\gamma^{i-1}_h}(|E_i| - \mathbb{E}[|E_i|\mid E_{i-1}])\\
        &\leq \frac{1}{\gamma^{i-1}_h}(|E_{i-1}| - |E_{i-1}| + 12|E_{i-1}|/k)\\
        & \leq \frac{48\Delta^2}{k\gamma^{i-1}_h}\\
        &\leq 336\Delta^{1.1}
    \end{align}

    \noindent where (2) follows since $|E_i|\leq |E_{i-1}|$ and $\mathbb{E}[|E_i|\mid E_{i-1}]\geq 1-12|E_{i-1}|/k$, where the latter follows from Lemma~\ref{lem:newlyKilledEdges}, (3) follows since $|E_{i-1}|\leq |E_1|\leq ((1+\alpha)\Delta)^2\leq 4\Delta^2$, and (4) follows since $k\geq \Delta/7$ and $\gamma_h^{i-1}\geq \Delta^{-0.1}$.

    \noindent The rest of the proof is exactly as the proof of Lemma~\ref{lem:E_iConcentratesLowerBound}, where instead of using Theorem~\ref{thm:SuperMartingaleConcentrationVariance}, we use Theorem~\ref{thm:SubMartingaleConcentrationVariance} to get concentration for $F_i$, and then deduce the desired bound for $|E_i|$. Observe that our analysis above is conditioned on the event $\mathcal{B}$ happening, where $\mathcal{B}$ is the event in which the bounds on $\mathbb{E}[|E_i|\mid E_{i-1}]$ hold. By Theorem~\ref{thm:SubMartingaleConcentrationVariance}, we have that:

    \begin{align*}\mathbb{P}[F_i\geq ((1+\alpha)\Delta)^2 + \Delta^{7/4}\mid \mathcal{B}]&\leq\mathbb{P}[F_i\geq F_1 + \Delta^{7/4}\mid \mathcal{B}]\\
    &\leq \exp{\left(-\frac{\Delta^{3.5}}{2\left((\sum_{j=1}^t 96\Delta^{3.1}/k) + 336\Delta^{1.1}\cdot \Delta^{7/4}/3\right)}\right)}\\
    &\leq \exp{\left(-\Delta^{3/10}\right)}\end{align*}

    \noindent Finally, by two simple applications of a union bound argument, the probability that there exists an $i\in [t]$ for which $F_i\geq ((1+\alpha)\Delta)^2 + \Delta^{7/4}$ is at most $\exp{(-\Delta^{1/13})}$.


    \noindent Hence, with probability at least $1-\exp{(-\Delta^{-1/3})}$ for every $i\in [t]$, $F_i\leq ((1+\alpha)\Delta)^2+\Delta^{7/4}\leq ((1+\alpha)\Delta)^2(1+\Delta^{-1/5})$. The proof of the lemma follows since $F_i = |E_i|/\gamma^{i-1}_h$
    
\end{proof}

\subsection{A Lower Bound on \texorpdfstring{$|\mathcal{M}_u|$}{Mu}}\label{sec:LowerBoundE}

Recall that by definition, whenever a surviving edge gets sampled at some iteration $i$, it gets added to the matching $\M_u$. We can now utilize the sharp concentration bounds on the number of surviving edges to obtain concentration bounds on the size of the matching $\M_u$. We first give a lower bound on $|\M_u|$ in this section and give an upper bound in the next.

\begin{lemma}\label{lem:LowerBoundMu}
    With probability at least $1-\exp{(-\Delta^{-1/15})}$ it holds that:

    $$|\mathcal{M}_u|\geq \frac{\Delta}{2}\cdot  \frac{(1-\alpha)^3}{(1+\alpha)^2}\cdot (1-\Delta^{-1/300})$$
    \begin{proof}
        It suffices to show that the size of $\mathcal{M}_u$ is large enough after $t=k\log\Delta/100$ iterations of the while loop in Algorithm~\ref{Alg: SeqLuby_u}, as the size of $\mathcal{M}_u$ can only increase after the $t$th iteration. Let $\mathcal{B}$ be the event in which the upper and lower bounds on $|E_1|,\cdots,|E_t|$ from Lemma~\ref{lem:E_iConcentratesLowerBound} and Lemma~\ref{lem:E_iConcentratesUpperBound} hold. Observe that the event $\mathcal{B}$ happens with probability at least $1-\exp{(-\Delta^{1/14})}$. We prove the lemma conditioned on the event $\mathcal{B}$ happening and then get rid of this condition by a simple union bound argument. To prove the lemma in the case where the event $\mathcal{B}$ happens we apply Theorem~\ref{thm:ShiftedMartingaleLowerBound}, as follows. Recall that $Z_i$ is an indicator random variable indicating whether we add an edge to $\mathcal{M}_u$ in the $i$th iteration of the while loop in Algorithm~\ref{Alg: SeqLuby_u}, and recall that $q_i$ is a random variable representing the probability that we add an edge in the $i$th iteration given the randomness so far. In other words: $q_i = \mathbb{P}[Z_i=1\mid Z_1,\cdots,Z_{i-1}] = \mathbb{E}[Z_i\mid Z_1,\cdots, Z_{i-1}]$. As discussed in Section~\ref{sec:OverviewDense}, we have that:

        $$q_i = \frac{|E_i|}{|E_u| - (i-1)}$$
        \noindent This is because the probability that we add an edge to the matching in the $i$th iteration is the probability that we pick a surviving edge among the remaining $|E_u| - (i-1)$ edges. Next, we show an upper and a lower bound on $\sum_{i=1}^t q_i$, conditioned on the event $\mathcal{B}$ happening. Observe that $|E_u|\leq ((1+\alpha)\Delta)(k+1)$, as the number of edges in $E_u$ is exactly the number of edges touching the nodes in $N^2(u)$ plus the number of edges touching $u$ itself. Since each node has degree at most $(1+\alpha)\Delta$, we get at most $(1+\alpha)\Delta(k+1)$ edges in total in $E_u$ (recall that $|N^2(u)|=k$).

        \subparagraph{A lower bound on $\sum_{i=1}^t q_i$.} Let $\gamma_{\ell} =  \Big(1-\frac{2}{k}\frac{(1+\alpha)}{(1-\alpha)}\left(1+\Delta^{-6/7}\right)\Big)$. Observe that: 

        \setcounter{equation}{0}
        \begin{align}
            \sum_{i=1}^t q_i &\geq \sum_{i=1}^t \frac{|E_i|}{(1+\alpha)\Delta(k+1) - (i-1)}\\
            &\geq \sum_{i=1}^t\frac{\gamma_{\ell}^{i-1}\cdot ((1-\alpha)\Delta)^2\cdot (1-\Delta^{-1/5})}{(1+\alpha)\Delta(k+1)}\\
            &=\frac{ ((1-\alpha)\Delta)^2\cdot (1-\Delta^{-1/5})}{(1+\alpha)\Delta(k+1)}\sum_{i=1}^t\gamma_{\ell}^{i-1}\\
            &=\frac{ (1-\alpha)^2\Delta\cdot (1-\Delta^{-1/5})}{(1+\alpha)(k+1)}\cdot \frac{1-\gamma_{\ell}^t}{1-\gamma_{\ell}}\\
            &= \frac{ (1-\alpha)^2\Delta\cdot (1-\Delta^{-1/5})}{(1+\alpha)(k+1)}\cdot \frac{1-\Big(1-\frac{2}{k}\frac{(1+\alpha)}{(1-\alpha)}\left(1+\Delta^{-6/7}\right)\Big)^t}{1-\Big(1-\frac{2}{k}\frac{(1+\alpha)}{(1-\alpha)}\left(1+\Delta^{-6/7}\right)\Big)}\\
            &\geq \frac{ (1-\alpha)^2\Delta\cdot (1-\Delta^{-1/5})}{(1+\alpha)(k+1)}\cdot\frac{1-\Delta^{-1/100}}{\frac{2}{k}\frac{(1+\alpha)}{(1-\alpha)}\left(1+\Delta^{-6/7}\right)}\\
            &\geq \frac{ (1-\alpha)^3\Delta\cdot k}{2(1+\alpha)^2(k+1)}\cdot\frac{(1-\Delta^{-1/5})(1-\Delta^{-1/100})}{(1+\Delta^{-6/7})}\\
            &\geq \frac{\Delta}{2}\cdot  \frac{(1-\alpha)^3}{(1+\alpha)^2}\cdot (1-\Delta^{-1/200})
        \end{align}

        \noindent where (2) follows from Lemma~\ref{lem:E_iConcentratesLowerBound}, (4) follows from the geometric series formula, (6) follows since $\gamma_{\ell}^{t}\leq \Delta^{-1/100}$, and (8) follows since $k\geq \Delta/7$ and from a simple arithmetic. Now we show an upper bound on $\sum_{i=1}^t q_i$. 

        \subparagraph{An upper bound on $\sum_{i=1}^t q_i$.} Let $\gamma_h =  \Big(1-\frac{2}{k}\frac{(1-\alpha)}{(1+\alpha)}\left(1-\Delta^{-6/7}\right)\Big)$. Observe that $|E_u|\geq ((1-\alpha)\Delta) k$, by just counting the degrees of the nodes in $N^2(u)$. Hence, we have that:

        \setcounter{equation}{0}
        \begin{align}
            \sum_{i=1}^t q_i &= \sum_{i=1}^t \frac{|E_i|}{|E_u| - (i-1)}\\
            &\leq \sum_{i=1}^t\frac{\gamma_h^{i-1}\cdot ((1+\alpha)\Delta)^2(1+\Delta^{-1/5})}{(1-\alpha)\Delta k-t}\\
            &=\frac{ ((1+\alpha)\Delta)^2(1+\Delta^{-1/5})}{(1-\alpha)\Delta k-k\log\Delta/100}\sum_{i=1}^t\gamma_{h}^{i-1}\\
            &\leq \frac{ ((1+\alpha)\Delta)^2(1+\Delta^{-1/5})(1+\Delta^{-8/10})}{(1-\alpha)\Delta k}\sum_{i=1}^t\gamma_{h}^{i-1}\\
            &=\frac{ (1+\alpha)^2\Delta(1+\Delta^{-1/5})(1+\Delta^{-8/10})}{(1-\alpha) k} \frac{1-\gamma_{h}^t}{1-\gamma_{h}}\\
            &\leq \frac{ (1+\alpha)^2\Delta(1+\Delta^{-1/5})(1+\Delta^{-8/10})}{(1-\alpha) k}\cdot \frac{1}{1-\Big(1-\frac{2}{k}\frac{(1-\alpha)}{(1+\alpha)}\left(1-\Delta^{-6/7}\right)\Big)}\\
            &\leq \frac{ (1+\alpha)^2\Delta(1+\Delta^{-1/5})(1+\Delta^{-8/10})}{(1-\alpha) k}\cdot\frac{1}{\frac{2}{k}\frac{(1-\alpha)}{(1+\alpha)}\left(1-\Delta^{-6/7}\right)}\\
            &\leq \frac{\Delta}{2}\cdot\frac{(1+\alpha)^3}{(1-\alpha)^2}\cdot\frac{(1+\Delta^{-1/3})(1+\Delta^{-8/10})}{(1-\Delta^{-6/7})}\\
            &\leq \frac{\Delta}{2}\cdot\frac{(1+\alpha)^3}{(1-\alpha)^2}\cdot (1+\Delta^{-1/100})
        \end{align}

        \noindent where (2) follows from Lemma~\ref{lem:E_iConcentratesUpperBound} and the rest follows from simple arithmetic. Having proved the upper and lower bounds on $\sum_{i=1}^t q_i$, we apply Theorem~\ref{thm:ShiftedMartingaleLowerBound} with $X_i=Z_i$, $S_t=\sum_{i=1}^t Z_i\leq |\mathcal{M}_u|$, $p_i=q_i$, $P^{\ell} = \frac{\Delta}{2}\cdot  \frac{(1-\alpha)^3}{(1+\alpha)^2}\cdot (1-\Delta^{-1/200})$, $P^h = \frac{\Delta}{2}\cdot\frac{(1+\alpha)^3}{(1-\alpha)^2}\cdot (1+\Delta^{-1/100})$,  $M=1$, and $\lambda = P^{\ell}-\Delta^{0.6}$ to get that:

        $$\mathbb{P}(|M_u|\leq \lambda\mid \mathcal{B})\leq \mathbb{P}[S_t\leq \lambda\mid \mathcal{B})\leq \exp{\left(-\frac{(P^{\ell}-\lambda)^2}{8P^h + 2(P^{\ell} - \lambda)/3}\right)}\leq \exp{\left(-\frac{\Delta^{1.2}}{100\Delta}\right)}\leq \exp{(-\Delta^{1/10})}$$

        \noindent Finally, since the event $\mathcal{B}$ happens with probability at least $1-\exp(-\Delta^{1/14})$, by a simple union bound, we get that $|\mathcal{M}_u|\geq P^{\ell} - \Delta^{0.6} \geq \frac{\Delta}{2}\cdot  \frac{(1-\alpha)^3}{(1+\alpha)^2}\cdot (1-\Delta^{-1/300})$ with probability at least $1-\exp{(-\Delta^{1/15})}$, as desired. 
    \end{proof}
\end{lemma}

\subsection{An Upper Bound on \texorpdfstring{$|\mathcal{M}_u|$}{Mu}}\label{sec:UpperBoundE}

Observe that all our analysis takes into consideration only the first $t=k\log\Delta/100$ iterations of the while loop in Algorithm~\ref{Alg: SeqLuby_u}. To show an upper bound on $|\mathcal{M}_u|$, it doesn't suffice to consider only the first $t$ iterations as we also need to make sure that not too many edges are added to $|\mathcal{M}_u|$ after the $t$th iteration. Our proof for the upper bound on $|\mathcal{M}_u|$ is split into two parts. First, in Lemma~\ref{lem:FewUnlabeledNodes}, we show that Algorithm~\ref{Alg: SeqLuby_u} doesn't add too many more edges to $|\mathcal{M}_u|$ after the $t$th iteration. Then, in Lemma~\ref{lem:UpperBoundMu}, we show an upper bound on the number of edges that are added to $|\mathcal{M}_u|$ up to the $t$th iteration, and deduce the upper bound on $|\mathcal{M}_u|$. Recall that we say that a node is labeled at iteration $i$ if one of its incident edges was sampled in one of the first $i$ iterations. Observe that a labeled node at iteration $i$ can't be matched at iteration $j>i$, as all of its incident edges are non-surviving at iteration $j$ (see also Definition~\ref{def:surviving edges}). In the following lemma, we show that the number of unlabeled neighbors of $u$ after the $t$'s iteration is small. 

\begin{lemma}\label{lem:FewUnlabeledNodes}
    With probability at least $1-\exp{(-\Delta^{1/12})}$, the number of unlabeled nodes in $N(u)$ after the $t$th iteration is processed is at most $\Delta^{0.9975}$.
\end{lemma}
\begin{proof}
    We use the scaled martingale trick that was briefly discussed in Section~\ref{sec:martingale-overview}. Let $\mathcal{B}$ be the event in which the upper and lower bounds on the probability that a node is labeled from Lemma~\ref{lem:ProbNodeIsLabeled} hold. We first provide an analysis conditioned on the event $\mathcal{B}$ happening, and then we remove this assumption by a simple union bound argument. By Lemma~\ref{lem:ProbNodeIsLabeled}, the event $\mathcal{B}$ happens with probability at least $1-\exp{(-\Delta^{1/11})}$. Let $W_i$ be the number of unlabeled nodes in $N(u)$ after the $i$th iteration is processed. Observe that by Lemma~\ref{lem:ProbNodeIsLabeled}, we have that:
    
    $$\mathbb{E}[W_i\mid W_{i-1}] \leq \left(1-\frac{1-\alpha}{(1+\alpha)k}\cdot (1-\Delta^{-7/8})\right) W_{i-1}$$

    \noindent This is because by Lemma~\ref{lem:ProbNodeIsLabeled}, each node $v$ is labeled with probability at least $\frac{1-\alpha}{(1+\alpha)k}\cdot (1-\Delta^{-7/8})$ at any iteration. Therefore, the total probability mass that the unlabeled nodes have at the beginning of the $i$th iteration is $W_{i-1}\cdot\frac{1-\alpha}{(1+\alpha)k}\cdot (1-\Delta^{-7/8})$.
    
    \noindent Next, let $\gamma = \left(1-\frac{1-\alpha}{(1+\alpha)k}\cdot (1-\Delta^{-7/8})\right)$. We define the following random variable:
    
    $$F_i = \frac{W_i}{\gamma^{i-1}}$$
    \noindent Observe that $F_i$ is a submartingale with a starting point $F_1 = W_1 = |N(u)|$. This is because for any $i\geq 2$:

    $$\mathbb{E}[F_i\mid F_{i-1}] = \frac{1}{\gamma^{i-1}}\mathbb{E}[W_i\mid W_{i-1}]\leq \frac{W_{i-1}}{\gamma^{i-2}} = F_{i-1}$$

    \noindent Hence, we would like to use Theorem~\ref{thm:SubMartingaleConcentrationVariance} to get a concentration result for $F_i$ and then deduce a concentration result for $W_i$. For this, we analyze $Var[F_i\mid F_{i-1}]$ and then $F_i - \mathbb{E}[F_i\mid F_{i-1}]$. Observe that:

    \setcounter{equation}{0}
    \begin{align}
        Var[F_i\mid F_{i-1}] &= Var[F_i - F_{i-1}/\gamma\mid F_{i-1}]\\
        & = \frac{1}{\gamma^{i-1}}Var[W_i-W_{i-1}\mid W_{i-1}]\\
        &= \frac{1}{\gamma^{i-1}}\mathbb{E}[(W_i-W_{i-1})^2\mid W_{i-1}]\\
        &\leq \frac{1}{\gamma^{i-1}}\mathbb{E}[|W_i-W_{i-1}|\mid W_{i-1}]\\
        &\leq \frac{1}{\gamma^{i-1}}\big(\mathbb{E}[W_{i-1}\mid W_{i-1}]-\mathbb{E}[W_{i}\mid W_{i-1}]\big)\\
        &\leq \frac{1}{\gamma^{i-1}}\big(W_{i-1} - (1 - 8/k)W_{i-1}\big)\\
        &\leq \frac{8W_{i-1}}{k\gamma^{i-1}}\\
        &\leq \frac{16\Delta^{1.1}}{k}
    \end{align}
    
    \noindent where (1) follows from Observation~\ref{obs:conditionalVariance}, (3) follows from the definition of variance, (4) follows since the maximum value of $|W_i-W_{i-1}|$ is $1$, (5) follows since $W_{i-1}\geq W_i$, (6) follows from Lemma~\ref{lem:ProbNodeIsLabeled} which implies that $\mathbb{E}[W_i\mid W_{i-1}]\geq (1-8/k)W_{i-1}$, and (8) follows since $W_i\leq W_1\leq (1+\alpha)\Delta\leq 2\Delta$ and $\gamma^{i-1}\geq \Delta^{-0.1}$.

    \noindent Next, we analyze $F_i - \mathbb{E}[F_i\mid F_{i-1}]$.

    \setcounter{equation}{0}
    \begin{align}
        F_i - \mathbb{E}[F_i\mid F_{i-1}] &= \frac{1}{\gamma^{i-1}}\big(W_i - \mathbb{E}[W_i\mid W_{i-1}]\big)\\
        &\leq \frac{1}{\gamma^{i-1}}\big(W_{i-1} - (1-8/k)W_{i-1}\big)\\
        & = \frac{8W_{i-1}}{k\gamma^{i-1}}\\
        &\leq 112\Delta^{0.1}
    \end{align}

    \noindent where (2) follows since by  Lemma~\ref{lem:ProbNodeIsLabeled}, $\mathbb{E}[W_i\mid W_{i-1}]\geq (1-8/k)W_{i-1}$, and (4) follows since $W_{i-1}\leq 2\Delta$, $k\geq \Delta/7$, and $\gamma^{i-1}\geq \Delta^{0.1}$. Therefore, by plugging the bounds on $Var[F_i\mid F_{i-1}]$ and $F_i-\mathbb{E}[F_i\mid F_{i-1}]$ into Theorem~\ref{thm:SubMartingaleConcentrationVariance} for $\lambda=\Delta$, we get that:

    \begin{align*}
        \mathbb{P}[F_t\geq (1+\alpha)\Delta + \lambda]\leq \mathbb{P}[F_t\geq F_1 + \lambda\mid \mathcal{B}]&\leq \exp{\left(-\frac{\lambda^2}{2\big((\sum_{i=1}^t16\Delta^{1.1}/k)+112\Delta^{0.1}\lambda\big)}\right)}\\
        &= \exp{\left(-\frac{\Delta^2}{2\big(16\Delta^{1.1}\log\Delta/100+224\Delta^{1.1})}\right)}\\
        &=\exp(-\Delta^{0.8})
    \end{align*}
    \noindent where the first inequality holds since $F_1\leq (1+\alpha)\Delta$. By a union bound over the event in which the event $\mathcal{B}$ doesn't happen we get that:

    $$\mathbb{P}[F_t\geq 3\Delta]\leq \exp(-\Delta^{1/12})$$

    \noindent Finally, since 
    
    $$W_t = F_t\cdot \gamma^{t-1}\leq 3\Delta\cdot \left(1-\frac{1-\alpha}{(1+\alpha)k}\cdot (1-\Delta^{-7/8})\right)^{t-1}\leq 3\Delta\cdot \exp{\left(-\frac{(k\log\Delta/100)(1-\Delta^{-7/8})}{3k}\right)}\leq \Delta^{1-1/400}$$ 
    
    \noindent the claim follows.
\end{proof}

\noindent Now we are ready to give an upper bound on $|\mathcal{M}_u|$.

\begin{lemma}\label{lem:UpperBoundMu}
    With probability at least $1-\exp{(-\Delta^{-1/15})}$ it holds that:

    $$|\mathcal{M}_u|\leq \frac{\Delta}{2}\cdot\frac{(1+\alpha)^3}{(1-\alpha)^2}\cdot (1+\Delta^{-1/500})$$
    \begin{proof}
        Let $\mathcal{B}$ be the event in which the statement from Lemma~\ref{lem:E_iConcentratesUpperBound} holds. That is, in the event $\mathcal{B}$ we have upper bounds on $|E_1|,\cdots, |E_t|$. As usual, we first prove the lemma conditioned on the event $\mathcal{B}$ happening. The proof is similar to the proof of Lemma~\ref{lem:LowerBoundMu} where we showed a lower bound on $|\mathcal{M}_u|$. Recall that $Z_i$ is an indicator random variable indicating whether we add an edge to $\mathcal{M}_u$ in the $i$th iteration of the while loop in Algorithm~\ref{Alg: SeqLuby_u}, and $q_i = \mathbb{E}[Z_i\mid Z_1,\cdots, Z_{i-1}]$. In the proof of Lemma~\ref{lem:LowerBoundMu} we showed that $\sum_{i=1}^t q_i\leq \frac{\Delta}{2}\cdot\frac{(1+\alpha)^3}{(1-\alpha)^2}\cdot (1+\Delta^{-1/100})$. Let $|\mathcal{M}^t_u|$ be the size of $\mathcal{M}_u$ after the first $t$ iterations were processed. We apply Theorem~\ref{thm:ShiftedMartingaleUpperBound} with $X_i=Z_i$, $S_t = |\mathcal{M}^t_u|$, $p_i=q_i$, $M=1$, $P=\frac{\Delta}{2}\cdot\frac{(1+\alpha)^3}{(1-\alpha)^2}\cdot (1+\Delta^{-1/100})$, and $\lambda = P+\Delta^{0.6}$ to get that:

        \begin{align*}
            \mathbb{P}[|\mathcal{M}^t_u|\geq \lambda\mid \mathcal{B}]\leq \exp{\left(-\frac{\Delta^{1.2}}{8\Delta}\right)}\leq \exp{(-\Delta^{0.1})}
        \end{align*}
        
        \noindent By a simple union bound argument (to get rid of the conditioning on $\mathcal{B}$), we get that $|\mathcal{M}^t_u|\leq P+\Delta^{0.6} \leq \frac{\Delta}{2}\cdot\frac{(1+\alpha)^3}{(1-\alpha)^2}\cdot (1+\Delta^{-1/200})$ with probability at least $1-\exp{(-\Delta^{1/14})}$. Finally, by Lemma~\ref{lem:FewUnlabeledNodes}, with probability at least $1-\exp{(-\Delta^{1/12})}$, we can't add more than $\Delta^{0.9975}$ edges to $\mathcal{M}_u$ after the $t$th iteration. Hence, with probability at least $1-\exp{(-\Delta^{1/15})}$, it holds that $|\mathcal{M}_u|\leq  \frac{\Delta}{2}\cdot\frac{(1+\alpha)^3}{(1-\alpha)^2}\cdot (1+\Delta^{-1/500})$, as desired.
        
    \end{proof}
\end{lemma}

\subsection{Finishing the Proof}
Since $\M_u$ is exactly the matching adjacent to nodes in $N(u)$, the bounds on $|\M_u|$ suffice to prove \Cref{lem:RecursiveRegularity}.

\begin{proof}[Proof of Lemma~\ref{lem:RecursiveRegularity}]
    First, by Observation~\ref{ob:SameMatchingSeqLuby_u}, Algorithm~\ref{alg:SeqL} and Algorithm~\ref{Alg: SeqLuby_u} produce the same distributions over matchings in $(N(u)\times N^2(u))\cap E$. Furthermore, by Lemma~\ref{lem:UpperBoundMu}, with probability at least $1-\exp{(-\Delta^{1/15})}$, it holds that $|\mathcal{M}_u|\leq \frac{\Delta}{2}\cdot\frac{(1+\alpha)^3}{(1-\alpha)^2}\cdot (1+\Delta^{-1/500})$. Hence, the number of unmatched neighbors of $u$ is at least:
    \setcounter{equation}{0}
    \begin{align}
        deg'(u)&\geq (1-\alpha)\Delta - |\mathcal{M}_u|\\
        &\geq (1-\alpha)\Delta - \frac{\Delta}{2}\cdot\frac{(1+\alpha)^3}{(1-\alpha)^2}\cdot (1+\Delta^{-1/500})\\
        &=\frac{\Delta}{2}\left(2-2\alpha - \frac{(1+\alpha)^3}{(1-\alpha)^2}\cdot (1+\Delta^{-1/500})\right)\\
        &\geq \frac{\Delta}{2}\left(2-2\alpha - (1+8\alpha)\cdot (1+\Delta^{-1/500})\right)\\
        &\geq \frac{\Delta}{2}\left(2-2\alpha - (1+8\alpha+\Delta^{-1/600})\right)\\
        &=\frac{\Delta}{2}\left(1-(10\alpha+\Delta^{-1/600})\right)
    \end{align}
    
    \noindent where (4) follows since $\frac{(1+\alpha)^3}{(1-\alpha)^2}\leq (1+8\alpha)$ for $\alpha\leq 1/10$, and rest follows from simple arithmetic. For the upper bound on $deg'(u)$, we use Lemma~\ref{lem:LowerBoundMu}, which says that with probability at least $1-\exp{(-\Delta^{1/15})}$, it holds that $ |\mathcal{M}_u|\geq \frac{\Delta}{2}\cdot  \frac{(1-\alpha)^3}{(1+\alpha)^2}\cdot (1-\Delta^{-1/300})$. Hence, we have that:

    \setcounter{equation}{0}
    \begin{align}
        deg'(u)&\leq (1+\alpha)\Delta - |\mathcal{M}_u|\\
        &\leq (1+\alpha)\Delta - \frac{\Delta}{2}\cdot  \frac{(1-\alpha)^3}{(1+\alpha)^2}\cdot (1-\Delta^{-1/300})\\
        &=\frac{\Delta}{2}\left(2+2\alpha -\frac{(1-\alpha)^3}{(1+\alpha)^2}\cdot (1-\Delta^{-1/300})\right)\\
        &\leq \frac{\Delta}{2}\left(2+2\alpha -(1-8\alpha)\cdot (1-\Delta^{-1/300})\right)\\
        &\leq \frac{\Delta}{2}\left(2+2\alpha -(1-8\alpha-\Delta^{-1/600})\right)\\
        &\leq \frac{\Delta}{2}\left(1+(10\alpha+\Delta^{-1/600})\right)
    \end{align}

    \noindent where (4) follows since $\frac{(1-\alpha)^3}{(1+\alpha)^2}\geq 1-8\alpha$, for $\alpha\leq 1/10$. Hence, by a simple union bound argument, we get the desired bounds on $deg'(u)$ with probability at least $1-\exp{(-\Delta^{1/16})}$.
      
\end{proof}
\section{Warmup: A \texorpdfstring{$\poly(1/\epsilon)$}{poly(1/eps)}-Round Algorithm for General Regular Graphs}
\label{sec:warmup}


In this section, we give a simple randomized algorithm that finds a $(1+\epsilon)$-approximate matching in a number of rounds which only depends on the accuracy $\epsilon$ and not the graph size or degree.

\thmpolywarmup*

Our algorithm for proving Theorem~\ref{thm:warmup} runs in two stages. In the first stage, we reduce the degree $\Delta$ to $\poly(1 / \epsilon)$ by sampling edges independently with probability $\Theta(1/(\Delta\epsilon^4))$ and then only consider the subgraph induced by nodes whose degree is approximately $\Theta(1 / \epsilon^4)$. In \Cref{lem:sampling-main}, we show that the resulting subgraph retains a large matching with high probability. In the second stage, we leverage the small-degree property of the sampled graph to find an almost perfect matching without any dependence on $n$ in the number of rounds. We now handle each stage and their proofs in separate subsections: \Cref{subsec:sampling} for the sampling stage and \Cref{subsec:constant-matching} for the matching stage.

\subsection{Sampling Stage}
\label{subsec:sampling}

In this subsection, we give our sampling stage algorithm and prove that the resulting graph retains an almost perfect matching. Algorithm~\ref{alg:sample} gives a formal description of the algorithm.

\begin{algorithm}[H]
	\SetAlgoLined
	\DontPrintSemicolon
	\KwData{A $\Delta$-regular unweighted graph $G=(V,E)$; an accuracy parameter $\epsilon \in (n^{-1/20}, 1/2)$}
	\KwResult{An unweighted graph $G'=(V', E')$ and a max degree $d$ which is $\min\{\frac{6000}{\epsilon^4}, \Delta\}$
                  with the guarantee the max degree of $G'$ is at most $d$}

        \If{$\Delta \le \frac{6000}{\epsilon^4}$}{
          Return $G' \gets G$ and max degree $d \gets \Delta$.
        }
        
        $E_1 \subseteq E \gets $ each edge in $E$ is included independently with probability $p' = \frac{3000}{\Delta \epsilon^4}$
        
        $G_1\gets (V,E_1)$

        $V_2\gets $ all vertices in $V$ with degree at most $2p'\Delta$ in $G_1$

        $G_2\gets G_1[V_2]$

        Return graph $G' \gets G_2$ and max degree $d \gets (2p'\Delta)$.
	\caption{SamplingStage($G$)}
        \label{alg:sample}
\end{algorithm}

\begin{restatable}{lemma}{samplinglemma}
\label{lem:sampling-main}
Let $OPT(H)$ denote the size of the maximum matching in graph $H$. Then, when we run Algorithm~\ref{alg:sample}, we have that $OPT(G_2) \ge (1 - 5\epsilon)OPT(G)$ with probability at least $1 - 5/n^{30}$.
\end{restatable}


We observe that if $\Delta \le \frac{6000}{\epsilon^4}$, then \Cref{lem:sampling-main} follows trivially (keeping the entire graph means we preserved the original matching). Thus, in the rest of this section, we assume that we used the sampling probability $p' = \frac{3000}{\Delta \epsilon^4} < 1$. Our plan is to show that the intermediate graph $G_1$ has an almost-perfect matching and that there are not many high-degree vertices so removing them does not significantly reduce the size of the matching.

More formally, we plan to invoke the following folklore result about almost-regular graphs admitting an almost-perfect matching. The result requires the following notation. Let $\kappa \in (0, 1/2)$ and $D \geq 1$ be fixed. We say an edge $e$ is $(\kappa, D)$-\emph{balanced} if and only if both its endpoints have degree in $((1 - \kappa)D, (1+\kappa)D)$. 

\begin{lemma}[Folklore]
\label{lem:regular-matching-size}
Let $G = (V, E)$ be an undirected, unweighted graph and $\tau_e,\tau_v, \kappa \in (0, 1/2)$ and $D\geq 1$ be fixed such that
\begin{enumerate}[noitemsep,topsep=0pt,parsep=0pt,partopsep=0pt,label=(\alph*)]
    \item At least $(1 - \tau_e)|E|$ edges are $(\kappa, D)$-balanced.
    \item At least $(1 - \tau_v)|V|$ vertices have degree in $((1-\kappa)D, (1+\kappa)D)$.
\end{enumerate}
Then $G$ has a matching of size at least $(1 - \tau_e - \tau_v - 2\kappa - \frac{1}{D+1})\frac{|V|}{2}$ consisting only of $(\kappa, D)$-balanced edges.
\end{lemma}

The proof of Lemma~\ref{lem:regular-matching-size} is provided in Appendix~\ref{app:general} for completeness. Now we just need to show that we can make the subgraph $G_1$ sampled in Algorithm~\ref{alg:sample} fit the conditions of \Cref{lem:regular-matching-size}. We expect vertices to have degree $p'\Delta$. Since we are going to need $\kappa$ to be on the order of $\epsilon$, we are aiming to have a lot of vertices with degree in the range $((1 - \epsilon)p'\Delta, (1 + \epsilon)p'\Delta)$. Unfortunately, as each edge is sampled independently with probability $p' = \Theta(\frac{1}{\Delta \epsilon^4})$, a standard Chernoff bound followed by a union bound attempting to guarantee that all vertices have degree between $((1 - \epsilon)p'\Delta, (1 + \epsilon)p'\Delta)$ results in too much failure probability. Instead, for the sake of analysis, we sample the edges in two separate phases and utilize concentration bounds for random variables with limited dependence to show that most of the vertices have approximately the correct degree. Formally, we show the following lemma.

\begin{lemma}
\label{lem:g1-deg}
With probability at least $1 - 4/n^{30}$, both of the following conclusions hold:
\begin{enumerate}[noitemsep,topsep=0pt,parsep=0pt,partopsep=0pt,label=(\alph*)]
    \item At most $(4e^{-1500/(27 \epsilon^2)})n$ vertices have degree outside of $((1 - \epsilon)p'\Delta, (1+\epsilon)p'\Delta)$ in $G_1$.
    \item At most $(12e^{-1500 / (27 \epsilon^2)})p'm$ edges in $G_1$ are \emph{not} $(\epsilon, p'\Delta)$-balanced.
\end{enumerate}
\end{lemma}

In order to prove \Cref{lem:g1-deg}, we break up the sampling stage into two phases as follows. Let $p = \min\{\frac{1000 \log n}{\Delta \epsilon^2}, 1\}$ and $q = p' / p$. For the sake of analysis, we assume that the set $E_1$ in Algorithm~\ref{alg:sample} is constructed as follows: let $E_0 \subseteq E$ be such that each edge $e$ is included in $E_0$ independently with probability $p$, and let $E_1 \subseteq E_0$ be such that each edge is chosen independently with probability $q$. Let $G_0 = (V, E_0)$. The following two propositions are simple consequences of Chernoff bounds.

\begin{proposition}
\label{prop:g0-deg}
With probability at least $1 - 1/n^{30}$, every vertex in $G_0$ has degree between $(1 - \epsilon/3)p\Delta$ and $(1 + \epsilon/3)p\Delta$.
\end{proposition}

\begin{proof}
Note that if $p = 1$, then the claim follows trivially. So suppose $p = \frac{1000 \log n}{\Delta \epsilon^2} < 1$.
Consider a vertex $v\in V$. For each neighbor $u$ of $v$ in $G$, let $X_{uv}$ denote the indicator variable of the presence of the edge $\{u,v\}$ in $G_0$. Let $X_v = \sum_{u \in N_G(v)} X_{uv}$ be the degree of $v$ in $G_0$. Note that $\mathbb{E}[X_v] = \sum_{u \in N_G(v)} \mathbb{E}[X_{uv}] = \sum_{u \in N_G(v)} p = p\Delta$. Thus, by Theorem \ref{thm:Chernoff} with $\delta\gets\epsilon/3$,

$$\mathbb{P}[|X_v - p\Delta| \ge \epsilon/3 p\Delta] \le 2e^{-\epsilon^2 p\Delta/27}  \le 2/n^{1000/27}$$

Thus, by a union bound over all $v\in V(G)$, the probability that there exists a vertex with degree outside of the desired range is at most $2n/n^{1000/27} \le 1/n^{30}$, as desired.
\end{proof}

\begin{proposition}\label{prop:g1-edges}
With probability at least $1 - 1/n^{30}$, $G_1$ has between $(1-\epsilon)p'm$ and $(1+\epsilon)p'm$ edges.
\end{proposition}

\begin{proof}
For an edge $e\in E(G)$, let $X_e = 1$ if $e\in E(G_1)$ and $X_e = 0$ otherwise. Note that $\mathbb{P}[X_e = 1] = pq = p'$ for all $e$ and that the $X_e$s are independent. Let $X = \sum_{e\in E(G)} X_e$. Since $G$ is $\Delta$-regular, $\mathbb{E}[X] = p'm = p'(\Delta n/2)$. By Theorem \ref{thm:Chernoff},

$$\mathbb{P}[|X - p'm| \ge  \epsilon p'm]\le 2e^{-(p' \Delta n/2)\epsilon^2/3} = 2e^{-500n/\epsilon^2} < 1/n^{30}$$
as desired.
\end{proof}

We are now ready to prove Lemma \ref{lem:g1-deg}.

\begin{proof}[Proof of \Cref{lem:g1-deg}]
For any $v \in V$, let $d_v$ denote the degree of $v$ in $G_0$. We say $G_0$ is \emph{good} if and only if $d_v \in ((1 - \epsilon/3)p\Delta, (1+\epsilon/3)p\Delta)$ for all $v \in V$. By \Cref{prop:g0-deg}, $\mathbb{P}[G_0 \text{ is good}] \geq 1 - 1/n^{30}$.

Let $Y_v$ be a random variable denoting the degree of $v$ in $G_1$. Since every edge in $G_0$ is included in $G_1$ independently with probability $q$, by \Cref{thm:Chernoff}, we have
\begin{align*}
  \mathbb{P}[|Y_v - q d_v| \geq (\epsilon/3)qd_v \mid G_0 \text{ is good}] \leq 2 e^{-\epsilon^2 q d_v/27} \leq 2 e^{-\epsilon^2 (1-\epsilon/3) pq\Delta / 27} \leq 2e^{-1500 / (27 \epsilon^2)}
  \end{align*}
where we used $d_v \geq (1-\epsilon/3)p\Delta)$ since $G_0$ is good in the second inequality. Unfortunately we note that this probability is not low enough to allow us to union bound over all vertices in $V$. However, as $G_0$ has bounded maximum degree (when $G_0$ is good), the random variables $\{Y_v\}$ exhibit limited number of dependencies. We thus utilize concentration bounds for sums of dependent random variables to show that most vertices have the appropriate degree.

Formally, let $Z_v$ be an indicator variable denoting the event $|Y_v - q d_v| \geq (\epsilon / 3)q d_v$. For convenience, let $\delta_0 := 2e^{-1500 / (27 \epsilon^2)}$. Then the previous inequality is equivalent to $\mathbb{P}[Z_v \mid G_0 \text{ is good}] \leq \delta_0$. Consider the collection of random variables $\mathcal{Z} = \{Z_v\}$. Given a graph $G_0$, variables $Z_u$ and $Z_v$ are dependent only when they are adjacent in $G_0$. Thus, if $G_0$ is good, then the dependency graph of $\mathcal{Z}$ has maximum degree at most $(1 + \epsilon/3)p\Delta$; which implies that it has chromatic number at most $1 + (1 + \epsilon/3)p\Delta \leq 4000 \log n / \epsilon^2$. Thus, applying \Cref{thm:CorrelatedChernoff} with $\lambda \gets n\delta_0$, we have
\begin{align*}
    &\mathbb{P}\left[\sum_{v \in V} Z_v - \sum_{v \in V} \mathbb{E}[Z_v \;\middle|\; G_0 \text{ is good}] > n\delta_0 \mid G_0 \text{ is good} \right] \le \exp\left(-\frac{2n\delta_0^2}{\chi(\mathcal{Z})}\right)\\
    &\mathbb{P}\left[\sum_{v \in V} Z_v > 2n\delta_0 \;\middle|\; G_0 \text{ is good} \right] \le \exp\left(-\frac{2n\delta_0^2}{\chi(\mathcal{Z})}\right)
    \le \exp\left(-\frac{2n \delta_0^2 \epsilon^2}{4000 \log n}\right) \leq 1/n^{100}
\end{align*}

In particular, this implies that the probability that the number of vertices whose degree in $G_1$ is outside the range $\left((1-\epsilon/3)^2 pq\Delta, (1+\epsilon/3)^2pq\Delta\right) \subset \left((1-\epsilon) p'\Delta, (1+\epsilon)p'\Delta\right)$ exceeds $2 n\delta_0$ is at most $\frac{1}{n^{100}}$ whenever $G_0$ is good. Let $B$ denote the bad event that at least $2n\delta_0$ nodes in $G_1$ have degree outside the range $\left((1-\epsilon) p'\Delta, (1+\epsilon)p'\Delta\right)$. Then we have $\mathbb{P}[B \mid G_0 \text{ is good}] \leq 1/n^{100}$. Overall, without the conditioning, we have
$\mathbb{P}[B] \leq \mathbb{P}[B \mid G_0 \text{ is good}] + \mathbb{P}[G_0 \text{ is not good}] \leq 1/n^{100} + 1/n^{30} \leq 2/n^{30}$. This completes the proof of the first statement in the lemma.

For the second part of the lemma, consider an arbitrary edge $e = \{u, v\} \in E(G_0)$ and let $W_e$ be an indicator variable for the event that $e \in G_1$ \emph{and} at least one of its end points $u$ or $v$ have their degree outside the range $((1 - \epsilon)p'\Delta, (1 + \epsilon)p'\Delta)$. Once again, let us condition on the event that $G_0$ is good. We note that $\mathbb{P}[W_e \mid G_0 \text{ is good}] \leq \mathbb{P}[e \in E(G_1) \text{ and } \max\{Z_u, Z_v\} = 1 \mid G_0 \text{ is good}]$. Further, the dependency graph of the collection of random variables $\mathcal{W} = \{W_e\}_{e \in E(G_0)}$ has maximum degree at most $(2(1 + \epsilon/3)p\Delta)^2$ since $W_e$ only depends on edges in the 2-neighborhood of edge $e$. Thus the chromatic number, $\chi(\mathcal{W}) \leq 1 + (2(1 + \epsilon/3)p\Delta)^2 \leq 16(10^6)(\log^2 n)/\epsilon^4$. Once again, applying \Cref{thm:CorrelatedChernoff} with $\lambda \gets \delta_0 q |E(G_0)|$, we have
\begin{align*}
    \mathbb{P}\left[\sum_{e\in E(G_0)} W_e > \delta_0 q|E(G_0)| + \sum_{e\in E(G_0)} \mathbb{E}[W_e | G_0\text{ is good}] \;\middle|\; G_0 \text{ is good}\right]
      &\le e^{-2|E(G_0)|q^2\delta_0^2\epsilon^4/(16(10)^6(\log^2 n))} \\
      &\le 1/n^{100}
\end{align*}
where we used $|E(G_0)| \geq n$ in the last inequality. Finally, we note that 
\begin{align*}
    \mathbb{E}[W_e | G_0 \text{ is good}] &\leq \mathbb{P}[\max(Z_u,Z_v) = 1 | e\in E(G_1), G_0\text{ is good}] \cdot \mathbb{P}[e\in E(G_1)| G_0\text{ is good}]\\
    &\le 2\mathbb{P}[Z_u = 1 | e\in E(G_1), G_0\text{ is good}] \cdot \mathbb{P}[e\in E(G_1)| G_0\text{ is good}] \le 2\delta_0q
\end{align*}
Substituting into the inequality above, we get
\begin{align*}
    \mathbb{P}\left[\sum_{e\in E(G_0)} W_e > 3\delta_0 q|E(G_0)| \;\middle|\; G_0 \text{ is good}\right] &\le 1/n^{100}
    \intertext{But, when $G_0$ is good, $|E(G_0| \leq (1 + \epsilon/3)p \Delta n /2 < 2p\Delta n/2$, so}
    \mathbb{P}\left[\sum_{e\in E(G_0)} W_e > 6\delta_0 pq \Delta n/2 \;\middle|\; G_0 \text{ is good}\right] &\le 1/n^{100}
\end{align*}
Overall, without the conditioning we have,
\begin{align*}
    \mathbb{P}\left[\sum_{e\in E(G_0)} W_e > 6\delta_0 pq \Delta n/2 \right]
      &\le \mathbb{P}\left[\sum_{e\in E(G_0)} W_e > 6\delta_0 pq \Delta n/2 \;\middle|\; G_0 \text{ is good}\right] + \mathbb{P}[G_0 \text{ is not good}]\\
      &\le \frac{1}{n^{100}} + \frac{1}{n^{30}} \leq \frac{2}{n^{30}}
\end{align*}
The lemma now follows from a union bound over the two statements.
\end{proof}

\Cref{lem:sampling-main} now follows directly from \Cref{lem:regular-matching-size} and \Cref{lem:g1-deg}.

\begin{proof}[Proof of \Cref{lem:sampling-main}]
If $\Delta \le \frac{6000}{\epsilon^4}$, then we have $G_2 = G$ and the lemma follows trivially. Thus, we assume that $\Delta > \frac{6000}{\epsilon^4}$. By \Cref{prop:g1-edges}, $|E(G_1)| \geq (1 - \epsilon)p'm$ with probability at least $1 - 1/n^{30}$. Thus, by the second statement of \Cref{lem:g1-deg}, the number of edges in $G_1$ that are not $(\epsilon, p'\Delta)$ balanced is at most $(12 e^{-1500/27\epsilon^2})/(1 - \epsilon)\cdot |E(G_1)| \leq 24 e^{-1500/27\epsilon^2} |E(G_1)| < \epsilon |E(G_1)|$. At the same time, by the first statement of \Cref{lem:g1-deg}, at most $(4 e^{-1500/(27\epsilon^2)})n \le \epsilon|V(G_1)|$ vertices of $G_1$ have degree not in $((1-\epsilon)p'\Delta, (1+\epsilon)p' \Delta)$.

Substituting $\tau_v = \tau_e = \kappa = \epsilon$ in \Cref{lem:regular-matching-size}, we get that restricting to just the nodes with degrees in $((1 - \epsilon)p'\Delta, (1 + \epsilon)p'\Delta)$ must have a matching of size at least $(1 - 4\epsilon - \frac{1}{D+1}) \frac{n}{2} \ge (1 - 5\epsilon) \frac{n}{2}$ where the last inequality used $\Delta > \frac{6000}{\epsilon^4}$. The probability bound follows from a union bound over the two lemmas. $G_2$ includes all these nodes because $\epsilon < 1/2$.
\end{proof}

\subsection{Matching Stage}
\label{subsec:constant-matching}

In this subsection, we give our matching stage algorithm which accepts a (low-degree) possibly-non-regular graph and returns an almost-perfect matching in a number of rounds which depends only on the maximum degree $d$ and the desired accuracy $\epsilon > 0$. Our main result for this subsection is the following.

\begin{restatable}{lemma}{warmupalgo}
\label{lem:constant-matching}
For $\epsilon > n^{-1/20}$, there is an $O(\epsilon^{-5} \log d)$-round randomized LOCAL algorithm where each vertex knows $d$ and $\epsilon$ that returns a matching of size at least $(OPT - \epsilon n)$ on $n$-vertex graphs with maximum degree $d$ with probability at least $1 - 1/n^{30}$.
Furthermore, this algorithm is a CONGEST algorithm if $\epsilon$ and $d$ are constant. 
\end{restatable}

To prove \Cref{lem:constant-matching}, we give an algorithm that improves a matching over many iterations; in each iteration, the algorithm attempts to find a large set of disjoint augmenting paths for the current matching. The algorithm finds these paths by constructing a hypergraph whose hyperedges each represent an augmenting path for the current matching. The algorithm finds an approximately maximum fractional matching in this hypergraph and then rounds that matching via independent random sampling and removing collisions. When the current matching is far from optimality, \Cref{prop:augment} certifies that this procedure actually finds a large matching. The algorithm is formally given by Algorithm~\ref{alg:constant-match}, and it relies the following fractional matching algorithm for hypergraphs.

\begin{theorem}[Theorem 4.8 of \cite{BenBasatEKS23} with $\alpha = 2$ and using the $\delta(e)$s computed in the algorithm]\label{thm:beks19}
    In any $f$-bounded hypergraph $G = (V,E)$ with $\epsilon\in (0,1)$ with maximum degree $\Delta$, there is an $O(\log \Delta + f\log (f/\epsilon))$-round deterministic CONGEST algorithm for computing an $(f+\epsilon)$-approximate fractional hypergraph matching in $G$.
\end{theorem}

The remainder of this subsection is a proof of \Cref{lem:constant-matching}. We begin by verifying the following property of $\mathcal{P}_i$:

\begin{proposition}\label{prop:disjointness}
$\mathcal{P}_i$ is a collection of vertex-disjoint $M_{i-1}$-augmenting paths.
\end{proposition}

\begin{proof}
By definition of $E_i$, $\mathcal{P}_i$ is a collection of $M_{i-1}$-augmenting paths, so it suffices to check that they are vertex-disjoint. Suppose, for the sake of contradiction, that there is a vertex $v\in V$ for which there exist $P_0,P_1\in \mathcal{P}_i$ with $P_0\ne P_1$ for which $v\in P_0$ and $v\in P_1$. By part (a) of the definition of $\mathcal{P}_i$, there exist $v_0\in P_0$ and $v_1\in P_1$ for which $Y_{v_0} = P_0$ and $Y_{v_1} = P_1$. By part (b) of the membership of $P_0$, since $P_0\cap Y_{v_1}\ne \emptyset$, $Y_{v_1} = P_0$, a contradiction to the fact that $P_0 \ne P_1$. Thus, the sets in $\mathcal{P}_i$ are vertex-disjoint, as desired.
\end{proof}

Note, first, that $|M_i| \ge |M_{i-1}|$, since all paths in $\mathcal{P}_i$ are augmenting paths. Let $OPT$ denote the size of a maximum matching in $G$. If $|M_T| \ge OPT - \epsilon n$, then we are done, so assume for the sake of contradiction that $|M_T| < OPT - \epsilon n$. This means that $|M_i| \le OPT - \epsilon n$ for all $i\in [T]$. We use this to show the following:

\begin{proposition}\label{prop:ei-size}
For any $i\in \{1,2,\hdots,T\}$, if $|M_{i-1}|\le OPT - \epsilon n$, then $\sum_{P\in E_i} x_i(P) \ge \frac{\epsilon n}{4(k + 1)}$.
\end{proposition}

\begin{proof}
By Proposition \ref{prop:augment} applied to $M\gets M_{i-1}$ and $\ell\gets (k-1)/2$, there exists a collection $\mathcal{P}$ of $M_{i-1}$-augmenting paths with length at most $k$ in $G$ for which

\begin{align*}
|\mathcal{P}| &\ge \frac{1}{2}(OPT(1 - 2/(k-1)) - |M_{i-1}|) \\
&\ge \frac{1}{2}(OPT(1 - \epsilon/2) - (OPT - \epsilon n))
= \frac{\epsilon}{2}(n - OPT/2)
\ge \frac{\epsilon n}{4}
\end{align*}

When the sets in $\mathcal{P}$ are viewed as hyperedges in $H_i$, $\mathcal{P}$ is a hypergraph matching thanks to the vertex disjointness of the sets. Thus, $H_i$ has a hypergraph matching with at least $\frac{\epsilon n}{4}$ hyperedges. Since $H_i$ is a $k$-bounded hypergraph, Theorem \ref{thm:beks19} implies that the total size of the fractional matching $x_i$ is at least $|\mathcal{P}|/(k+1/2) \ge \frac{\epsilon n}{4(k+1)}$ as desired.
\end{proof}

\begin{algorithm}[H]
	\SetAlgoLined
	\DontPrintSemicolon
	\KwData{An unweighted graph $G=(V,E)$; an accuracy parameter $\epsilon \in (n^{-1/20}, 1/2)$}
	\KwResult{A matching $M$ in $G$}
	
	$M_0\gets\emptyset$

        $k\gets 4/\epsilon + 1$

        $T\gets 10^4 / \epsilon^4 + 1$

        $\tau\gets 1/(4k^2)$

        \For{$i\in \{1,2,\hdots,T\}$}{
            $E_i\gets $ the set of all $M_{i-1}$-augmenting paths $P$ with $|P|\le k$

            $H_i\gets (V,E_i)$, where each $P\in E_i$ yields a hyperedge between the vertices in $P$

            $x_i\gets $ the fractional matching in $H_i$ given by Theorem \ref{thm:beks19} with $\epsilon = 1/2$

            \tcc{$\mathcal{P}_i$ will contain disjoint $M_{i-1}$-augmenting paths we are choosing}

            $\mathcal{P}_i\gets \emptyset$

            \For{each vertex $v\in V$ independently}{
                $\mathcal{X}_v\gets \{P \text{ } \forall P\in E_i \text{ for which } v\in P\}\cup \{\star\}$
                
                $Y_v\gets $ a randomly chosen member of $\mathcal{X}_v$, with $P\in \mathcal{X}_v$ chosen with probability $\tau x_i(P)$, and $\star$ chosen with probability $1 - \tau \sum_{Q\ne\star \in \mathcal{X}_v} x_i(Q)$
            }

            \For{each hyperedge $P\in E_i$}{
                Add $P$ to $\mathcal{P}_i$ if (a) there exists a $v\in P$ for which $Y_v = P$ and (b) for any $u\in V$ for which $Y_u\cap P\neq\emptyset$, $Y_u = P$
            }

            $M_i\gets$ augmentation of $M_{i-1}$ by $\mathcal{P}_i$
        }
	
	\Return $M_T$
	
	\caption{ConstantMatch($G$)}\label{alg:constant-match}
\end{algorithm}

~\\We now use the lower from Proposition~\ref{prop:ei-size} to show that the rounding part of Algorithm \ref{alg:constant-match} finds a large collection of augmenting paths. To analyze the sampling steps, we use McDiarmid's Inequality:

\begin{theorem}[\cite{McDiarmid89}]\label{thm:mcdiarmid}
Let $\mathcal{X}_1, \mathcal{X}_2, \hdots, \mathcal{X}_n$ be sets, $c_1,c_2,\hdots,c_n\in \mathbb{R}$, and $f:\mathcal{X}_1\times\mathcal{X}_2\times\hdots\times\mathcal{X}_n\rightarrow\mathbb{R}$ be a function with the property that, for any $i\in [n]$, $x_1\in \mathcal{X}_1,x_2\in\mathcal{X}_2,\hdots,x_n\in\mathcal{X}_n$, and $x_i'\in \mathcal{X}_i$,

$$|f(x_1,\hdots,x_{i-1},x_i,x_{i+1},\hdots,x_n) - f(x_1,\hdots,x_{i-1},x_i',x_{i+1},\hdots,x_n)|\le c_i$$

Then, for any $\delta > 0$,

$$\mathbb{P}[|f(X_1,X_2,\hdots,X_n) - \mathbb{E}[f(X_1,X_2,\hdots,X_n)]| \ge \delta] \le 2\exp\left(-\frac{2\delta^2}{\sum_{i=1}^n c_i^2}\right)$$
\end{theorem}

Fix an $i\in \{1,2,\hdots,T\}$ for the rest of this section. Define the function $f:\prod_{v\in V} \mathcal{X}_v\rightarrow \mathbb{R}$ to be $f(Y) := |\mathcal{P}_i|$, where $Y$ is the $n$-tuple of $Y_v$s for all $v\in V$. This choice of function is inspired by Lemma 5.1 of \cite{GGKMR18}. We now prepare to use Theorem \ref{thm:mcdiarmid} by showing the following two results:

\begin{proposition}\label{prop:hypermatch-lipschitz}
Let $Y$ and $Y'$ be two different tuples indexed by $V$ for which there exists exactly one $v\in V$ for which $Y_v\ne Y_v'$. Then $|f(Y) - f(Y')| \le 2k \le 10/\epsilon$.
\end{proposition}

\begin{proof}
Let $\mathcal{P}_i$ and $\mathcal{P}_i'$ be the sets resulting from $Y$ and $Y'$ respectively. It suffices to show that $||\mathcal{P}_i| - |\mathcal{P}_i'|| \le k$ when $Y_v = \star$ and $Y_v' = P \ne \star$, as all remaining cases can be covered by the triangle inequality. Let $p_1,p_2,\hdots,p_{\ell}$ with $\ell\le k$ be the members of $P$. By part (b) of the definition, for all $j\in \{1,2,\hdots,\ell\}$, there exists at most one $P_j\in \mathcal{P}_i\cup \{\star\}$ for which $p_j\in P_j$. Every other set in $\mathcal{P}_i$ does not intersect $P$, so $\mathcal{P}_i\subseteq \mathcal{P}_i'\cup \{P_1,P_2,\hdots,P_{\ell}\}$. Furthermore, $\mathcal{P}_i'\subseteq \mathcal{P}_i\cup \{P\}$. Therefore,

$$|\mathcal{P}_i| - \ell \le |\mathcal{P}_i'|\le |\mathcal{P}_i| + 1$$

as desired (since $\ell\le k$).
\end{proof}

\begin{proposition}\label{prop:hypermatch-expectation}
If $|M_{i-1}|\le OPT - \epsilon n$, then

$$\mathbb{E}_Y[f(Y)] \ge \frac{(1 - 2k^2\tau)\tau \epsilon n}{4(k+1)} \ge \frac{\epsilon^4 n}{5000}$$
\end{proposition}

\begin{proof}
We start by lower bounding the probability that any specific $P = \{p_1,p_2,\hdots,p_{\ell}\}\in E_i$ is added to $\mathcal{P}_i$. First, note that

\begin{align*}
&\mathbb{P}_Y[\text{there exists } j\in \{1,2,\hdots,\ell\} \text{ for which both } Y_{p_j} = P \text{ and } Y_{p_{j'}} = \star \text{ for all } j'\ne j]\\
&= \sum_{j=1}^{\ell} \tau x_i(P) \prod_{j'\ne j}\left(1 - \tau\sum_{Q\ne\star\in \mathcal{X}_{p_{j'}}} x_i(Q)\right)\\
&\ge \sum_{j=1}^{\ell}\tau x_i(P)(1 - \tau)^{\ell-1}\\
&\ge \tau(1-\tau)^k x_i(P)
\end{align*}

since the $Y_{p_j}$s are independent. For any $u,w\in V$, let $E_i(u,w)$ denote the set of all hyperedges $Q\in E_i$ for which $u,w\in Q$. Note that for any $j\in \{1,2,\hdots,\ell\}$

\begin{align*}
\mathbb{P}_Y[\text{for all } u\in V\setminus P, p_j\notin Y_u] &= \prod_{u\in V\setminus P} \mathbb{P}_{Y_u}[p_j\notin Y_u]\\
&= \prod_{u\in V\setminus P}\left(1 - \tau\sum_{Q\in E_i(u,p_j)} x_i(Q)\right)\\
&\ge 1 - \tau\sum_{u\in V}\sum_{Q\in E_i(u,p_j)} x_i(Q)\\
&= 1 - \tau\sum_{Q\in E_i: p_j\in Q}\sum_{u\ne p_j\in Q} x_i(Q)\\
&\ge 1 - k\tau\sum_{Q\in E_i: p_j\in Q} x_i(Q)\\
&\ge 1 - k\tau
\end{align*}

By a union bound,

$$\mathbb{P}_Y[\text{for all } u\in V\setminus P, P\cap Y_u = \emptyset] \ge 1 - k^2\tau$$

The first event only depends on $V\setminus P$, while the second only depends on $P$. Thus, by independence,

\begin{align*}
&\mathbb{P}_Y[\text{there exists } j\in \{1,2,\hdots,\ell\} \text{ for which both } Y_{p_j} = P \text{ and } Y_{p_{j'}} = \star \text{ and for all $u\in V\setminus P$}, P\cap Y_u = \emptyset]\\
&\qquad\ge (1 - k^2\tau)\tau(1 - \tau)^kx_i(P)
\end{align*}

Such $P$ are added to $\mathcal{P}_i$, as $Y_{p_j} = P$ (satisfying condition (a)), $Y_u = \star$ for all $u\in P$ with $u\ne v_j$, and $Y_u\cap P = \emptyset$ for all $u\in V\setminus P$, so condition (b) is never triggered. Thus,

$$\mathbb{P}_Y[P\in \mathcal{P}_i] \ge (1 - k^2\tau)\tau(1-\tau)^kx_i(P) \ge (1 - 2k^2\tau)\tau x_i(P)$$

and

\begin{align*}
\mathbb{E}_Y[f(Y)] &= \sum_{P\in E_i} \mathbb{P}_Y[P\in \mathcal{P}_i]\\
&\ge (1 - 2k^2\tau)\tau \sum_{P\in E_i} x_i(P)\\
&\ge \frac{(1 - 2k^2\tau)\tau \epsilon n}{4(k+1)}
\end{align*}

by Proposition \ref{prop:ei-size}, as desired.
\end{proof}

We are now ready to use McDiarmid's Inequality to lower bound the improvement in each iteration:

\begin{proposition}\label{prop:hypermatch-result}
$|\mathcal{P}_i| \ge \frac{\epsilon^4 n}{10000}$ with probability at least $1 - 2\exp\left(-\epsilon^6n/10^{10}\right)$
\end{proposition}

\begin{proof}
By Theorem \ref{thm:mcdiarmid}, Proposition \ref{prop:hypermatch-lipschitz}, and Proposition \ref{prop:hypermatch-expectation}, with $\delta\gets \frac{\epsilon^4n}{10000}$,

$$\mathbb{P}_Y[f(Y) \le \frac{\epsilon^4 n}{10000}] \le 2\exp\left(-\frac{2\delta^2}{(10/\epsilon)^2n}\right) \le 2\exp\left(-\epsilon^6n/10^{10}\right)$$

as desired.
\end{proof}

\begin{proof}[Proof of Lemma \ref{lem:constant-matching}]
Let $T'$ be the minimum value for which $|M_i| \ge OPT - \epsilon n$, or $T' = T+1$ if no such $T'$ exists. $T'$ is a random variable. We now upper bound the probability that  $T' = T+1$. Since $|M_i| = |\mathcal{P}_i| + |M_{i-1}|$ for all $i$, $|M_T| \le n$, and $|M_0| = 0$, there must exist an $i$ for which $|\mathcal{P}_i|\le \frac{n}{T} < \epsilon^4 n/10^4$. Thus,

\begin{align*}
\mathbb{P}[T' = T+1] &\le \sum_{i=1}^T\mathbb{P}\left[T' = T+1 \text{ and } |\mathcal{P}_i| < \frac{\epsilon^4 n}{10^4}\right]\\
&\le \sum_{i=1}^T\mathbb{P}\left[|M_{i-1}| < OPT - \epsilon n \text{ and } |\mathcal{P}_i| < \frac{\epsilon^4 n}{10^4}\right]\\
&\le \sum_{i=1}^T\mathbb{P}\left[|\mathcal{P}_i| < \frac{\epsilon^4 n}{10^4} \;\middle|\; |M_{i-1}| < OPT - \epsilon n \right]\\
&\le 2T\exp(-\epsilon^6n/10^{10})\\
&\le \frac{1}{n^{30}}
\end{align*}

where the second to last inequality follows from Proposition \ref{prop:hypermatch-result}. Thus, with the desired probability, $T' < T + 1$, in which case the algorithm finds a matching with the desired size. The maximum degree of $H_i$ is at most $d^k$, so each iteration of the for loop takes $\log (d^k) + k\log(k/\epsilon) = O((\log (d/\epsilon))/\epsilon)$, as all other operations takes a constant number of rounds in LOCAL, and a constant number of rounds in CONGEST when $d$ and $\epsilon$ are constant. Thus, multiplying by $T$ gives the desired runtime.
\end{proof}

\subsection{The End-to-End Algorithm}

In this subsection, we are finally ready to show that \Cref{thm:warmup} follows from \Cref{lem:sampling-main} and \Cref{lem:constant-matching}.

\begin{proof}[Proof of \Cref{thm:warmup}]
    We show that SamplingStage($G$) (Algorithm~\ref{alg:sample}) followed by Algorithm~\ref{alg:constant-match} returns the desired output in the desired runtime (for $\epsilon > n^{-1/20}$). By definition, the maximum degree of $G_2$ is at most $2p'\Delta = 6000/\epsilon^4$. We use this fact to bound both the runtime and the approximation error:

    \emph{Runtime:} \Cref{lem:constant-matching} is applied to the $G_2$ produced by Algorithm \ref{alg:sample}, so $d = 6000/\epsilon^4$ in this case and $O(\epsilon^{-5} \log d) = O(\epsilon^{-5}\log(1/\epsilon))$ as desired. Since the runtime of Algorithm~\ref{alg:sample} is a constant number of rounds, the overall runtime is still just $O(\epsilon^{-5} \log(1/\epsilon))$. Note that the algorithm for \Cref{lem:constant-matching} can be used, as each vertex in $G_2$ knows both $d$ (which only depends on the original regular graph's degree $\Delta$ and $\epsilon$) and $\epsilon$.

    \emph{Approximation:} Let $OPT$ denote the size of the maximum matching in the input graph $G$. By \Cref{lem:sampling-main}, we have that $OPT(G_2) \ge (1 - 5\epsilon)OPT$ with probability at least $1 - \frac{5}{n^{30}}$. We note that since $G$ is $\Delta$-regular, we have $OPT \ge (1 - \frac{1}{\Delta+1})\frac{n}{2} \ge \frac{n}{3}$.

    Therefore, applying Lemma \ref{lem:constant-matching} results in a matching with size at least $(1 - 5\epsilon)OPT - \epsilon n \geq (1 - 5\epsilon)OPT - 3 \epsilon OPT = (1 - 8\epsilon)OPT$ with probability at least $1 - 5/n^{30} - 1/n^{30} > 1 - 6/n^{30}$. The theorem now follows by replacing $\epsilon$ with $\epsilon/8$, which only changes run time by a constant factor.
\end{proof}





\section{Lower Bounds}
\label{sec:lower}

In this section, we complement our algorithms by giving lower bounds for bipartite regular graphs of low degree. Our bounds are based on a lower bound construction of Ben-Basat, Kawarabayashi, and Schwartzman (BKS), who proved $\Omega(1 / \epsilon)$ lower bounds for a variety of problems in the LOCAL model, including maximum matching~\cite{BenbasatKS19}. The high-level idea from that construction is as follows. First, design a symmetric subgraph gadget (for them, a path sufficed) with radius roughly equal to the number of rounds so that no matter how the gadget's boundary nodes are hooked up to the surrounding graph, the innermost node(s) can never hope to learn about anything outside of the gadget. Next, due to the symmetrical design of the gadget, we can choose to randomly insert it either forwards or backwards. No matter how an algorithm (even randomized) originally chooses to match that innermost gadget node(s), this equalizes the probability of matching it with its forward neighbor and backwards neighbor. Finally, there is now a constant probability that a constant number of these innermost matching edges are now locked in a way that induces a single error, meaning we expect there to be roughly one error per (gadget size) number of nodes. Against Monte Carlo algorithms, there is an additional step where the independence of our forwards/backwards insertion decisions can be used to execute a Chernoff bound that makes it exponentially unlikely that the number of mistakes can be significantly lower than expected.

The efficiency of this argument hinges on the number of nodes in the gadget. Let $r$ be the number of rounds available to the algorithm. In the original construction \cite{BenbasatKS19}, the path gadgets had $O(r)$ nodes, so roughly $\epsilon^{-1}$ rounds winds up translating into a $(1 + \epsilon)$ multiplicative error. Our main technical contribution is showing that we can match this efficiency when constructing a cycle (a bipartite regular graph of degree two) and that as the degree $\Delta$ scales up, we can design gadgets with only $O(\Delta r)$ nodes. We make the following observations about this dependence on $\Delta$: (i) any lower bound, whether it uses this particular gadget framework or not, must eventually worsen as $\Delta$ increases because our upper bounds establish that higher degree regular graphs are easy and (ii) for this particular gadget framework, $\Theta(\Delta r)$ is asymptotically the best possible gadget size dependence on $\Delta$ and $r$. The latter can be observed by considering a breadth-first search (BFS) tree rooted at one of the innermost nodes; this tree must have $r$ levels before we ever reach any node outside the gadget. There must be a path from the root to a node at level $r$, otherwise the root will not be connected to the outside graph; this path is a witness to the fact that there is at least one node at every level. Next, observe that the node at level $i \in [1, r]$ must have $\Delta$ unique neighbors, which must be either in levels $i - 1$, $i$, or $i + 1$ (otherwise the BFS tree is incorrect); this holds even for non-bipartite graphs and for bipartite graphs they can only be in levels $i - 1$ or $i + 1$. In any case, this means we can partition the $r$ levels into groups of three (if non-bipartite) or two (if bipartite) such that each group must have at least $\Delta$ nodes. Hence there are $\Omega(\Delta r)$ nodes in such a gadget.

Our (asymptotically optimal) gadgets formally yield the following lower bounds against deterministic and Monte Carlo algorithms.

\begin{restatable}{theorem}{lowergeneral}\label{thm:lower-general}
  For any degree $\Delta \ge 2$ and error $\epsilon \in (0, \frac{1}{160\Delta + 32})$, any deterministic LOCAL algorithm that computes a $(1 + \epsilon)$-multiplicative approximation for \maximummatching{} on bipartite $\Delta$-regular graphs with at least $n \ge \Omega(\Delta^{-1}\epsilon^{-1})$ nodes requires $\Omega(\Delta^{-1}\epsilon^{-1})$ rounds.

  For any degree $\Delta \ge 2$, error $\epsilon \in (0, \frac{1}{320\Delta + 64})$, and failure probability $\delta \in (0, 1)$, any Monte Carlo LOCAL algorithm that computes a $(1 + \epsilon)$-multiplicative approximation with probability at least $1 - \delta$ for \maximummatching{} on bipartite $\Delta$-regular graphs with at least $n \ge \Omega(\Delta^{-1}\epsilon^{-1} \ln (1 - \delta)^{-1})$ nodes requires $\Omega(\Delta^{-1}\epsilon^{-1})$ rounds.
\end{restatable}

The complexity of our proof and the resulting constant factors heavily depend on the degree, so we will split our analysis and gadget design into the following cases, from easiest to hardest: degree two (\Cref{subsec:lower-degree-two}), even degree (\Cref{subsec:lower-degree-even}), and finally general degree (\Cref{subsec:lower-degree-general}). The first two subsections prove specializations of \Cref{thm:lower-general} with better bounds.

\subsection{Bipartite Regular Graphs of Degree Two (Cycles)}
\label{subsec:lower-degree-two}

As a warm-up, we consider the (easy difficulty) degree two case. This case is also covered by \Cref{subsec:lower-degree-even}, which handles even degree, but the construction here will be simpler, be useful for building intuition (especially for readers less familiar with the original BKS construction), and yield a better constant factors. We will prove the following subresult in this subsection.

\begin{theorem}\label{thm:lower-degree-two}
  For any error $\epsilon \in (0, \frac{1}{16})$, any deterministic LOCAL algorithm that computes a $(1 + \epsilon)$-multiplicative approximation for \maximummatching{} on bipartite $2$-regular graphs with at least $n \ge \Omega(\epsilon^{-1})$ nodes requires $\Omega(\epsilon^{-1})$ rounds.

  For any error $\epsilon \in (0, \frac{1}{32})$, and failure probability $\delta \in (0, 1)$, any Monte Carlo LOCAL algorithm that computes a $(1 + \epsilon)$-multiplicative approximation with probability at least $1 - \delta$ for \maximummatching{} on bipartite $2$-regular graphs with at least $n \ge \Omega(\epsilon^{-1} \ln (1 - \delta)^{-1})$ nodes requires $\Omega(\epsilon^{-1})$ rounds.
\end{theorem}

\begin{proof}
The overall proof plan is to use the BKS path gadget (with slightly different notation), but embed it a little differently to get an even-length cycle. We give all the proof details here so it is not necessary for a reader to be familiar with the BKS proof.

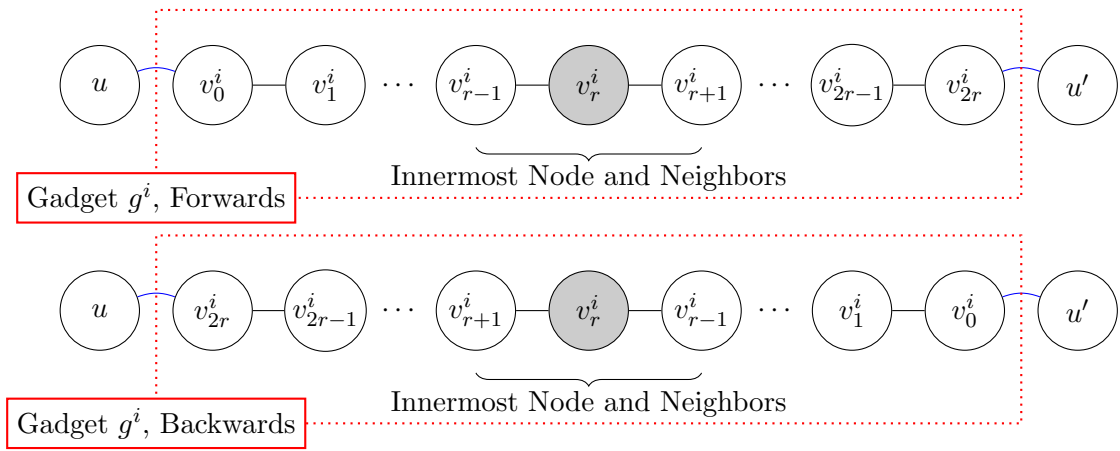
\begin{figure}
\centering
\begin{tikzpicture}[%
  auto,
  scale=1.0,
  hollownode/.style={
    circle,
    draw=black,
    minimum size=30pt,
    inner sep=1pt,
  },
  shadednode/.style={
    circle,
    draw=black,
    fill=black!20,
    minimum size=30pt,
    inner sep=1pt,
  },
  redblock/.style={
    rectangle,
    solid,
    draw=red,
    fill=white,
    text=black,
    align=center,
  },
  ]
  \node[hollownode] (v-) at (-1.5, 0) {$u$};
  \node[hollownode] (v0) at ( 0, 0) {$v^i_0$};
  \node[hollownode] (v1) at ( 1.5, 0) {$v^i_1$};
  \node             (v2) at ( 2.5, 0) {$\cdots$};
  \node[hollownode] (v3) at ( 3.5, 0) {$v^i_{r-1}$};
  \node[shadednode] (v4) at ( 5, 0) {$v^i_{r}$};
  \node[hollownode] (v5) at ( 6.5, 0) {$v^i_{r+1}$};
  \node             (v6) at ( 7.5, 0) {$\cdots$};
  \node[hollownode] (v7) at ( 8.5, 0) {$v^i_{2r-1}$};
  \node[hollownode] (v8) at (10, 0) {$v^i_{2r}$};
  \node[hollownode] (v9) at (11.5, 0) {$u'$};

  \draw[blue] (v-) to[bend left=20] (v0);
  \draw (v0) -- (v1);
  \draw (v3) -- (v4) -- (v5);
  \draw (v7) -- (v8);
  \draw[blue] (v8) to[bend left=20] (v9);

  \draw [decorate,decoration={brace,amplitude=5pt,mirror,raise=5ex}]
  (3.5,0) -- (6.5,0) node[midway,yshift=-4em]{Innermost Node and Neighbors};

  \draw[red, thick, dotted] (-0.75, -1.5) node[redblock] {Gadget $\gadget{i}$, Forwards} -- (-0.75, 1) -- (10.75, 1) -- (10.75, -1.5) -- cycle;

  \node[hollownode] (x-) at (-1.5, -3) {$u$};
  \node[hollownode] (x0) at ( 0, -3) {$v^i_{2r}$};
  \node[hollownode] (x1) at ( 1.5, -3) {$v^i_{2r-1}$};
  \node             (x2) at ( 2.5, -3) {$\cdots$};
  \node[hollownode] (x3) at ( 3.5, -3) {$v^i_{r+1}$};
  \node[shadednode] (x4) at ( 5, -3) {$v^i_{r}$};
  \node[hollownode] (x5) at ( 6.5, -3) {$v^i_{r-1}$};
  \node             (x6) at ( 7.5, -3) {$\cdots$};
  \node[hollownode] (x7) at ( 8.5, -3) {$v^i_1$};
  \node[hollownode] (x8) at (10, -3) {$v^i_0$};
  \node[hollownode] (x9) at (11.5, -3) {$u'$};

  \draw[blue] (x-) to[bend left=20] (x0);
  \draw (x0) -- (x1);
  \draw (x3) -- (x4) -- (x5);
  \draw (x7) -- (x8);
  \draw[blue] (x8) to[bend left=20] (x9);

  \draw [decorate,decoration={brace,amplitude=5pt,mirror,raise=5ex}]
  (3.5,-3) -- (6.5,-3) node[midway,yshift=-4em]{Innermost Node and Neighbors};

  \draw[red, thick, dotted] (-0.75, -4.5) node[redblock] {Gadget $\gadget{i}$, Backwards} -- (-0.75, -2) -- (10.75, -2) -- (10.75, -4.5) -- cycle;
\end{tikzpicture}
\caption{(Degree Two Case) These are the two ways to our (path) gadget may be inserted between $u$ and $u'$. In $r$ rounds, the shaded innermost node $v^i_r$ can only learn about the internal contents of the gadget and nothing else in the surrounding graph, no matter whether the gadget was inserted forwards or backwards. The gadget is connected to two anchor nodes via the blue curved edges.}
\label{fig:tikz-path}
\end{figure}

We want our gadget to help guard against LOCAL algorithms which run for some $r$ rounds (which will be one round fewer than the $\Omega(\epsilon^{-1})$ rounds we require in the theorem statement). \Cref{fig:tikz-path} depicts the gadget design, which we now formally describe. For the $i$th copy of our (path) gadget, let us label the nodes $v^i_0, v^i_1, ..., v^i_{2r}$ and connect $v^i_j$ with $v^i_{j+1}$ for $j \in \{0, 1, ..., 2r-1\}$. Our gadget will also have two external edges coming out of the path endpoints $v^i_0$ and $v^i_{2r}$. These will connect to some external nodes $u, u'$. Let's say that the gadget is inserted forwards if $v^i_0$ is connected to $u$ and $v^i_{2r}$ is connected to $v$ and the gadget is inserted backwards if $v^i_{2r}$ is connected to $u$ and $v^i_0$ is connected to $v$. Importantly, for the innermost node $v^i_r$, it is impossible to tell whether the gadget is inserted forwards or backwards within $r$ rounds.

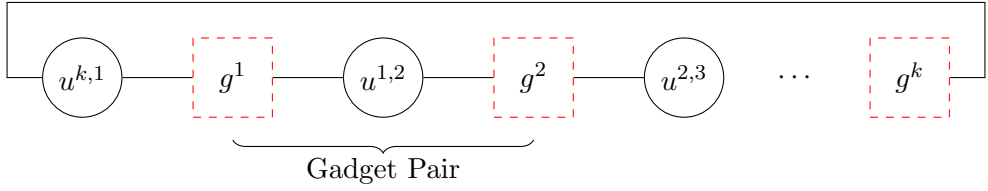
\begin{figure}
\centering
\begin{tikzpicture}[%
  auto,
  scale=1.0,
  gadgetnode/.style={
    rectangle,
    dashed,
    draw=red,
    minimum size=30pt,
    inner sep=1pt,
  },
  hollownode/.style={
    circle,
    draw=black,
    minimum size=30pt,
    inner sep=1pt,
  },
  ]
  \node[hollownode] (u01) at (0, 0) {$u^{k, 1}$};
  \node[gadgetnode] (g1) at (2, 0) {$\gadget{1}$};
  \node[hollownode] (u12) at (4, 0) {$u^{1, 2}$};
  \node[gadgetnode] (g2) at (6, 0) {$\gadget{2}$};
  \node[hollownode] (u23) at (8, 0) {$u^{2, 3}$};
  \node             (g3) at (9.5, 0) {$\cdots$};
  \node[gadgetnode] (g4) at (11, 0) {$\gadget{k}$};

  \draw (u01) -- (g1) -- (u12) -- (g2) -- (u23);
  \draw (g4) -- (12, 0) -- (12, 1) -- (-1, 1) -- (-1, 0) -- (u01);

  \draw [decorate,decoration={brace,amplitude=5pt,mirror,raise=5ex}]
  (2,0) -- (6,0) node[midway,yshift=-4em]{Gadget Pair};
\end{tikzpicture}
\caption{(Degree Two Case) For even $k$, we generate a distribution of counterexamples by inserting $k$ gadgets between $k$ anchor nodes $u^{1, 2}, u^{2, 3}, \cdots, u^{k-1, k}, u^{k, 1}$. Each gadget is independently and uniformly at random chosen to be inserted forwards or backwards. We then pair up gadgets, starting with $\{\gadget{1}, \gadget{2}\}$.}
\label{fig:tikz-gadgets}
\end{figure}

Now that we have described our gadget, we are ready to discuss how to use it. Let $k$ be the number of gadgets that we use; since we will pair them up later, we know that $k$ is even. A constant number of gadgets will suffice against deterministic LOCAL algorithms, but we will require on the order of $\ln (1 - \delta)^{-1}$ gadgets against Monte Carlo LOCAL algorithms to establish an failure probability of $\delta$. We introduce $k$ additional anchor nodes, outside of any gadget, and label them $u^{1, 2}, u^{2, 3}, \ldots, u^{k-1, k}, u^{k, 1}$, with the intention that anchor node $u^{i, j}$ is connected to gadgets $\gadget{i}$ and $\gadget{j}$. We then insert each gadget independently and uniformly at random either forwards or backwards between its two anchor nodes. The final result is a distribution over (counterexample) graphs; see \Cref{fig:tikz-gadgets} for a depiction.

As a preliminary observation, each graph in the support of distribution is a cycle on $k(2r + 2)$ nodes. Since this is an even number of nodes, the graphs in our distribution are indeed bipartite.

We are now ready to begin reasoning about the algorithm's performance on this distribution. Let us focus our attention on some innermost gadget node $v^i_r$ (note that each gadget pair has two of these). The algorithm may do one of the following to this gadget node (a deterministic LOCAL algorithm does one of them deterministically, and a Monte Carlo LOCAL algorithm will produce some distribution over these three possibilities):
\begin{itemize}
  \item leave the node $v^i_r$ unmatched,
  \item match $v^i_r$ with $v^i_{r-1}$,
  \item or match $v^i_r$ with $v^i_{r+1}$.
\end{itemize}

For technical reasons, we will post-process the algorithm's output to make the first case impossible. Whenever $v^i_r$ is unmatched, we will match it with $v^i_{r+1}$, unmatching $v^i_{r+1}$ with its previous match. This can only improve the size of the matching, so if we can show that the post-processed matching is not a $(1 + \epsilon)$-approximation, the algorithm's matching is also not a $(1 + \epsilon)$-approximation.

Next, we say that $v^i_r$ is matched to the right if it is matched with $v^i_{r+1}$ and its gadget is inserted forwards, or if it is matched with $v^i_{r-1}$ and its gadget is inserted backwards. Similarly, we say that it is matched to the left if it is matched with $v^i_{r-1}$ and its gadget is inserted forwards, or if it is matched with $v^i_{r+1}$ and its gadget is inserted backwards. Our post-processing has guaranteed that exactly one of these possibilities occurs, and our random insertions guarantee each innermost node is independently and uniformly at random matched to the right or left (we have essentially XOR'd the algorithm's decision pattern with a random bit string). Interestingly, this means that we can even guard against a Monte Carlo algorithm whose computation at various nodes all have access to some shared public randomness.

We now pair up adjacent gadgets; each pair will have an independent chance of creating an unmatched node. For $i \in \{1, 2, ..., k/2\}$, the $i$th gadget pair includes the nodes of $\gadget{2i-1}$, the nodes of $\gadget{2i}$, as well as the anchor node between them, namely $u^{2i-1, 2i}$. We know that there is a $1/2$ probability that the two innermost gadget nodes $v^{2i-1}_r$ and $v^{2i}_r$ have one match left and one match right. When this event occurs, there are an odd number of nodes between these two matching edges, so one of the nodes between these two innermost gadget nodes must be unmatched. Furthermore, this $1/2$ failure chance occurs independently across all our gadget pairs.

Against deterministic LOCAL algorithms, we are essentially done. We now choose the minimum number of gadgets possible: $k = 2$. We also choose $r = \floor{\frac{1}{9} (\epsilon^{-1} - 7)}$ which is $\Omega(\epsilon^{-1})$; since $\epsilon < 1/16$ this guarantees $r \ge 1$ and hence $v^i_{r-1}$ and $v^i_{r+1}$ actually exist. Furthermore, our overall construction has $n = k(2r + 2) = 4r + 4 = \Omega(\epsilon^{-1})$ nodes, which the algorithm was guaranteed to be able to handle. Observe that in expectation (over the randomness of our counterexample distribution) the algorithm leaves half a node unmatched across these two gadgets. It must get performance at least this bad on some specific graph in the support of this distribution. On that counterexample, its multiplicative approximation is at least:
\begin{align*}
  \frac{n}{n - 1/2} &= 1 + \frac{1}{2n-1} \\
                    &= 1 + \frac{1}{8r + 7} \\
                    &= 1 + \frac{1}{8\floor{\frac{1}{9} (\epsilon^{-1} - 7)} + 7} \\
                    &\ge 1 + \frac{1}{\frac{8}{9} (\epsilon^{-1} - 7) + 7} \\
                    &> 1 + \epsilon & \text{(since } \epsilon < 1/16 \implies \epsilon < 1/7 \text{)}
\end{align*}
This contradicts the $(1 + \epsilon)$-multiplicative approximation guarantee of the deterministic LOCAL algorithm, completing that half of the proof.

We now consider Monte Carlo LOCAL algorithms. We now need additional gadgets to force the failure probability to $\delta$, so we will be choosing a new value for $k$. Since each gadget pair is independent, we can use a Chernoff bound to control the probability that the algorithm fails on too few gadget pairs; it will be exponentially small in the total number of gadgets.

Due to linearity of expectation, the expected number of gadget pairs which fail is $\mu = \frac{k}{2} \cdot \frac12 = \frac{k}{4}$. We need some additional multiplicative error to run our Chernoff bound argument, so let us consider the likelihood that at most half as many, $\frac{k}{8}$ fail. Let $X$ be a random variable for the total number of gadget pair failures; since $X$ is a sum of independent Bernoulli variables we know that:
\begin{align*}
  \Pr[X \le (1 - \gamma) \mu] &\le e^{-\gamma^2 \mu / 2} \qquad \forall \gamma \in (0, 1) \\
  \Pr[X \le k/8] &\le e^{-(1/2)^2 (k/4) / 2} \\
  \Pr[X \le k/8] &\le e^{-k/32}
\end{align*}
To arrive at a contradiction, we want to show that the left-hand side is strictly less than $1 - \delta$. Hence it suffices to show that:
\begin{align*}
  e^{-k/32} &< 1 - \delta \\
  -k/32 &< \ln (1 - \delta) \\
  k &> 32 \ln (1 - \delta)^{-1}
\end{align*}

We can easily satisfy this by choosing $k = 2\ceil{16 \ln (1 - \delta)^{-1}} + 2$; note that this is guaranteed to be even and will contribute an asymptotic factor of $\Omega(\ln (1 - \delta)^{-1})$. Since we have now forced $\Pr[X \le k/8] < 1 - \delta$, we know that the complement probability is $\Pr[X > k/8] > \delta$. When this many failures occur, our Monte Carlo algorithm will have an approximation ratio which is strictly larger than:
\begin{align*}
  \frac{n}{n - k/8} &= 1 - \frac{k/8}{n - k/8} \\
                     &= 1 - \frac{k/8}{k(2r+2) - k/8} \\
                     &= 1 - \frac{1}{16r+15}
\end{align*}
Similar to the deterministic case, we can choose $r = \floor{\frac{1}{17} (\epsilon^{-1} - 15)}$ which is $\Omega(\epsilon^{-1})$ and since $\epsilon < 1/32$ this guarantees $r \ge 1$ which ensures $v^i_{r-1}$ and $v^i_{r+1}$ exist. Additionally our construction has $n = k(2r + 2) = \Omega(\epsilon^{-1} \ln (1 - \delta)^{-1})$ nodes which the algorithm was guaranteed to be able to handle. We now finish the approximation ratio analysis; the ratio is strictly larger than:
\begin{align*}
  \frac{n}{n - k/8} &= 1 - \frac{1}{16\floor{\frac{1}{17} (\epsilon^{-1} - 15)}+15} \\
                     &\ge 1 - \frac{1}{\frac{16}{17} (\epsilon^{-1} - 15) +15} \\
                     &> 1 + \epsilon & \text{(since } \epsilon < 1/32 \implies \epsilon < 1/15 \text{)}
\end{align*}
Hence with probability strictly greater than $\delta$, the approximation ratio is strictly worse than $(1 + \epsilon)$. This contradicts the approximation guarantee of the Monte Carlo LOCAL algorithm, completing the second part of the proof.
\end{proof}

\begin{remark}
  In addition to being able to guard against Monte Carlo algorithms with access to shared public randomness, it also does not matter if the algorithms know the graph size $n$ upfront. It is also possible to improve the upper bounds on allowed $\epsilon$ with additional observations, e.g. removing anchor nodes or arguing that in the deterministic case at least two nodes must actually be unmatched. At some point, this hinges on what a LOCAL algorithm is actually allowed to do in zero rounds.
\end{remark}

\subsection{Bipartite Regular Graphs of Even Degree}
\label{subsec:lower-degree-even}

We now consider the (medium difficulty) even degree case. We prove the following subresult in this subsection.

\begin{theorem}\label{thm:lower-degree-even}
  For any even degree $\Delta \ge 2$ and error $\epsilon \in (0, \frac{1}{16\Delta + 9})$, any deterministic LOCAL algorithm that computes a $(1 + \epsilon)$-multiplicative approximation for \maximummatching{} on bipartite $\Delta$-regular graphs with at least $n \ge \Omega(\Delta^{-1}\epsilon^{-1})$ nodes requires $\Omega(\Delta^{-1}\epsilon^{-1})$ rounds.

  For any even degree $\Delta \ge 2$, error $\epsilon \in (0, \frac{1}{32\Delta + 17})$, and failure probability $\delta \in (0, 1)$, any Monte Carlo LOCAL algorithm that computes a $(1 + \epsilon)$-multiplicative approximation with probability at least $1 - \delta$ for \maximummatching{} on bipartite $\Delta$-regular graphs with at least $n \ge \Omega(\Delta^{-1}\epsilon^{-1} \ln (1 - \delta)^{-1})$ nodes requires $\Omega(\Delta^{-1}\epsilon^{-1})$ rounds.
\end{theorem}

\begin{proof}
  In order to support larger degree than the degree two proof, we plan to replace each node $v$ from that construction with a ``node subgadget'' $\subgadget{v}$. Just as each node used to be connected to two other nodes, we will be able to connect each node subgadget to two other node subgadgets.

  Informally, our subgadget will offer $\Delta / 2$ nodes that are each missing two degree and hence can be connected to adjacent node subgadgets, for a total of $\Delta$ outgoing edges. We will then expect the typical solution to use exactly one of these external edges, which simulates the behavior of a node in the original proof. The fact that we have $\Delta$ outgoing edges is not a coincidence; a combination of regularity and bipartiteness mean that if we want all subgadget nodes with external edges to be on the same side of the bipartition, the number of external edges is necessarily a multiple of $\Delta$ (since each node on that side generates $\Delta$ edges and each node on the other side consumes $\Delta$ edges).

  \begin{figure}
\centering
\begin{tikzpicture}[%
  auto,
  scale=1.0,
  hollownode/.style={
    circle,
    draw=black,
    minimum size=35pt,
    inner sep=1pt,
  },
  redblock/.style={
    rectangle,
    solid,
    draw=red,
    fill=white,
    text=black,
    align=center,
  },
  ]

  \foreach \i in {1, ..., 4} {
    \node[hollownode] (ell\i) at (0, -1.5 * \i) {$\ell^v_{\i}$};
  }
  \foreach \j in {1, ..., 5} {
    \node[hollownode] (r\j) at (3, 0 - 1.5 * \j) {$r^v_{\j}$};
  }

  \foreach \i in {1, ..., 4} {
    \foreach \j in {1, ..., 5} {
      \ifnum \numexpr\i/2\relax=\j
        \draw[red, dashed] (ell\i) -- (r\j);
      \else
        \draw (ell\i) -- (r\j);
      \fi
    }
  }

  \draw[red, thick, dotted] (-1, -8.5) node[redblock] {Subgadget $\subgadget{v}$} -- (-1, -0.5) -- (4, -0.5) -- (4, -8.5) -- cycle;

  \foreach \j in {1, ..., 2} {
    \node[hollownode] (r\j') at (5, -1.5 * \j) {$r^{v'}_{\j}$};
    \node[hollownode] (r\j'') at (7, -1.5 * \j) {$r^{v''}_{\j}$};
    \draw[blue] (r\j) to[bend right=30] (r\j');
    \draw[blue] (r\j) to[bend left=30] (r\j'');
  }
\end{tikzpicture}
\caption{(Even Degree Case) This is our node subgadget $\subgadget{v}$ for degree $\Delta = 4$, which takes on the role of a node $v$ from the original construction. The red dashed lines emphasize missing edges (if they were not missing, the gadget would internally look like biclique $K_{\Delta, \Delta + 1}$). The subgadget is connected to two other subgadgets $\subgadget{v'}, \subgadget{v''}$ via blue curved edges to correspond with the original node being connected to two other nodes $v'$ and $v''$. To maintain the bipartiteness of the resulting graph, the nodes being replaced by these node subgadgets must all be from \emph{even-length} cycles.}
\label{fig:tikz-biclique}
\end{figure}
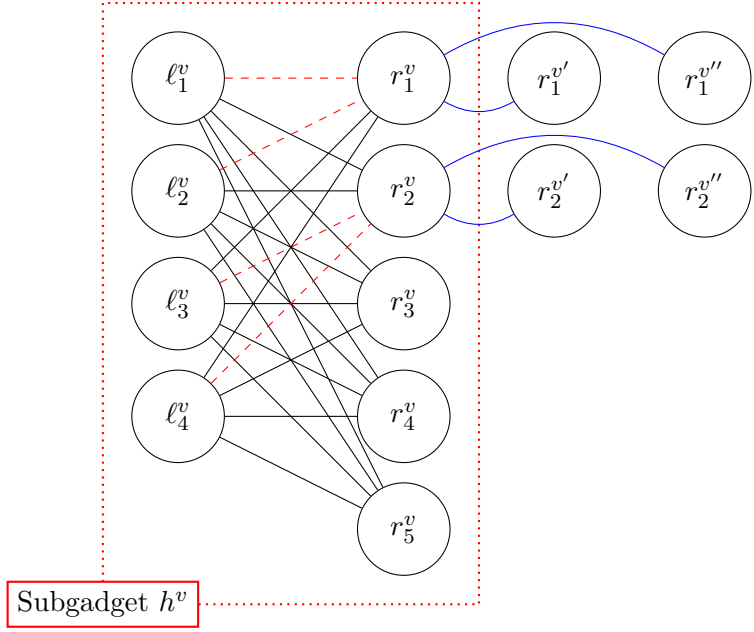

  \Cref{fig:tikz-biclique} depicts the subgadget design for degree $\Delta = 4$, and we now formally describe it for all even $\Delta \ge 2$. To replace what was originally node $v$, our (node) subgadget consists of the $(2\Delta + 2)$ nodes $\ell^v_1, \ell^v_2, ..., \ell^v_{\Delta}, r^v_1, r^v_2, ..., r^v_{\Delta + 1}$. Internally, we connect $\ell^v_{j'}$ with $r^v_j$ for all for all $j' \in \{1, 2, ..., \Delta\}$ and $j \in \{1, 2, ..., \Delta+1\}$ unless $j' \in \{2j - 1, 2j\}$. Our gadget will also have $\Delta$ external edges coming out of nodes $r^v_1, r^v_2, ..., r^v_{\Delta/2}$ (two each). If the node $v$ was connected to nodes $v'$ and $v''$, then we have external edges from $r^v_j$ to $r^{v'}_j$ and $r^{v''}_j$ for all $j \in \{1, 2, ..., \Delta/2\}$.

  Now that we have described our subgadget, we copy the path gadget construction from the proof of \Cref{thm:lower-degree-two}, replacing every node with a node subgadget. Reusing some notation from that proof, our path gadgets will guard against algorithms that run for $r$ rounds and we will be using $k$ path gadgets in total, where $k$ is even and to be decided later. We also use $k$ additional anchor nodes which also get turned into node subgadgets. In total, the total number of nodes we use is now $k(2r+2)(2\Delta+1)$.

  Let us consider what the algorithm does on a single node subgadget. For convenience, let's define the following sets of subgadget $v$'s nodes:
  \begin{align*}
    X^v &= \{\ell^v_1, ..., \ell^v_{\Delta}\}
        & (\abs{X^v} &= \Delta) \\
    Y^v &= \{r^v_1, ..., r^v_{\Delta / 2}\}
        & (\abs{Y^v} &= \Delta / 2) \\
    Z^v &= \{r^v_{\Delta / 2 + 1}, ..., r^v_{\Delta + 1}\}
        & (\abs{Z^v} &= \Delta / 2 + 1)
  \end{align*}
  The most natural thing to do for subgadget $v$ is to match the $\Delta / 2 + 1$ nodes of $Z^v$ to an equal number of nodes in $X^v$, then match the remaining nodes of $X^v$ to an equal number of nodes in $Y^v$. This leaves one remaining node from $Y^v$ to match with an external node.

  In fact, we will force the algorithm to do this on each node subgadget $\subgadget{v^i_r}$ that corresponds to an innermost node $v^i_r$ in some path gadget $\gadget{i}$. For convenience, let $v = v^i_r$. We will postprocess the algorithm's decisions on subgadget $\subgadget{v}$, and we will maintain that the total size of the matching does not decrease, so if the post-processed matching fails the approximation guarantee, so did the algorithm's original matching. First, if any node in $Z^v$ is unmatched, we match it to a node in $X^v$ which is not currently matched to another node in $Z^v$. This is always possible because $Z^v$ and $X^v$ are fully connected and $\abs{X^v} = \Delta \ge \Delta / 2 + 1 = \abs{Z^v}$. We possibly need to unmatch that node of $Z^v$ with a node in $Y^v$, but even with this in mind, we have not reduced the total size of the matching. Next, if any node in $X^v$ is unmatched after this, we match it to a node in $Y^v$ which is not currently matched to another node in $Y^v$. This is always possible because $X^v$ and $Y^v$ are almost fully connected; each node in $X^v$ is only missing a connection to a single node in $Y^v$, and between $Y^v$ and $Z^v$ there are $\Delta + 1$ nodes total, enough to guarantee a match for every node in $X^v$ even with the missing edges. We possibly need to unmatch that node of $Y^v$ with an external node, but again this does not reduce the total size of the matching. Finally, if any node in $Y^v$ is unmatched after this, we match it to an external node in $\subgadget{v^i_{r+1}}$, possibly unmatching that external node with a different external node. Since we are only doing this post-processing for innermost nodes $v = v^i_r$, we can safely do this all without collisions. As a result of our post-processing, exactly one node of $\subgadget{v^i_r}$ is matched to an external node, and the size of the matching has not decreased.

  At this point, we can follow the original proof with only minor changes. We say that $\subgadget{v^i_r}$ is matched to the right if its only external edge matched is with $\subgadget{v^i_{r+1}}$ and its gadget $\gadget{i}$ is inserted forwards, or if its only external edge matched is with $\subgadget{v^i_{r-1}}$ and its gadget $\gadget{i}$ is inserted backwards. Similarly, we say that it is matched to the left if its only external edge edge matched is with $\subgadget{v^i_{r-1}}$ and its gadget $\gadget{i}$ is inserted forwards, or if its only external edge matched is with $\subgadget{v^i_{r+1}}$ and its gadget $\gadget{i}$ is inserted backwards. Again, our post-processing has guaranteed that exactly one of these possibilities occurs, and our random insertions guarantee each innermost node subgadget is independently an uniformly at random matched to the right or left.

  We now pair up adjacent gadgets as before; for $i \in \{1, 2, ..., k/2\}$, the $i$th gadget pair includes the node subgadgets of $\gadget{2i-1}$, the node subgadgets of $\gadget{2i}$, as well as the node subgadget for the anchor node $u^{2i-1, 2i}$. We know that there is a $1/2$ probability that the two innermost node subgadgets have one match left and one match right. When this event occurs, there are an odd number of nodes between these two matching edges (an odd number of subgadgets times an odd number of nodes per subgadget), so one of these nodes must be unmatched. Furthermore, this $1/2$ failure chance occurs independently across all our gadget pairs.

  Against deterministic LOCAL algorithms we again choose a minimal $k = 2$ for a single gadget pair. On average against our distribution, the algorithm leaves $1/2$ of a node unmatched; we can begin to work out its multiplicative approximation ratio to be at least:
  \begin{align*}
  \frac{n}{n - 1/2} &= 1 + \frac{1}{2n-1} \\
                    &> 1 + \frac{1}{2n} \\
                    &= 1 + \frac{1}{2k(2r+2)(2\Delta+1)} \\
                    &= 1 + \frac{1}{8(2\Delta+1)(r+1)}
  \end{align*}
  To turn this into $(1 + \epsilon)$, we let $r = \floor{\frac{\epsilon^{-1} - 1}{16\Delta + 8}}$. This is $\Omega(\Delta^{-1}\epsilon^{-1})$ as claimed, and as long as $\epsilon < \frac{1}{16\Delta + 9}$, we have $r \ge 1$ so that our post-processing was valid (i.e. $v^i_{r-1}$ and $v^i_{r+1}$ exist). Additionally our construction has $\Omega(\Delta^{-1}\epsilon^{-1})$ nodes which the algorithm was guaranteed to be able to handle. The upshot is that after making these choices for $k$ and $r$, we have contradicted the $(1 + \epsilon)$-multiplicative approximation guarantee of the deterministic LOCAL algorithm, completing that half of the proof.

  To deal with the other half, Monte Carlo LOCAL algorithms, we again invoke a Chernoff bound. Since the failure probability is the same, the Chernoff calculation is identical and we need to choose $k = 2\ceil{16\ln(1 - \delta)^{-1}} + 2$ so that the number of gadget pair failures $X$ satisfies $\Pr[X \le k/8] < 1 - \delta$. Hence the complement probability is $\Pr[X > k/8] < \delta$, but this many failures would make our Monte Carlo algorithm have an approximation which is strictly larger than:
  \begin{align*}
    \frac{n}{n - k/8} &= 1 + \frac{k/8}{n - k/8} \\
                      &= 1 + \frac{k/8}{k(2r+2)(2\Delta+1) - k/8} \\
                      &> 1 + \frac{1}{16(2\Delta+1)(r+1)}
  \end{align*}
  To turn this into $(1 + \epsilon)$, we let $r = \floor{\frac{\epsilon^{-1} - 1}{32\Delta + 16}}$. This is $\Omega(\Delta^{-1}\epsilon^{-1})$ as claimed, and as long as $\epsilon < \frac{1}{32\Delta + 17}$, we have $r \ge 1$ so that our post-processing was valid. Additionally our construction has $\Omega(\Delta^{-1}\epsilon^{-1}\ln (1 - \delta)^{-1})$ nodes which the algorithm was guaranteed to be able to handle. Hence with probability strictly greater than $\delta$, the approximation ratio is strictly worse than $(1 + \epsilon)$. This contradicts the approximation guarantee of the Monte Carlo LOCAL algorithm, completing the second part of the proof.
\end{proof}

\subsection{Bipartite Regular Graphs of General Degree}
\label{subsec:lower-degree-general}

We now consider the (hard difficulty) general degree case. Unfortunately, the way we got the even degree proof to work via node subgadgets runs into complications when try the same technique for odd degree. If we design a node subgadget with an odd number of nodes, we will have an odd number of external edges and the algorithm will be able to break symmetry (the subgadget's two neighbors will not look the same). If we design a node subgadget with an even number of nodes, then we will want to match an even number of nodes to external nodes and the algorithm will be able to evenly split them between the subgadget's two neighbors.

Hence we will need a different attack angle to handle odd degrees. The key insight is to view the path gadget as a degree-two way to simulate a very long edge (so that the algorithm cannot see the neighbor on the other side of the long edge). We will design edge subgadgets to do the same thing for higher degrees.

We are now ready to prove \Cref{thm:lower-general}, which is restated below for convenience.

\lowergeneral*

\begin{proof}
  In order to support larger degree than the degree two proof, we plan to create a (bipartite, regular) degree five graph and then replace each edge $(u, v)$ in it with an ``edge subgadget'' $\subgadget{u, v}$. We choose degree five to be the original degree because of two factors: we want to write every degree $\Delta \ge 2$ as a sum $\Delta = 3x + 2y$ where $x, y \in \{0, 1, ..., \Delta\}$ and we needed the coefficients on $x$ and $y$ in that expression to both be at least two. It is easy to prove that every degree $\Delta \ge 2$ fits this criteria; if $\Delta$ is even, then it can be written as $\Delta = 0 + 2y$ for some $y \ge 1$ and we can then always choose $x = 0$. If $\Delta$ is odd, then it is at least $\Delta \ge 3$ and so $\Delta - 3$ is even; we can hence write $\Delta = 3 + 2y$ for some $y \ge 0$ and then always choose $x = 1$.
  
  This property that $\Delta = 3x + 2y$ means that for any specific value of $\Delta$ we will only use two different edge subgadgets. Our ``blue'' subgadget will induce degree $x$ and our ``red'' subgadget will induce degree $y$, and hence three blue subgadgets and two blue subgadgets induce a real degree of $\Delta$.

  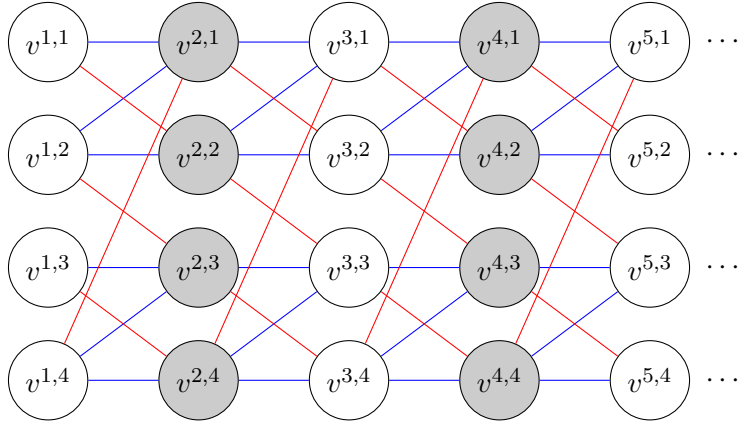
\begin{figure}
\centering
\begin{tikzpicture}[%
  auto,
  scale=1.0,
  hollownode/.style={
    circle,
    draw=black,
    minimum size=30pt,
    inner sep=1pt,
  },
  shadednode/.style={
    circle,
    draw=black,
    fill=black!20,
    minimum size=30pt,
    inner sep=1pt,
  },
  redblock/.style={
    rectangle,
    solid,
    draw=red,
    fill=white,
    text=black,
    align=center,
  },
  ]
  \foreach \i in {1, 3, 5} {
    \foreach \j in {1, 2, 3, 4} {
      \node[hollownode] (v\i\j) at (2 * \i, -1.5 * \j) {$v^{\i, \j}$};
    }
  }
  \foreach \i in {2, 4} {
    \foreach \j in {1, 2, 3, 4} {
      \node[shadednode] (v\i\j) at (2 * \i, -1.5 * \j) {$v^{\i, \j}$};
    }
  }
  \node at (11, -1.5) {$\cdots$};
  \node at (11, -3) {$\cdots$};
  \node at (11, -4.5) {$\cdots$};
  \node at (11, -6) {$\cdots$};

  \foreach \i / \j in {1/2, 2/3, 3/4, 4/5} {
    \draw[blue] (v\i1) -- (v\j1) -- (v\i2) -- (v\j2);
    \draw[blue] (v\i3) -- (v\j3) -- (v\i4) -- (v\j4);
    \draw[red]  (v\i1) -- (v\j2);
    \draw[red]  (v\i2) -- (v\j3);
    \draw[red]  (v\i3) -- (v\j4);
    \draw[red]  (v\i4) -- (v\j1);
  }

\end{tikzpicture}
\caption{(General Degree Case) These are nodes from five adjacent layers of our base graph, which is bipartite, regular, and has degree five. The edges are colored blue and red so that every node has three incident blue edges and two incident red edges. Each of these colored edges $(u, v)$ will be replaced by an edge subgadget $\subgadget{u, v}$.}
\label{fig:tikz-penta}
\end{figure}

  The exact degree five base graph is not actually important (we can always extract five disjoint perfect matchings and arbitrarily color them three blue, two red), but for concreteness we will reason about the following layered graph. Let $k$, the number of layers, be a parameter chosen later; we guarantee that $k$ will be a multiple of four (each set of four layers corresponds to a gadget pair in the original proof). We have four nodes per layer, so our graph consists of $4k$ nodes labelled $v^{i, j}$ for $i \in \{1, 2, ..., k\}$ and $j \in \{1, 2, 3, 4\}$. For all $i, i' \in \{1, 2, ..., k\}$, where $i + 1 \equiv i' \pmod{k}$, and for all $(j, j') \in \{(1, 1), (2, 1), (2, 2), (3, 3), (4, 3), (4, 4)\}$ we have an blue edge $(v^{i, j}, v^{i', j'})$. For all $i, i' \in \{1, 2, ..., k\}$, where $i + 1 \equiv i' \pmod{k}$, and for all $(j, j') \in \{(1, 2), (2, 3), (3, 4), (4, 1)\}$ we have an red edge $(v^{i, j}, v^{i', j'})$. This base graph is depicted in \Cref{fig:tikz-penta}.

  \begin{figure}
\centering
\begin{tikzpicture}[%
  auto,
  scale=1.0,
  hollownode/.style={
    circle,
    draw=black,
    minimum size=30pt,
    inner sep=1pt,
  },
  redblock/.style={
    rectangle,
    solid,
    draw=red,
    fill=white,
    text=black,
    align=center,
  },
  ]
  \node[hollownode] (u) at (1, 0) {$u$};
  \foreach \i in {1, 2} {
    \node[hollownode] (w\i1) at (3*\i, 3.75)  {$w^{u, v}_{\i, 1}$};
    \node             (w\i2) at (3*\i, 2.25)  {$\vdots$};
    \node[hollownode] (w\i3) at (3*\i, 0.75)  {$w^{u, v}_{\i, z}$};
    \node[hollownode] (w\i4) at (3*\i, -0.75) {$w^{u, v}_{\i, z+1}$};
    \node             (w\i5) at (3*\i, -2.25) {$\vdots$};
    \node[hollownode] (w\i6) at (3*\i, -3.75) {$w^{u, v}_{\i, \Delta}$};
  }
  \draw (u) -- (w11);
  \draw (u) -- (w13);
  \draw[red, dashed] (w11) -- (w21);
  \draw (w11) -- (w23);
  \draw (w11) -- (w24);
  \draw (w11) -- (w26);
  \draw (w13) -- (w21);
  \draw[red, dashed] (w13) -- (w23);
  \draw (w13) -- (w24);
  \draw (w13) -- (w26);
  \draw (w14) -- (w21);
  \draw (w14) -- (w23);
  \draw (w14) -- (w24);
  \draw (w14) -- (w26);
  \draw (w16) -- (w21);
  \draw (w16) -- (w23);
  \draw (w16) -- (w24);
  \draw (w16) -- (w26);

  \node (w231) at (7.5, 3.75) {$\hdots$};
  \node (w233) at (7.5, 0.75) {$\hdots$};
  
  \foreach \i/\j in {3/\rho-1, 4/\rho} {
    \node[hollownode] (w\i1) at (3*\i, 3.75)  {$w^{u, v}_{\j, 1}$};
    \node             (w\i2) at (3*\i, 2.25)  {$\vdots$};
    \node[hollownode] (w\i3) at (3*\i, 0.75)  {$w^{u, v}_{\j, z}$};
    \node[hollownode] (w\i4) at (3*\i, -0.75) {$w^{u, v}_{\j, z+1}$};
    \node             (w\i5) at (3*\i, -2.25) {$\vdots$};
    \node[hollownode] (w\i6) at (3*\i, -3.75) {$w^{u, v}_{\j, \Delta}$};
  }
  \node[hollownode]  (v) at (14, 0) {$v$};

  \draw (w21) -- (w231) -- (w31);
  \draw (w23) -- (w233) -- (w33);
  
  \draw[red, dashed] (w31) -- (w41);
  \draw (w31) -- (w43);
  \draw (w31) -- (w44);
  \draw (w31) -- (w46);
  \draw (w33) -- (w41);
  \draw[red, dashed] (w33) -- (w43);
  \draw (w33) -- (w44);
  \draw (w33) -- (w46);
  \draw (w34) -- (w41);
  \draw (w34) -- (w43);
  \draw (w34) -- (w44);
  \draw (w34) -- (w46);
  \draw (w36) -- (w41);
  \draw (w36) -- (w43);
  \draw (w36) -- (w44);
  \draw (w36) -- (w46);
  \draw (w41) -- (v);
  \draw (w43) -- (v);

  \draw[red, thick, dotted] (2, -4.75) node[redblock] {Subgadget $\subgadget{u, v}$} -- (2, 4.75) -- (13, 4.75) -- (13, -4.75) -- cycle;
\end{tikzpicture}
\caption{(General Degree Case) This is our edge subgadget $\subgadget{u, v}$, which simulates $z$ parallel edges between $u$ and $v$ of length $2\rho + 1$. The red dashed lines emphasize missing edges. To maintain the bipartiteness of the resulting graph, $u$ and $v$ must be on opposite sides of the bipartition.}
\label{fig:tikz-edge}
\end{figure}
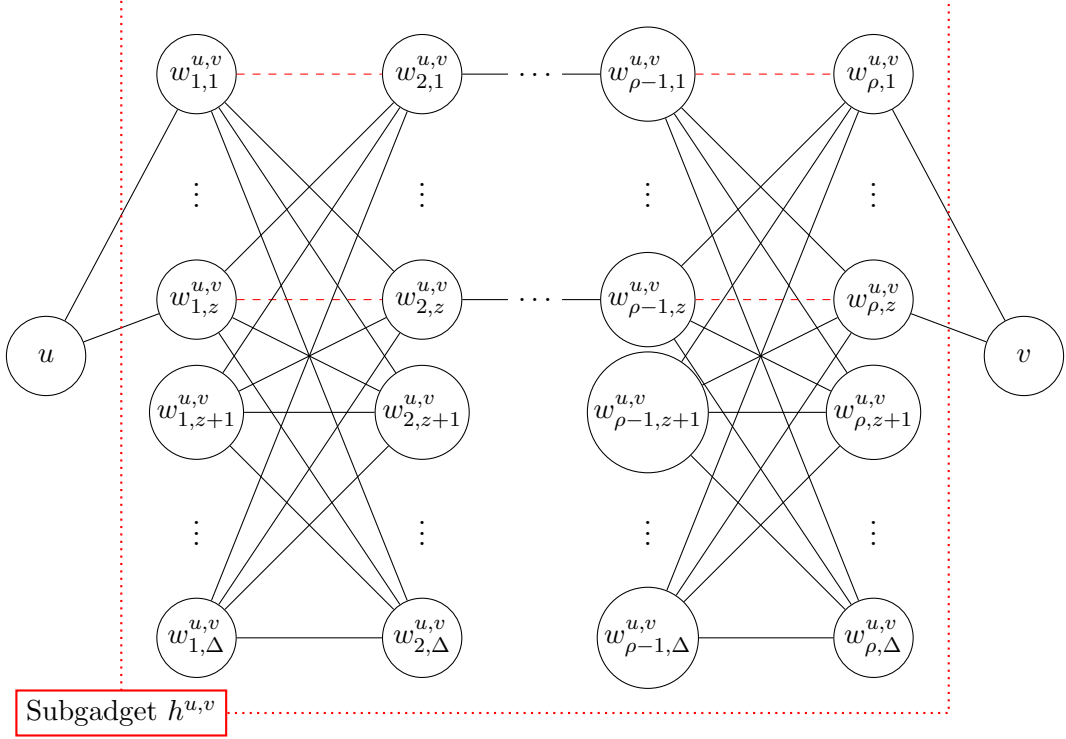

  We are now ready to present our edge subgadget $\subgadget{u, v}$. It takes in three parameters: the degree $\Delta$, an induced degree $z \in \{0, 1, ..., \Delta\}$, and a length $\rho$ which is a positive even integer. Informally, our edge subgadget behaves like $z$ parallel edges between $u$ and $v$ that have a length of $2\rho + 1$ (i.e. it takes that many hops to get from $u$ to $v$), implying that a $2\rho$-round algorithm cannot see the identity of $v$ when making a matching decision for $u$. This is done by adding $\Delta \rho$ additional nodes along with edges that use up all of their $\Delta$ degree along with using up $z$ degree from both $u$ and $v$.

  We now formally describe the edge subgadget, which is depicted in \Cref{fig:tikz-edge}. Our subgadget consists of $(\Delta \rho)$ nodes $\{w^{u, v}_{i, j}\}_{i \in \{1, 2, ..., \rho\}, j \in \{1, 2, ..., \Delta\}}$. We think of our nodes as being arranged in $\rho$ layers with $\Delta$ nodes in each layer. We say that a node $w^{u, v}_{i, j}$ is in layer $i$, and each layer is only connected to adjacent layers (except for the first layer $i = 1$ and last layer $i = \rho$, which are also connected to $u$ and $v$ respectively).

  When going up from an odd layer $i \in \{1, 3, ..., \rho - 1\}$ to the even layer $i + 1$, the nodes are fully connected except for a matching between the first $z$ nodes. That is, $w^{u, v}_{i, j}$ is connected to $w^{u, v}_{i+1, j'}$ for all $j, j' \in \{1, 2, ..., \Delta\}$ unless $j = j' \le z$. When going up from an even layer $i \in \{2, 4, ..., \rho - 2\}$ to the even layer $i + 1$, the only connection is a matching between the first $z$ nodes. That is, $w^{u, v}_{i, j}$ is connected to $w^{u, v}_{i+1, j}$ for all $j \in \{1, 2, ..., z\}$. Finally, $u$ is connected to $w^{u, v}_{1, j}$ for $j \in \{1, 2, ..., z\}$ and $v$ is connected to $w^{u, v}_{\rho, j}$ for $j \in \{1, 2, ..., z\}$.

  After replacing all the edges in our base graph with edge subgadgets, we have a total of $n = 4k + 10k(\Delta \rho)$ nodes. Next, we will identify the equivalents of path gadgets and gadget pairs within our graph, so that we can argue that algorithms will wind up making mistakes at a certain rate.

  Our new version of a gadget consists of an innermost node $v^{i, j}$ with \emph{even} $i$ along with the five edge subgadgets incident to that node. Our construction has provided the following properties: (i) gadgets do not overlap and (ii) each gadget can be removed and reinserted into the graph one of $3!2! = 12$ ways without changing the overall structure of the graph, since the blue edge subgadgets can be freely permuted with each other and the red edge subgadgets can be freely permuted with each other. This leaves all nodes $v^{i, j}$ with \emph{odd} $i$ serving as anchor nodes between gadgets.

  We now group gadgets as follows; for $i \in \{1, 2, ..., k/4\}$, the $i$th gadget group includes the five gadgets whose innermost nodes are $v^{4i-2, 1}, v^{4i-2, 2}, v^{4i-2, 4}, v^{4i, 1}, v^{4i, 2}$ (each consisting of that node and the five edge subgadgets incident to that node), as well as the anchor node $v^{4i-1, 1}$. Intuitively, the situation is that we are examining an anchor node and its five neighbors; the algorithm needs exactly one of these neighbors to be matched towards our anchor node, but it cannot do so precisely due to our random method for inserting gadgets. The algorithm gets to decide how to split a gadget's probability of matching its innermost node towards a red edge or towards a blue edge, but after that the red edges get shuffled and the blue edges get shuffled and so any decision has a constant probability of getting messed up.

  In order to have independence between gadget groups later, it will be convenient to run a Yao's minimax principle argument. Let $Y$ be a random variable representing the algorithm's randomness (even shared public randomness between the nodes), and $y$ be a specific value for this randomness. If the algorithm is deterministic, then this does not matter and we can say e.g. $y$ can only be the empty string. If we condition on $Y = y$, then the algorithm behaves deterministically and in particular, every gadget's innermost node is deterministically matched to either a node in a red edge subgadget or a blue edge subgadget. After these decisions are locked in, we randomly insert all gadgets and consider what happens to our gadget groups.

  We want to lower bound the probability that some node in the gadget group is unmatched. One sufficient condition for this to occur is that an even number of the five gadgets in the gadget group match towards the anchor node. The algorithm's decision (which can depend on $Y$) amounts to choosing which subset of gadgets get to participate, but participating gadgets have a probability of either $1/2$ or $1/3$ of changing the parity of the outcome. Since the probability of changing the parity is never greater than $1/2$, this means that the even-parity probability after a coin flip (viewed as a weighted average) is closer to the even-parity probability before the flip than the odd-parity probability before the flip; hence the even-parity probability can never drop below $1/2$. Furthermore, due to conditioning $Y = y$, all gadget groups are independent Bernoulli variables (possibly with different probabilities).

  Against deterministic LOCAL algorithms the conditioning doesn't matter, and we choose a minimal $k = 4$ giving us a single gadget group. On average against our distribution, the algorithm leaves at least $1/2$ of a node unmatched; we can begin to work out its multiplicative approximation ratio to be at least:
  \begin{align*}
    \frac{n}{n-1/2} &= 1 + \frac{1}{2n - 1} \\
                     &> 1 + \frac{1}{2n} \\
                     &= 1 + \frac{1}{8k + 20k\Delta\rho} \\
                     &= 1 + \frac{1}{32 + 80\Delta\rho}
  \end{align*}
  To turn this into $(1 + \epsilon)$, we let $\rho = \floor{\frac{\epsilon^{-1} - 32}{80\Delta}}$. This is $\Omega(\Delta^{-1} \epsilon^{-1})$ as claimed, and as long as $\epsilon < \frac{1}{160\Delta + 32}$ we have $\rho \ge 2$ so that our edge subgadgets are valid. Additionally our construction has $\Omega(\Delta^{-1} \epsilon^{-1})$ nodes which the algorithm was guaranteed to be able to handle and we guarded against $2\rho$ round algorithms which is $\Omega(\Delta^{-1} \epsilon^{-1})$ as well. The upshot is that after making these choices for $k$ and $\rho$, we have contradicted the $(1 + \epsilon)$-multiplicative approximation guarantee of the deterministic LOCAL algorithm, completing that half of the proof.

  We need to be a little careful rerunning our Chernoff bound, since the failure probability may not be exactly $1/2$, although it turns out that a lower bound on the probability suffices. The expected number of gadget groups which fail is $\mu$, which is between $\frac12 \frac{k}{4} = \frac{k}{8}$ and $\frac{k}{4}$ (the exact value in this range can depend on $y$). We need some multiplicative error to run our Chernoff bound argument, so let us consider the situation where at most $\frac{k}{16}$ fail. Let $X$ be a random variable for the total number of gadget pair failures; since $X$ is a sum of independent Bernoulli variables (independent conditioned on $Y = y$) we know that:
  \begin{align*}
    \Pr [X \le (1 - \gamma) \mu | Y = y]
      &\le e^{-\gamma^2 \mu / 2} \qquad \forall \gamma \in (0, 1) \\
    \Pr [X \le \mu/2 | Y = y]
      &\le e^{- \mu / 8} \\
    \Pr [X \le k/16 | Y = y]
      &\le e^{- k / 64}
  \end{align*}
  
  To arrive at a contradiction, we want to show that the left-hand side is strictly less than $1 - \delta$. Hence it suffices to show that:
  \begin{align*}
    e^{- k / 64} &< 1 - \delta \\
    -k/64 &< \ln (1 - \delta) \\
    k &> 64 \ln (1 - \delta)^{-1}
  \end{align*}

  We satisfy this by choosing $k = 4\ceil{16\ln(1 - \delta)^{-1}} + 4$; note that this is guaranteed to be a multiple of four and will contribute an asymptotic factor of $\Omega(\ln (1 - \delta)^{-1})$. Since we have now forced $\Pr[X \le k/16] < 1 - \delta$, we know that the complement probability is $\Pr[X > k/16] < \delta$. When this many failures occur, our Monte Carlo algorithm will have an approximation ratio which is strictly larger than:
  \begin{align*}
    \frac{n}{n - k/16} &= 1 + \frac{k/16}{n - k/16} \\
                       &= 1 + \frac{k/16}{4k + 10k\Delta\rho - k/16} \\
                       &= 1 + \frac{1}{64 + 160\Delta\rho - 1} \\
                       &> 1 + \frac{1}{64 + 160\Delta\rho}
  \end{align*}
  To turn this into $(1 + \epsilon)$, we let $\rho = \floor{\frac{\epsilon^{-1} - 64}{160\Delta}}$. This is $\Omega(\Delta^{-1}\epsilon^{-1})$ as claimed, and as long as $\epsilon < \frac{1}{320\Delta + 64}$, we have $\rho \ge 2$ so that our edge subgadgets are valid. Additionally our construction has $\Omega(\Delta^{-1}\epsilon^{-1}\ln (1 - \delta)^{-1})$ nodes which the algorithm was guaranteed to be able to handle and we guarded against $2\rho$ round algorithms which is $\Omega(\Delta^{-1}\epsilon^{-1})$. Hence with probability strictly greater than $\delta$ (just over our counterexample distribution), conditioned on $Y = y$, the approximation ratio is strictly worse than $(1 + \epsilon)$. But we now finish our Yao's minimax argument and observe that since this argument worked for every value of $y$, we can safely uncondition this statement by taking a weighted average. Hence with probability strictly greater than $\delta$ (over both our counterexample distribution and the randomness of the algorithm), the approximation ratio is strictly worse than $(1 + \epsilon)$.
  
  This contradicts the approximation guarantee of the Monte Carlo LOCAL algorithm, completing the second part of the proof.
\end{proof}

\bibliographystyle{plainurl}
\bibliography{refs}

\begin{thebibliography}{100}

\bibitem{AggarwalMSZ03}
Gagan Aggarwal, Rajeev Motwani, Devavrat Shah, and An~Zhu.
\newblock Switch scheduling via randomized edge coloring.
\newblock In {\em FOCS}, page 502, 2003.

\bibitem{AhmadiK20}
Mohamad Ahmadi and Fabian Kuhn.
\newblock Distributed maximum matching verification in {CONGEST}.
\newblock In {\em {DISC}}, volume 179 of {\em LIPIcs}, pages 37:1--37:18, 2020.

\bibitem{AhmadiKO18}
Mohamad Ahmadi, Fabian Kuhn, and Rotem Oshman.
\newblock Distributed approximate maximum matching in the {CONGEST} model.
\newblock In {\em {DISC}}, volume 121 of {\em LIPIcs}, pages 6:1--6:17, 2018.

\bibitem{Alon03}
Noga Alon.
\newblock A simple algorithm for edge-coloring bipartite multigraphs.
\newblock {\em Information Processing Letters}, 85(6):301--302, March 2003.

\bibitem{alon1986fast}
Noga Alon, L{\'a}szl{\'o} Babai, and Alon Itai.
\newblock A fast and simple randomized parallel algorithm for the maximal independent set problem.
\newblock {\em Journal of algorithms}, 7(4):567--583, 1986.

\bibitem{AlonBK10}
Noga Alon, Sonny Ben{-}Shimon, and Michael Krivelevich.
\newblock A note on regular ramsey graphs.
\newblock {\em J. Graph Theory}, 64(3):244--249, 2010.
\newblock URL: \url{https://doi.org/10.1002/jgt.20453}, \href {https://doi.org/10.1002/JGT.20453} {\path{doi:10.1002/JGT.20453}}.

\bibitem{AlonCHKRS10}
Noga Alon, Amin Coja{-}Oghlan, Hi{\^{e}}p H{\`{a}}n, Mihyun Kang, Vojtech R{\"{o}}dl, and Mathias Schacht.
\newblock Quasi-randomness and algorithmic regularity for graphs with general degree distributions.
\newblock {\em {SIAM} J. Comput.}, 39(6):2336--2362, 2010.
\newblock \href {https://doi.org/10.1137/070709529} {\path{doi:10.1137/070709529}}.

\bibitem{AlonFK84a}
Noga Alon, Shmuel Friedland, and Gil Kalai.
\newblock Every 4-regular graph plus an edge contains a 3-regular subgraph.
\newblock {\em J. Comb. Theory {B}}, 37(1):92--93, 1984.
\newblock \href {https://doi.org/10.1016/0095-8956(84)90048-0} {\path{doi:10.1016/0095-8956(84)90048-0}}.

\bibitem{AlonFK84}
Noga Alon, Shmuel Friedland, and Gil Kalai.
\newblock Regular subgraphs of almost regular graphs.
\newblock {\em J. Comb. Theory {B}}, 37(1):79--91, 1984.
\newblock \href {https://doi.org/10.1016/0095-8956(84)90047-9} {\path{doi:10.1016/0095-8956(84)90047-9}}.

\bibitem{AlonM21}
Noga Alon and Guy Moshkovitz.
\newblock Limitations on regularity lemmas for clustering graphs.
\newblock {\em Adv. Appl. Math.}, 124:102135, 2021.
\newblock URL: \url{https://doi.org/10.1016/j.aam.2020.102135}, \href {https://doi.org/10.1016/J.AAM.2020.102135} {\path{doi:10.1016/J.AAM.2020.102135}}.

\bibitem{AlonP11}
Noga Alon and Pawel Pralat.
\newblock Modular orientations of random and quasi-random regular graphs.
\newblock {\em Comb. Probab. Comput.}, 20(3):321--329, 2011.
\newblock \href {https://doi.org/10.1017/S0963548310000544} {\path{doi:10.1017/S0963548310000544}}.

\bibitem{ArarCCSW18}
Moab Arar, Shiri Chechik, Sarel Cohen, Cliff Stein, and David Wajc.
\newblock Dynamic matching: Reducing integral algorithms to approximately-maximal fractional algorithms.
\newblock In {\em {ICALP}}, volume 107 of {\em LIPIcs}, pages 7:1--7:16, 2018.

\bibitem{AsratianK99}
Armen~S. Asratian and Nikolai~N. Kuzjurin.
\newblock On the number of nearly perfect matchings in almost regular uniform hypergraphs.
\newblock {\em Discret. Math.}, 207(1-3):1--8, 1999.

\bibitem{AssadiBKL23}
Sepehr Assadi, Soheil Behnezhad, Sanjeev Khanna, and Huan Li.
\newblock On regularity lemma and barriers in streaming and dynamic matching.
\newblock In {\em {STOC}}, pages 131--144, 2023.

\bibitem{AssadiS23}
Sepehr Assadi and Janani Sundaresan.
\newblock Hidden permutations to the rescue: Multi-pass streaming lower bounds for approximate matchings.
\newblock In {\em {FOCS}}, pages 909--932, 2023.

\bibitem{Babai80}
L{\'{a}}szl{\'{o}} Babai.
\newblock On the complexity of canonical labeling of strongly regular graphs.
\newblock {\em {SIAM} J. Comput.}, 9(1):212--216, 1980.
\newblock \href {https://doi.org/10.1137/0209018} {\path{doi:10.1137/0209018}}.

\bibitem{Babai14}
L{\'{a}}szl{\'{o}} Babai.
\newblock On the automorphism groups of strongly regular graphs {I}.
\newblock In {\em ITCS}, pages 359--368. {ACM}, 2014.
\newblock \href {https://doi.org/10.1145/2554797.2554830} {\path{doi:10.1145/2554797.2554830}}.

\bibitem{BabaiCSTW13}
L{\'{a}}szl{\'{o}} Babai, Xi~Chen, Xiaorui Sun, Shang{-}Hua Teng, and John Wilmes.
\newblock Faster canonical forms for strongly regular graphs.
\newblock In {\em FOCS}, pages 157--166. {IEEE} Computer Society, 2013.
\newblock \href {https://doi.org/10.1109/FOCS.2013.25} {\path{doi:10.1109/FOCS.2013.25}}.

\bibitem{BalliuB0O24}
Alkida Balliu, Thomas Boudier, Sebastian Brandt, and Dennis Olivetti.
\newblock Tight lower bounds in the supported {LOCAL} model.
\newblock In {\em PODC}, pages 95--105. {ACM}, 2024.
\newblock \href {https://doi.org/10.1145/3662158.3662798} {\path{doi:10.1145/3662158.3662798}}.

\bibitem{Balliu0CORS19}
Alkida Balliu, Sebastian Brandt, Yi{-}Jun Chang, Dennis Olivetti, Mika{\"{e}}l Rabie, and Jukka Suomela.
\newblock The distributed complexity of locally checkable problems on paths is decidable.
\newblock In Peter Robinson and Faith Ellen, editors, {\em Proceedings of the 2019 {ACM} Symposium on Principles of Distributed Computing, {PODC} 2019, Toronto, ON, Canada, July 29 - August 2, 2019}, pages 262--271. {ACM}, 2019.
\newblock \href {https://doi.org/10.1145/3293611.3331606} {\path{doi:10.1145/3293611.3331606}}.

\bibitem{Balliu0COSS22}
Alkida Balliu, Sebastian Brandt, Yi{-}Jun Chang, Dennis Olivetti, Jan Studen{\'{y}}, and Jukka Suomela.
\newblock Efficient classification of locally checkable problems in regular trees.
\newblock In Christian Scheideler, editor, {\em 36th International Symposium on Distributed Computing, {DISC} 2022, October 25-27, 2022, Augusta, Georgia, {USA}}, volume 246 of {\em LIPIcs}, pages 8:1--8:19. Schloss Dagstuhl - Leibniz-Zentrum f{\"{u}}r Informatik, 2022.
\newblock URL: \url{https://doi.org/10.4230/LIPIcs.DISC.2022.8}, \href {https://doi.org/10.4230/LIPICS.DISC.2022.8} {\path{doi:10.4230/LIPICS.DISC.2022.8}}.

\bibitem{BalliuBCOSST23}
Alkida Balliu, Sebastian Brandt, Yi{-}Jun Chang, Dennis Olivetti, Jan Studen{\'{y}}, Jukka Suomela, and Aleksandr Tereshchenko.
\newblock Locally checkable problems in rooted trees.
\newblock {\em Distributed Comput.}, 36(3):277--311, 2023.
\newblock URL: \url{https://doi.org/10.1007/s00446-022-00435-9}, \href {https://doi.org/10.1007/S00446-022-00435-9} {\path{doi:10.1007/S00446-022-00435-9}}.

\bibitem{BalliuBHORS21}
Alkida Balliu, Sebastian Brandt, Juho Hirvonen, Dennis Olivetti, Mika{\"{e}}l Rabie, and Jukka Suomela.
\newblock Lower bounds for maximal matchings and maximal independent sets.
\newblock {\em J. {ACM}}, 68(5):39:1--39:30, 2021.
\newblock \href {https://doi.org/10.1145/3461458} {\path{doi:10.1145/3461458}}.

\bibitem{Balliu0KO21}
Alkida Balliu, Sebastian Brandt, Fabian Kuhn, and Dennis Olivetti.
\newblock Improved distributed lower bounds for {MIS} and bounded (out-)degree dominating sets in trees.
\newblock In {\em {PODC}}, pages 283--293. {ACM}, 2021.
\newblock \href {https://doi.org/10.1145/3465084.3467901} {\path{doi:10.1145/3465084.3467901}}.

\bibitem{Balliu0KO22}
Alkida Balliu, Sebastian Brandt, Fabian Kuhn, and Dennis Olivetti.
\newblock Distributed {$\Delta$}-coloring plays hide-and-seek.
\newblock In {\em {STOC}}, pages 464--477. {ACM}, 2022.
\newblock \href {https://doi.org/10.1145/3519935.3520027} {\path{doi:10.1145/3519935.3520027}}.

\bibitem{Balliu0KO23}
Alkida Balliu, Sebastian Brandt, Fabian Kuhn, and Dennis Olivetti.
\newblock Distributed maximal matching and maximal independent set on hypergraphs.
\newblock In {\em {SODA}}, pages 2632--2676. {SIAM}, 2023.
\newblock URL: \url{https://doi.org/10.1137/1.9781611977554.ch100}, \href {https://doi.org/10.1137/1.9781611977554.CH100} {\path{doi:10.1137/1.9781611977554.CH100}}.

\bibitem{BalliuBO22}
Alkida Balliu, Sebastian Brandt, and Dennis Olivetti.
\newblock Distributed lower bounds for ruling sets.
\newblock {\em {SIAM} J. Comput.}, 51(1):70--115, 2022.
\newblock URL: \url{https://doi.org/10.1137/20m1381770}, \href {https://doi.org/10.1137/20M1381770} {\path{doi:10.1137/20M1381770}}.

\bibitem{BalliuGKO23}
Alkida Balliu, Mohsen Ghaffari, Fabian Kuhn, and Dennis Olivetti.
\newblock Node and edge averaged complexities of local graph problems.
\newblock {\em Distributed Comput.}, 36(4):451--473, 2023.
\newblock URL: \url{https://doi.org/10.1007/s00446-023-00453-1}, \href {https://doi.org/10.1007/S00446-023-00453-1} {\path{doi:10.1007/S00446-023-00453-1}}.

\bibitem{Bar-YehudaCGS17}
Reuven Bar{-}Yehuda, Keren Censor{-}Hillel, Mohsen Ghaffari, and Gregory Schwartzman.
\newblock Distributed approximation of maximum independent set and maximum matching.
\newblock In {\em {PODC}}, pages 165--174, 2017.

\bibitem{Bar-YehudaCS16}
Reuven Bar{-}Yehuda, Keren Censor{-}Hillel, and Gregory Schwartzman.
\newblock A distributed (2+{\(\epsilon\)})-approximation for vertex cover in {O}(log{\(\delta\)}/{\(\epsilon\)} log log {\(\delta\)}) rounds.
\newblock In {\em PODC}, pages 3--8. {ACM}, 2016.
\newblock \href {https://doi.org/10.1145/2933057.2933086} {\path{doi:10.1145/2933057.2933086}}.

\bibitem{BarenboimEPS12}
Leonid Barenboim, Michael Elkin, Seth Pettie, and Johannes Schneider.
\newblock The locality of distributed symmetry breaking.
\newblock In {\em 53rd Annual {IEEE} Symposium on Foundations of Computer Science, {FOCS} 2012, New Brunswick, NJ, USA, October 20-23, 2012}, pages 321--330. {IEEE} Computer Society, 2012.
\newblock \href {https://doi.org/10.1109/FOCS.2012.60} {\path{doi:10.1109/FOCS.2012.60}}.

\bibitem{BarenboimT18}
Leonid Barenboim and Yaniv Tzur.
\newblock Distributed symmetry-breaking with improved vertex-averaged complexity.
\newblock {\em CoRR}, abs/1805.09861, 2018.
\newblock URL: \url{http://arxiv.org/abs/1805.09861}, \href {https://arxiv.org/abs/1805.09861} {\path{arXiv:1805.09861}}.

\bibitem{BenBasatEKS23}
Ran Ben{-}Basat, Guy Even, Ken{-}ichi Kawarabayashi, and Gregory Schwartzman.
\newblock Optimal distributed covering algorithms.
\newblock {\em Distributed Comput.}, 36(1):45--55, 2023.
\newblock URL: \url{https://doi.org/10.1007/s00446-021-00391-w}, \href {https://doi.org/10.1007/S00446-021-00391-W} {\path{doi:10.1007/S00446-021-00391-W}}.

\bibitem{BenbasatKS19}
Ran Ben-Basat, {Ken-ichi} Kawarabayashi, and Gregory Schwartzman.
\newblock Parameterized distributed algorithms.
\newblock In {\em DISC}, 2019.

\bibitem{BhattacharyaKS23}
Sayan Bhattacharya, Peter Kiss, and Thatchaphol Saranurak.
\newblock Dynamic (1+{$\epsilon$})-approximate matching size in truly sublinear update time.
\newblock In {\em FOCS}, pages 1563--1588, 2023.

\bibitem{BhattacharyaKSW23}
Sayan Bhattacharya, Peter Kiss, Thatchaphol Saranurak, and David Wajc.
\newblock Dynamic matching with better-than-2 approximation in polylogarithmic update time.
\newblock In {\em SODA}, pages 100--128, 2023.

\bibitem{Brandt19}
Sebastian Brandt.
\newblock An automatic speedup theorem for distributed problems.
\newblock In {\em {PODC}}, pages 379--388. {ACM}, 2019.
\newblock \href {https://doi.org/10.1145/3293611.3331611} {\path{doi:10.1145/3293611.3331611}}.

\bibitem{BrandtFHKLRSU16}
Sebastian Brandt, Orr Fischer, Juho Hirvonen, Barbara Keller, Tuomo Lempi{\"{a}}inen, Joel Rybicki, Jukka Suomela, and Jara Uitto.
\newblock A lower bound for the distributed lov{\'{a}}sz local lemma.
\newblock In {\em STOC}, pages 479--488. {ACM}, 2016.
\newblock \href {https://doi.org/10.1145/2897518.2897570} {\path{doi:10.1145/2897518.2897570}}.

\bibitem{BrandtHKLOPRSU17}
Sebastian Brandt, Juho Hirvonen, Janne~H. Korhonen, Tuomo Lempi{\"{a}}inen, Patric R.~J. {\"{O}}sterg{\aa}rd, Christopher Purcell, Joel Rybicki, Jukka Suomela, and Przemyslaw Uznanski.
\newblock {LCL} problems on grids.
\newblock In Elad~Michael Schiller and Alexander~A. Schwarzmann, editors, {\em Proceedings of the {ACM} Symposium on Principles of Distributed Computing, {PODC} 2017, Washington, DC, USA, July 25-27, 2017}, pages 101--110. {ACM}, 2017.
\newblock \href {https://doi.org/10.1145/3087801.3087833} {\path{doi:10.1145/3087801.3087833}}.

\bibitem{0002MNSU25}
Sebastian Brandt, Yannic Maus, Ananth Narayanan, Florian Schager, and Jara Uitto.
\newblock On the locality of hall's theorem.
\newblock In Yossi Azar and Debmalya Panigrahi, editors, {\em Proceedings of the 2025 Annual {ACM-SIAM} Symposium on Discrete Algorithms, {SODA} 2025, New Orleans, LA, USA, January 12-15, 2025}, pages 4198--4226. {SIAM}, 2025.
\newblock \href {https://doi.org/10.1137/1.9781611978322.143} {\path{doi:10.1137/1.9781611978322.143}}.

\bibitem{Bo20}
Sebastian Brandt and Dennis Olivetti.
\newblock Truly tight-in-{\(\Delta\)} bounds for bipartite maximal matching and variants.
\newblock In Yuval Emek and Christian Cachin, editors, {\em {PODC}}, pages 69--78. {ACM}, 2020.
\newblock \href {https://doi.org/10.1145/3382734.3405745} {\path{doi:10.1145/3382734.3405745}}.

\bibitem{BuchbinderNW23}
Niv Buchbinder, Joseph~(Seffi) Naor, and David Wajc.
\newblock Lossless online rounding for online bipartite matching (despite its impossibility).
\newblock In {\em SODA}, pages 2030--2068, 2023.

\bibitem{BuryGMMSVZ19}
Marc Bury, Elena Grigorescu, Andrew McGregor, Morteza Monemizadeh, Chris Schwiegelshohn, Sofya Vorotnikova, and Samson Zhou.
\newblock Structural results on matching estimation with applications to streaming.
\newblock {\em Algorithmica}, 81(1):367--392, 2019.
\newblock URL: \url{https://doi.org/10.1007/s00453-018-0449-y}, \href {https://doi.org/10.1007/S00453-018-0449-Y} {\path{doi:10.1007/S00453-018-0449-Y}}.

\bibitem{CarrollGT09}
Teena Carroll, David~J. Galvin, and Prasad Tetali.
\newblock Matchings and independent sets of a fixed size in regular graphs.
\newblock {\em J. Comb. Theory {A}}, 116(7):1219--1227, 2009.
\newblock URL: \url{https://doi.org/10.1016/j.jcta.2008.12.008}, \href {https://doi.org/10.1016/J.JCTA.2008.12.008} {\path{doi:10.1016/J.JCTA.2008.12.008}}.

\bibitem{ChangKP19}
Yi{-}Jun Chang, Tsvi Kopelowitz, and Seth Pettie.
\newblock An exponential separation between randomized and deterministic complexity in the {LOCAL} model.
\newblock {\em {SIAM} J. Comput.}, 48(1):122--143, 2019.
\newblock \href {https://doi.org/10.1137/17M1117537} {\path{doi:10.1137/17M1117537}}.

\bibitem{ChatterjeeGP20}
Soumyottam Chatterjee, Robert Gmyr, and Gopal Pandurangan.
\newblock Sleeping is efficient: {MIS} in \emph{O}(1)-rounds node-averaged awake complexity.
\newblock In Yuval Emek and Christian Cachin, editors, {\em PODC}, pages 99--108. {ACM}, 2020.
\newblock \href {https://doi.org/10.1145/3382734.3405718} {\path{doi:10.1145/3382734.3405718}}.

\bibitem{ChungL06}
Fan R.~K. Chung and Lincoln Lu.
\newblock Survey: Concentration inequalities and martingale inequalities: {A} survey.
\newblock {\em Internet Math.}, 3(1):79--127, 2006.
\newblock \href {https://doi.org/10.1080/15427951.2006.10129115} {\path{doi:10.1080/15427951.2006.10129115}}.

\bibitem{CohenW18}
Ilan~Reuven Cohen and David Wajc.
\newblock Randomized online matching in regular graphs.
\newblock In {\em SODA}, pages 960--979, 2018.

\bibitem{ColeH82}
Richard Cole and John~E. Hopcroft.
\newblock On edge coloring bipartite graphs.
\newblock {\em {SIAM} J. Comput.}, 11(3):540--546, 1982.
\newblock \href {https://doi.org/10.1137/0211043} {\path{doi:10.1137/0211043}}.

\bibitem{CzygrinowH03}
Andrzej Czygrinow and Michal Hanckowiak.
\newblock Distributed algorithm for better approximation of the maximum matching.
\newblock In {\em {COCOON}}, volume 2697 of {\em Lecture Notes in Computer Science}, pages 242--251, 2003.

\bibitem{DahlhausK92}
Elias Dahlhaus and Marek Karpinski.
\newblock Perfect matching for regular graphs is {AC}{\textdegree}-hard for the general matching problem.
\newblock {\em J. Comput. Syst. Sci.}, 44(1):94--102, 1992.
\newblock \href {https://doi.org/10.1016/0022-0000(92)90005-4} {\path{doi:10.1016/0022-0000(92)90005-4}}.

\bibitem{EvenMR15}
Guy Even, Moti Medina, and Dana Ron.
\newblock Distributed maximum matching in bounded degree graphs.
\newblock In {\em {ICDCN}}, pages 18:1--18:10, 2015.

\bibitem{Feuilloley17}
Laurent Feuilloley.
\newblock How long it takes for an ordinary node with an ordinary {ID} to output?
\newblock In Shantanu Das and S{\'{e}}bastien Tixeuil, editors, {\em SIROCCO}, volume 10641 of {\em Lecture Notes in Computer Science}, pages 263--282. Springer, 2017.
\newblock \href {https://doi.org/10.1007/978-3-319-72050-0\_16} {\path{doi:10.1007/978-3-319-72050-0\_16}}.

\bibitem{Fischer20}
Manuela Fischer.
\newblock Improved deterministic distributed matching via rounding.
\newblock {\em Distributed Comput.}, 33(3-4):279--291, 2020.
\newblock URL: \url{https://doi.org/10.1007/s00446-018-0344-4}, \href {https://doi.org/10.1007/S00446-018-0344-4} {\path{doi:10.1007/S00446-018-0344-4}}.

\bibitem{DBLP:conf/focs/FischerGK17}
Manuela Fischer, Mohsen Ghaffari, and Fabian Kuhn.
\newblock Deterministic distributed edge-coloring via hypergraph maximal matching.
\newblock In Chris Umans, editor, {\em 58th {IEEE} Annual Symposium on Foundations of Computer Science, {FOCS} 2017, Berkeley, CA, USA, October 15-17, 2017}, pages 180--191. {IEEE} Computer Society, 2017.
\newblock \href {https://doi.org/10.1109/FOCS.2017.25} {\path{doi:10.1109/FOCS.2017.25}}.

\bibitem{FischerMU22}
Manuela Fischer, Slobodan Mitrovic, and Jara Uitto.
\newblock Deterministic (1+\emph{{\(\epsilon\)}})-approximate maximum matching with poly(1/\emph{{\(\epsilon\)}}) passes in the semi-streaming model and beyond.
\newblock In {\em STOC}, pages 248--260, 2022.

\bibitem{FlaxmanH07}
Abraham~D. Flaxman and Shlomo Hoory.
\newblock Maximum matchings in regular graphs of high girth.
\newblock {\em Electron. J. Comb.}, 14(1), 2007.
\newblock \href {https://doi.org/10.37236/1002} {\path{doi:10.37236/1002}}.

\bibitem{GGKMR18}
Mohsen Ghaffari, Themis Gouleakis, Christian Konrad, Slobodan Mitrovic, and Ronitt Rubinfeld.
\newblock Improved massively parallel computation algorithms for {MIS}, matching, and vertex cover.
\newblock In {\em PODC}, pages 129--138. {ACM}, 2018.
\newblock \href {https://doi.org/10.1145/3212734.3212743} {\path{doi:10.1145/3212734.3212743}}.

\bibitem{G23}
Mohsen Ghaffari and Christoph Grunau.
\newblock Faster deterministic distributed {MIS} and approximate matching.
\newblock In Barna Saha and Rocco~A. Servedio, editors, {\em Proceedings of the 55th Annual {ACM} Symposium on Theory of Computing, {STOC} 2023, Orlando, FL, USA, June 20-23, 2023}, pages 1777--1790. {ACM}, 2023.
\newblock \href {https://doi.org/10.1145/3564246.3585243} {\path{doi:10.1145/3564246.3585243}}.

\bibitem{ecomposition}
Mohsen Ghaffari and Christoph Grunau.
\newblock Near-optimal deterministic network decomposition and ruling set, and improved {MIS}.
\newblock In {\em 65th {IEEE} Annual Symposium on Foundations of Computer Science, {FOCS} 2024, Chicago, IL, USA, October 27-30, 2024}, pages 2148--2179. {IEEE}, 2024.
\newblock \href {https://doi.org/10.1109/FOCS61266.2024.00007} {\path{doi:10.1109/FOCS61266.2024.00007}}.

\bibitem{GhaffariHK18}
Mohsen Ghaffari, David~G. Harris, and Fabian Kuhn.
\newblock On derandomizing local distributed algorithms.
\newblock In {\em FOCS}, pages 662--673, 2018.

\bibitem{GhaffariKMU18}
Mohsen Ghaffari, Fabian Kuhn, Yannic Maus, and Jara Uitto.
\newblock Deterministic distributed edge-coloring with fewer colors.
\newblock In {\em STOC}, pages 418--430, 2018.

\bibitem{GoelKK13}
Ashish Goel, Michael Kapralov, and Sanjeev Khanna.
\newblock Perfect matchings in {O(n log n)} time in regular bipartite graphs.
\newblock {\em SIAM Journal on Computing}, 42(3):1392--1404, 2013.
\newblock \href {https://arxiv.org/abs/https://doi.org/10.1137/100812513} {\path{arXiv:https://doi.org/10.1137/100812513}}, \href {https://doi.org/10.1137/100812513} {\path{doi:10.1137/100812513}}.

\bibitem{GoldwasserG17}
Shafi Goldwasser and Ofer Grossman.
\newblock Bipartite perfect matching in pseudo-deterministic {NC}.
\newblock In {\em {ICALP}}, volume~80 of {\em LIPIcs}, pages 87:1--87:13, 2017.
\newblock URL: \url{https://doi.org/10.4230/LIPIcs.ICALP.2017.87}, \href {https://doi.org/10.4230/LIPICS.ICALP.2017.87} {\path{doi:10.4230/LIPICS.ICALP.2017.87}}.

\bibitem{Golowich23}
Louis Golowich.
\newblock A new berry-esseen theorem for expander walks.
\newblock In {\em STOC}, pages 10--22. {ACM}, 2023.
\newblock \href {https://doi.org/10.1145/3564246.3585141} {\path{doi:10.1145/3564246.3585141}}.

\bibitem{0001LSSS19}
Fabrizio Grandoni, Stefano Leonardi, Piotr Sankowski, Chris Schwiegelshohn, and Shay Solomon.
\newblock {(1} + {\(\epsilon\)})-approximate incremental matching in constant deterministic amortized time.
\newblock In {\em SODA}, pages 1886--1898, 2019.

\bibitem{GrunauR022}
Christoph Grunau, V{\'{a}}clav Rozhon, and Sebastian Brandt.
\newblock The landscape of distributed complexities on trees and beyond.
\newblock In Alessia Milani and Philipp Woelfel, editors, {\em {PODC} '22: {ACM} Symposium on Principles of Distributed Computing, Salerno, Italy, July 25 - 29, 2022}, pages 37--47. {ACM}, 2022.
\newblock \href {https://doi.org/10.1145/3519270.3538452} {\path{doi:10.1145/3519270.3538452}}.

\bibitem{GuptaGPW19}
Anupam Gupta, Guru Guruganesh, Binghui Peng, and David Wajc.
\newblock Stochastic online metric matching.
\newblock In {\em ICALP}, volume 132 of {\em LIPIcs}, pages 67:1--67:14, 2019.

\bibitem{GuptaP13}
Manoj Gupta and Richard Peng.
\newblock Fully dynamic {(1+ $\epsilon$)}-approximate matchings.
\newblock In {\em FOCS}, pages 548--557, 2013.

\bibitem{Harris19}
David~G. Harris.
\newblock Distributed local approximation algorithms for maximum matching in graphs and hypergraphs.
\newblock In {\em FOCS}, pages 700--724, 2019.

\bibitem{HopcroftK73}
John~E. Hopcroft and Richard~M. Karp.
\newblock An n\({}^{\mbox{5/2}}\) algorithm for maximum matchings in bipartite graphs.
\newblock {\em {SIAM} J. Comput.}, 2(4):225--231, 1973.
\newblock \href {https://doi.org/10.1137/0202019} {\path{doi:10.1137/0202019}}.

\bibitem{IsraelI86}
Amos Israeli and Alon Itai.
\newblock A fast and simple randomized parallel algorithm for maximal matching.
\newblock {\em Inf. Process. Lett.}, 22(2):77--80, 1986.
\newblock \href {https://doi.org/10.1016/0020-0190(86)90144-4} {\path{doi:10.1016/0020-0190(86)90144-4}}.

\bibitem{IzumiKY24}
Taisuke Izumi, Naoki Kitamura, and Yutaro Yamaguchi.
\newblock A nearly linear-time distributed algorithm for exact maximum matching.
\newblock In David~P. Woodruff, editor, {\em Proceedings of the 2024 {ACM-SIAM} Symposium on Discrete Algorithms, {SODA} 2024, Alexandria, VA, USA, January 7-10, 2024}, pages 4062--4082. {SIAM}, 2024.
\newblock \href {https://doi.org/10.1137/1.9781611977912.141} {\path{doi:10.1137/1.9781611977912.141}}.

\bibitem{Janson08}
Svante Janson.
\newblock Large deviations for sums of partly dependent random variables.
\newblock {\em Random Structures \& Algorithms}, 24(3):234--248, 2004.
\newblock URL: \url{https://onlinelibrary.wiley.com/doi/abs/10.1002/rsa.20008}, \href {https://arxiv.org/abs/https://onlinelibrary.wiley.com/doi/pdf/10.1002/rsa.20008} {\path{arXiv:https://onlinelibrary.wiley.com/doi/pdf/10.1002/rsa.20008}}, \href {https://doi.org/10.1002/rsa.20008} {\path{doi:10.1002/rsa.20008}}.

\bibitem{KahnK98}
Jeff Kahn and Jeong~Han Kim.
\newblock Random matchings in regular graphs.
\newblock {\em Comb.}, 18(2):201--226, 1998.
\newblock \href {https://doi.org/10.1007/PL00009817} {\path{doi:10.1007/PL00009817}}.

\bibitem{KarandeMT11}
Chinmay Karande, Aranyak Mehta, and Pushkar Tripathi.
\newblock Online bipartite matching with unknown distributions.
\newblock In {\em STOC}, pages 587--596, 2011.

\bibitem{KarpVV90}
Richard~M. Karp, Umesh~V. Vazirani, and Vijay~V. Vazirani.
\newblock An optimal algorithm for on-line bipartite matching.
\newblock In {\em STOC}, pages 352--358, 1990.

\bibitem{KawarabayashiKS20}
Ken{-}ichi Kawarabayashi, Seri Khoury, Aaron Schild, and Gregory Schwartzman.
\newblock Improved distributed approximations for maximum independent set.
\newblock In {\em DISC}, volume 179 of {\em LIPIcs}, pages 35:1--35:16, 2020.

\bibitem{KS25}
Seri Khoury and Aaron Schild.
\newblock Round elimination via self-reduction: Closing gaps for distributed maximal matching.
\newblock {\em CoRR}, abs/2505.15654, 2025.
\newblock URL: \url{https://doi.org/10.48550/arXiv.2505.15654}, \href {https://arxiv.org/abs/2505.15654} {\path{arXiv:2505.15654}}, \href {https://doi.org/10.48550/ARXIV.2505.15654} {\path{doi:10.48550/ARXIV.2505.15654}}.

\bibitem{Kucera87}
Ludek Kucera.
\newblock Canonical labeling of regular graphs in linear average time.
\newblock In {\em FOCS}, pages 271--279. {IEEE} Computer Society, 1987.
\newblock \href {https://doi.org/10.1109/SFCS.1987.11} {\path{doi:10.1109/SFCS.1987.11}}.

\bibitem{KuhnMW06}
Fabian Kuhn, Thomas Moscibroda, and Roger Wattenhofer.
\newblock The price of being near-sighted.
\newblock In {\em SODA}, pages 980--989. {ACM} Press, 2006.
\newblock URL: \url{http://dl.acm.org/citation.cfm?id=1109557.1109666}.

\bibitem{KuhnMW16}
Fabian Kuhn, Thomas Moscibroda, and Roger Wattenhofer.
\newblock Local computation: Lower and upper bounds.
\newblock {\em J. {ACM}}, 63(2):17:1--17:44, 2016.
\newblock \href {https://doi.org/10.1145/2742012} {\path{doi:10.1145/2742012}}.

\bibitem{Lampert18}
Christoph~H Lampert, Liva Ralaivola, and Alexander Zimin.
\newblock Dependency-dependent bounds for sums of dependent random variables.
\newblock {\em arXiv preprint arXiv:1811.01404}, 2018.

\bibitem{Linial92}
Nathan Linial.
\newblock Locality in distributed graph algorithms.
\newblock {\em {SIAM} J. Comput.}, 21(1):193--201, 1992.
\newblock \href {https://doi.org/10.1137/0221015} {\path{doi:10.1137/0221015}}.

\bibitem{LotkerPP15}
Zvi Lotker, Boaz Patt{-}Shamir, and Seth Pettie.
\newblock Improved distributed approximate matching.
\newblock {\em J. {ACM}}, 62(5):38:1--38:17, 2015.
\newblock \href {https://doi.org/10.1145/2786753} {\path{doi:10.1145/2786753}}.

\bibitem{LotkerPR09}
Zvi Lotker, Boaz Patt{-}Shamir, and Adi Ros{\'{e}}n.
\newblock Distributed approximate matching.
\newblock {\em {SIAM} J. Comput.}, 39(2):445--460, 2009.
\newblock \href {https://doi.org/10.1137/080714403} {\path{doi:10.1137/080714403}}.

\bibitem{Luby86}
Michael Luby.
\newblock A simple parallel algorithm for the maximal independent set problem.
\newblock {\em {SIAM} J. Comput.}, 15(4):1036--1053, 1986.
\newblock \href {https://doi.org/10.1137/0215074} {\path{doi:10.1137/0215074}}.

\bibitem{Madry13}
Aleksander Madry.
\newblock Navigating central path with electrical flows: From flows to matchings, and back.
\newblock In {\em FOCS}, pages 253--262, 2013.

\bibitem{McDiarmid89}
Colin McDiarmid.
\newblock {\em On the method of bounded differences}, page 148–188.
\newblock London Mathematical Society Lecture Note Series. Cambridge University Press, 1989.

\bibitem{McGregor05}
Andrew McGregor.
\newblock Finding graph matchings in data streams.
\newblock In {\em {APPROX}}, volume 3624 of {\em Lecture Notes in Computer Science}, pages 170--181, 2005.

\bibitem{Mehta13}
Aranyak Mehta.
\newblock Online matching and ad allocation.
\newblock {\em Found. Trends Theor. Comput. Sci.}, 8(4):265--368, 2013.
\newblock \href {https://doi.org/10.1561/0400000057} {\path{doi:10.1561/0400000057}}.

\bibitem{MotwaniR95}
R.~Motwani and P.~Raghavan.
\newblock {\em Randomized Algorithms}.
\newblock Cambridge International Series on Parallel Computation. Cambridge University Press, 1995.
\newblock URL: \url{https://books.google.com/books?id=QKVY4mDivBEC}.

\bibitem{Naor91}
Moni Naor.
\newblock A lower bound on probabilistic algorithms for distributive ring coloring.
\newblock {\em {SIAM} J. Discret. Math.}, 4(3):409--412, 1991.
\newblock \href {https://doi.org/10.1137/0404036} {\path{doi:10.1137/0404036}}.

\bibitem{NaorS95}
Moni Naor and Larry~J. Stockmeyer.
\newblock What can be computed locally?
\newblock {\em {SIAM} J. Comput.}, 24(6):1259--1277, 1995.
\newblock \href {https://doi.org/10.1137/S0097539793254571} {\path{doi:10.1137/S0097539793254571}}.

\bibitem{PazS19}
Ami Paz and Gregory Schwartzman.
\newblock A (2+{\(\epsilon\)})-approximation for maximum weight matching in the semi-streaming model.
\newblock {\em {ACM} Trans. Algorithms}, 15(2):18:1--18:15, 2019.
\newblock \href {https://doi.org/10.1145/3274668} {\path{doi:10.1145/3274668}}.

\bibitem{peleg}
David Peleg.
\newblock {\em Distributed computing: a locality-sensitive approach}.
\newblock Society for Industrial and Applied Mathematics, USA, 2000.

\bibitem{PettieSanders04}
Seth Pettie and Peter Sanders.
\newblock A simpler linear time {$2/3 - \epsilon$} approximation for maximum weight matching.
\newblock {\em Information Processing Letters}, 91(6):271--276, 2004.

\bibitem{PlummerLovasz86}
M.D. Plummer and L.~Lov{\'a}sz.
\newblock {\em Matching Theory}.
\newblock ISSN. Elsevier Science, 1986.
\newblock URL: \url{https://books.google.com/books?id=mycZP-J344wC}.

\bibitem{abs-2406-19430}
V{\'{a}}clav Rozhon.
\newblock Invitation to local algorithms.
\newblock {\em CoRR}, abs/2406.19430, 2024.
\newblock URL: \url{https://doi.org/10.48550/arXiv.2406.19430}, \href {https://arxiv.org/abs/2406.19430} {\path{arXiv:2406.19430}}, \href {https://doi.org/10.48550/ARXIV.2406.19430} {\path{doi:10.48550/ARXIV.2406.19430}}.

\bibitem{Solomon16}
Shay Solomon.
\newblock Fully dynamic maximal matching in constant update time.
\newblock In {\em {FOCS}}, pages 325--334, 2016.

\bibitem{Spielman96}
Daniel~A. Spielman.
\newblock Faster isomorphism testing of strongly regular graphs.
\newblock In {\em STOC}, pages 576--584. {ACM}, 1996.
\newblock \href {https://doi.org/10.1145/237814.238006} {\path{doi:10.1145/237814.238006}}.

\bibitem{Suomela13}
Jukka Suomela.
\newblock Survey of local algorithms.
\newblock {\em {ACM} Comput. Surv.}, 45(2):24:1--24:40, 2013.
\newblock \href {https://doi.org/10.1145/2431211.2431223} {\path{doi:10.1145/2431211.2431223}}.

\bibitem{BrandLNPSS0W20}
Jan van~den Brand, Yin~Tat Lee, Danupon Nanongkai, Richard Peng, Thatchaphol Saranurak, Aaron Sidford, Zhao Song, and Di~Wang.
\newblock Bipartite matching in nearly-linear time on moderately dense graphs.
\newblock In {\em FOCS}, pages 919--930, 2020.

\bibitem{Yuster13}
Raphael Yuster.
\newblock Maximum matching in regular and almost regular graphs.
\newblock {\em Algorithmica}, 66(1):87--92, 2013.
\newblock URL: \url{https://doi.org/10.1007/s00453-012-9625-7}, \href {https://doi.org/10.1007/S00453-012-9625-7} {\path{doi:10.1007/S00453-012-9625-7}}.

\bibitem{Zamir23}
Or~Zamir.
\newblock Algorithmic applications of hypergraph and partition containers.
\newblock In {\em STOC}, pages 985--998. {ACM}, 2023.
\newblock \href {https://doi.org/10.1145/3564246.3585163} {\path{doi:10.1145/3564246.3585163}}.

\end{thebibliography}

\appendix
\section{Preliminaries}\label{app:prelims}

\subsection{Matchings}

We use the following well-known characterization of the matching polytope.
\begin{theorem}[Folklore, e.g. \cite{PlummerLovasz86}] \label{thm:polytope}
Let $G$ be an undirected graph and let $\mathcal{M}$ denote the matching polytope for $G$; that is the convex hull of all 0-1 vectors in $\mathbb{R}^m$ that are indicator vectors of matchings in $G$. Then, $\mathcal{M}$ can also be written as the intersection of the following families of halfspaces:

\begin{enumerate}
    \item $x_e\ge 0$ for all $e\in E(G)$.
    \item $\displaystyle\sum_{e\sim v} x_e \le 1$ for all $v\in V(G)$.
    \item $\displaystyle\sum_{e\in E(G[S])} x_e \le \frac{|S|-1}{2}$ for all sets $S\subseteq V(G)$ with odd size.
\end{enumerate}
\end{theorem}

For a matching $M$ of graph $G$, an augmenting path $P$ is a path in $G$ that alternates between edges in $M$ and those not in $M$ with the additional property that both its end points are unmatched in $M$. Let $\mathcal{P}$ be a collection of vertex disjoint augmenting paths, then $M' = M \cup \{\mathcal{P} \setminus M\} \setminus \{M \cap \mathcal{P}\}$ is a new matching of size $|M| + |\mathcal{P}|$ and we say that $M'$ is obtained by \emph{augmenting} $M$ with $\mathcal{P}$.

\begin{proposition}[Theorem 2.1 of \cite{PettieSanders04}, statement from Proposition 7.1 of \cite{Harris19}]\label{prop:augment}
Let $M$ be an arbitrary matching of $G$, and $OPT$ be the size of a maximum weight matching in $G$. Then, for any $\ell \geq 1$, there exists a collection $\mathcal{P}$ with $|\mathcal{P}| \geq \frac{1}{2}\left(OPT(1 - 1/\ell) - |M|\right)$of vertex disjoint augmenting paths where each path consist of at most $2\ell + 1$ edges.
\end{proposition}

\subsection{Concentration Inequalities}
We utilize the following concentration bound for sums of random variables with limited dependence as well as standard Chernoff bounds.

\begin{theorem}[Inequality (3) of \cite{Lampert18}, derived from Theorem 2.1 of \cite{Janson08}]\label{thm:CorrelatedChernoff}
Consider $n$ random variables $\mathcal A = \{X_1,X_2,\hdots,X_n\}$ with the property that $0\le X_i\le 1$ for all $i$ almost surely. Then

$$\mathbb{P}\left[\sum_{i=1}^n X_i - \sum_{i=1}^n \mathbb{E}[X_i] > \lambda\right] \le \exp\left(-\frac{2\lambda^2}{n \cdot \chi(\mathcal{A})}\right)$$

where $\chi(\mathcal A)$ is the chromatic number of the dependency graph of $\mathcal A$.
\end{theorem}

\begin{theorem}[Chernoff Bound]\label{thm:Chernoff}
Let $X = \sum_{i=1}^n X_i$, where the $X_i$s are independent random variables with value 0 or 1. Let $\mu = \mathbb{E}[X]$. Then
\begin{enumerate}
    \item $\mathbb{P}[X \ge (1 + \delta)\mu] \le e^{-\frac{\delta^2\mu}{2+\delta}}$ for all $\delta > 0$.
    \item $\mathbb{P}[X \le (1 - \delta)\mu] \le e^{-\mu\delta^2/2}$ for all $0 < \delta < 1$.
    \item $\mathbb{P}[|X - \mu| \ge \delta\mu] \le 2e^{-\mu\delta^2/3}$ for all $0 < \delta < 1$.
\end{enumerate}
\end{theorem}

\subsection{Martingales}\label{app:martingales}

We start with the following useful folklore observations about the conditional expectation of a function of a random variable, and conditional variance.
\begin{observation}\label{obs:ConditioningOnFunctionOfRV}
    Let $Y$ and $Z$ be two random variables such that each is a function of the other. For any random variable $X$, we have that: 
    $$\mathbb{E}[X\mid Y] = \mathbb{E}[X\mid Z]$$
\end{observation}

\begin{observation}\label{obs:conditionalVariance}
    Let $X$ and $Y$ be random variables. It holds that
    $Var[X\mid Y] = Var[X+f(Y)\mid Y]$
    where $f(Y)$ is a function of $Y$. 
\end{observation}

Next, we define martingales, supermartingales, and submartingales.  
\begin{definition}\label{def:martingale}\textbf{[Martingale, Supermartingale, and Submartingale]}\\
    A sequence of random variables $X_1,\cdots, X_t$ is called a martingale if $\mathbb{E}[X_i\mid X_1,\cdots, X_{i-1}] = X_{i-1}$ for any $i\geq 2$, a supermartingale if $\mathbb{E}[X_i\mid X_1,
    \cdots, X_{i-1}]\geq X_{i-1}$ for any $i\geq 2$, and a submartingale if $\mathbb{E}[X_i\mid X_1,\cdots, X_{i-1}]\leq X_{i-1}$ for any $i\geq 2$.
\end{definition}

We note that in the above definition of super/submartingales, we follow the definition of Chung and Lu~\cite{ChungL06}. In some other textbooks, the terms are reversed and the condition $\mathbb{E}[X_i\mid X_1,
    \cdots, X_{i-1}]\geq X_{i-1}$ corresponds to a submartingale, and the condition $\mathbb{E}[X_i\mid X_1,
    \cdots, X_{i-1}]\leq X_{i-1}$ corresponds to a supermartingale.

\subsubsection{Martingale Inequalities}

In this section we state some known martingale inequalities.


\begin{theorem}\label{thm:BoundedVarianceMartingaleUpperBound}\textbf{(Theorem 6.1 in~\cite{ChungL06}) Bounded Variance Martingale: Upper Bound}\\
    Let $X_1,\cdots, X_t$ be a martingale sequence satisfying:
\begin{enumerate}
    \item $Var[X_i\mid X_1,\cdots, X_{i-1}] \leq \phi_i$
    \item $|X_i-X_{i-1}| \leq M$
\end{enumerate}
$\forall i\geq 2$, where $\phi_i$ and $M$ are non-negative constants. Then for $\lambda>0$, 

$$\mathbb{P}[X_t-\mathbb{E}[X_t]\geq \lambda]\leq \exp\left({-\frac{\lambda^2}{2\left((\sum_{i=1}^t \phi_i) + M\lambda/3\right)}}\right)$$
 
\end{theorem}

\begin{theorem}\label{thm:BoundedVarianceMartingaleLowerBound}\textbf{(Theorem 6.5 in~\cite{ChungL06}) Bounded Variance Martingale: Lower Bound}\\
    Let $X_1,\cdots, X_t$ be a martingale sequence satisfying:
\begin{enumerate}
    \item $Var[X_i\mid X_1,\cdots, X_{i-1}] \leq \phi_i$
    \item $X_{i-1} - X_i \leq M$
\end{enumerate}
 $\forall i\geq 2$, where $\phi_i$ and $M$ are non-negative constants. Then for $\lambda>0$,   

$$\mathbb{P}[X_t-\mathbb{E}[X_t] \leq -\lambda]\leq \exp\left(-\frac{\lambda^2}{2\left((\sum_{i=1}^t \phi_i ) + M\lambda/3\right)}\right)$$
 
\end{theorem}

\begin{theorem}\label{thm:SuperMartingaleConcentrationVariance}\textbf{(Theorem 7.5 in~\cite{ChungL06}) Bounded Variance Supermartingale}\\
Let $X_1,\cdots, X_n$ be a supermartingale satisfying:
\begin{enumerate}
    \item $Var[X_i\mid X_1,\cdots, X_{i-1}] \leq \phi_i$
    \item $\mathbb{E}[X_i\mid X_1,\cdots, X_{i-1}] - X_i \leq M$
\end{enumerate}
$\forall i\geq 2$, where $\phi_i$ and $M$ are non-negative constants. Then for any $0<\lambda\leq X_1$, 
$$\mathbb{P}[X_t\leq X_1 - \lambda]\leq \exp\left({-\frac{\lambda^2}{2\left((\sum_{i=1}^t \phi_i) + M\lambda/3\right)}}\right)$$
\end{theorem}

\begin{theorem}\label{thm:SubMartingaleConcentrationVariance}\textbf{(Theorem 7.3 in~\cite{ChungL06}) Bounded Variance Submartingale}\\
Let $X=X_1,\cdots, X_n$ be a submartingale satisfying:
\begin{enumerate}
    \item $Var[X_i\mid X_1,\cdots, X_{i-1}] \leq \phi_i$
    \item $ X_i -\mathbb{E}[X_i\mid X_1,\cdots, X_{i-1}] \leq M$
\end{enumerate}
$\forall i\geq 2$, where $\phi_i$ and $M$ are non-negative constants. Then for $\lambda>0$,   
$$\mathbb{P}[X_t\geq X_1 + \lambda]\leq \exp\left({-\frac{\lambda^2}{2\left((\sum_{i=1}^t \phi_i) + M\lambda/3\right)}}\right)$$
\end{theorem}

\subsection{The Shifted Martingale}\label{sec:ShiftedMartingale}

In this Section we prove the following two theorem by using the shifted martingale trick (that was briefly discussed in Section~\ref{sec:martingale-overview}).

\begin{theorem}\textbf{[The Shifted Martingale Upper Bound]}\label{thm:ShiftedMartingaleUpperBound}\\
     For $t>0$, let $X_1,\cdots, X_t$ be non-negative random variables, $S_i = \sum_{j=1}^i X_j$, and $p_i = \mathbb{E}[X_i \mid X_1,\cdots, X_{i-1}]$. Let $P_1,\cdots,P_t$ and $M$ be fixed non-negative numbers, and assume that $X_i\leq M$ and $p_i\leq P_i$ for all $i\in [t]$. For $P = \sum_{i=1}^t P_i$ and $\lambda>P$, it holds that:

     $$\mathbb{P}[S_t\geq \lambda]\leq \exp{\left(-\frac{(\lambda -  P)^2}{8\cdot M\cdot P + 2M(\lambda -  P)/3}\right)}$$
\end{theorem}
 \begin{proof}
    First, we define the random variable $Y_i$ as follows. 
    
     \begin{align*}
        Y_i = \begin{cases}
            0 &\text{$i=0$}\\
             S_i - \mathbb{E}[S_i\mid Y_0,\cdots, Y_{i-1}] + Y_{i-1} &\text{$i>0$}
         \end{cases}
     \end{align*}

    \paragraph{Roadmap of the proof:} Our goal is to show that the sequence $Y_0,\cdots, Y_t$ is a martingale, prove a concentration result for $Y_t$ by using Theorem~\ref{thm:BoundedVarianceMartingaleUpperBound}, and then deduce a concentration result for $S_t$. To use Theorem~\ref{thm:BoundedVarianceMartingaleUpperBound}, we need to bound the variance $Var[Y_i\mid Y_0,\cdots, Y_{i-1}]$ and the value of $|Y_i - Y_{i-1}|$. The proof is divided into four steps. In the first step, we show that the sequence $Y_1,\cdots, Y_t$ is a martingale. In the second step  we show that $Y_i-Y_{i-1} = X_i-p_i$, which implies that $|Y_i - Y_{i-1}|\leq 2M$. In the third step, we use the property from the second step to bound the variance $Var[Y_i\mid Y_0,\cdots, Y_{i-1}]$. Finally, in the fourth step, we plug these bounds into Theorem~\ref{thm:BoundedVarianceMartingaleUpperBound} to get a concentration result for $Y_t$, and deduce the desired concentration result for $S_t$.

     \noindent\textbf{First step: the sequence \textbf{$Y_0,\cdots, Y_t$ is a martingale}.} We show that for any $i\geq 1$,  $\mathbb{E}[Y_i\mid Y_0,\cdots,Y_{i-1}] = Y_{i-1}$. Observe that:
    \begin{align*}
        \mathbb{E}[Y_i|Y_0,\cdots, Y_{i-1}] &= \mathbb{E}[S_i\mid Y_0,\cdots, Y_{i-1}] - \mathbb{E}[\mathbb{E}[S_i\mid Y_0,\cdots, Y_{i-1}]\mid Y_0,\cdots, Y_{i-1}] + \mathbb{E}[Y_{i-1}\mid Y_0,\cdots, Y_{i-1}]\\
         &= \mathbb{E}[S_i\mid Y_0,\cdots, Y_{i-1}] - \mathbb{E}[S_i\mid Y_0,\cdots, Y_{i-1}] + \mathbb{E}[Y_{i-1}\mid Y_{i-1}]\\
        &=Y_{i-1}
    \end{align*}

    \noindent where the second equality follows since for any two random variables $X,Y$, $\mathbb{E}[\mathbb{E}[X\mid Y]\mid Y] = \mathbb{E}[X\mid Y]$.
    \\
    
    \noindent\textbf{Second step: \textbf{$Y_i-Y_{i-1} = X_i-p_i$}.} The claim trivially holds for $i=1$. For $i>1$, observe that:
     \begin{align}
    Y_i - Y_{i-1} &= S_i - \mathbb{E}[S_{i-1}\mid Y_0,\cdots, Y_{i-1}]\\
     &=S_i - \mathbb{E}[S_i-S_{i-1}+S_{i-1}]\mid Y_0,\cdots, Y_{i-1}]\\
     & = S_i - \mathbb{E}[S_i-S_{i-1}\mid Y_0,\cdots, Y_{i-1}] - \mathbb{E}[S_{i-1}\mid Y_0,\cdots,Y_{i-1}]\\
     &= S_i - \mathbb{E}[S_i-S_{i-1}\mid S_1,\cdots, S_{i-1}] - \mathbb{E}[S_{i-1}\mid S_1,\cdots,S_{i-1}]\\
     &=S_i - S_{i-1} - \mathbb{E}[X_i\mid X_1,\cdots, X_{i-1}]\\
     &=X_i - \mathbb{E}[X_i\mid X_1,\cdots, X_{i-1}]\\
     &=X_i - p_i
 \end{align}

 \noindent where (3) follows from linearity of expectation, (4) follows from Observation~\ref{obs:ConditioningOnFunctionOfRV} which implies that $\mathbb{E}[S_{i-1}\mid Y_0,\cdots,Y_{i-1}] = \mathbb{E}[S_{i-1}\mid S_1,\cdots, S_{i-1}]$, since $(Y_0,\cdots, Y_{i-1})$ and $(S_1,\cdots, S_{i-1})$ are functions of each other, and (5)  follows since $\mathbb{E}[S_{i-1}\mid S_1,\cdots, S_{i-1}] = \mathbb{E}[S_{i-1}\mid S_{i-1}] = S_{i-1}$, and from another application of Observation~\ref{obs:ConditioningOnFunctionOfRV} since $(X_1,\cdots, X_{i-1})$ and $(S_1,\cdots, S_{i-1})$ are functions of each other, so conditioning on either of them is equivalent.

\setcounter{equation}{0}
 \noindent\textbf{Third step: $Var[Y_i\mid Y_0,\cdots,Y_{i-1}]\leq 4MP_i$}. Observe that 
    \begin{align}
        Var[Y_i\mid Y_0,\cdots,Y_{i-1}] &= Var[Y_i - Y_{i-1}\mid Y_0,\cdots,Y_{i-1}]\\
         &\leq \mathbb{E}[(Y_i - Y_{i-1})^2\mid Y_0,\cdots,Y_{i-1}]\\
         &\leq 2M\cdot \mathbb{E}[|X_i-p_i|\mid Y_0,\cdots,Y_{i-1}]\\
         &\leq 2M\cdot(\mathbb{E}[X_i\mid Y_0,\cdots,Y_{i-1}] + \mathbb{E}[p_i\mid Y_0,\cdots,Y_{i-1}])\\
         &=2M\cdot(\mathbb{E}[X_i\mid X_1,\cdots,X_{i-1}] + \mathbb{E}[p_i\mid X_1,\cdots,X_{i-1}])\\
         &\leq 2M(p_i + \mathbb{E}[p_i\mid X_1,\cdots,X_{i-1}])\\
         &\leq 4MP_i
     \end{align}

     \noindent where (1) follows from Observation~\ref{obs:conditionalVariance}, (2) follows from the definition of variance, (3) follows since we showed that $Y_i - Y_{i-1}= X_i - p_i$ in the second step, which also implies that $|Y_i - Y_{i-1}|\leq |X_i|+|p_i|\leq 2M$, (4) follows from the triangle inequality, linearity of expectation, and since the $X_i$'s are non-negative, (5) follows from Observation~\ref{obs:ConditioningOnFunctionOfRV}, and (7) follows since $p_i\leq P_i$ for all $i$.
    
     ~\\\\\noindent\textbf{Fourth step: concentration of $Y_t$ and $S_t$}. First, observe that $Y_0=\mathbb{E}[Y_t] = 0$. Moreover, having proved that the sequence $Y_1,\cdots, Y_t$ is a martingale, $|Y_i-Y_{i-1}|\leq 2M$, and $Var[Y_i\mid Y_0,\cdots,Y_{i-1}]\leq 4MP_i$, we can plug these bounds into  Theorem~\ref{thm:BoundedVarianceMartingaleUpperBound} to that for $\lambda'>0$ and $P=\sum_{i=1}^t P_i$:

     \begin{align*}
        \mathbb{P}[Y_t\geq \lambda']\leq \exp\left(-{\frac{\lambda'^2}{8\cdot M\cdot P + 2M\cdot \lambda'/3}}\right)
     \end{align*}

     \noindent Furthermore, since $Y_i-Y_{i-1} = X_i - p_i$ for all $i\in [t]$, it implies that $Y_t = S_t - \sum_{i=1}^t p_i$. Hence, $Y_t \geq S_t - P$, which implies that for $\lambda>P$:

     \begin{align*}
         \mathbb{P}[S_t\geq \lambda]\leq \mathbb{P}[Y_t+ P\geq \lambda]=\mathbb{P}[Y_t\geq \lambda- P]\leq \exp{\left(-\frac{(\lambda -  P)^2}{8 M\cdot P + 2M(\lambda -  P)/3}\right)}
    \end{align*}
     \noindent as desired.
 \end{proof}

\begin{theorem}\textbf{[The Shifted Martingale Lower Bound]}\label{thm:ShiftedMartingaleLowerBound}\\
     For $t>0$, let $X_1,\cdots, X_t$ be non-negative random variables, $S_i = \sum_{j=1}^i X_j$, and $p_i = \mathbb{E}[X_i \mid X_1,\cdots, X_{i-1}]$. Let $P^{\ell}_1,\cdots, P^{\ell}_t, P^h_1,\cdots, P^h_t$ and $M$ be fixed non-negative numbers, and assume that $X_i\leq M$ and $P^{\ell}_i\leq p_i\leq P^{h}_i$ for all $i\in [t]$. Let $P^{\ell} = \sum_{i=1}^t P^{\ell}_i$ and $P^h=\sum_{i=1}^t P^h_i$. For $P^{\ell}>\lambda$, it holds that:

     $$\mathbb{P}[S_t \leq \lambda] \leq \exp\left({-\frac{(P^{\ell}-\lambda)^2}{8\cdot M\cdot P^h + 2M\cdot (P^{\ell}-\lambda)/3}}\right)$$

 \end{theorem}
\begin{proof}
     The proof is very similar to the proof of Theorem~\ref{thm:ShiftedMartingaleUpperBound}. We start with defining the shifted random variable $Y_i$ similarly to how we defined it in the proof of Theorem~\ref{thm:ShiftedMartingaleUpperBound}.

    \begin{align*}
        Y_i = \begin{cases}
            0 &\text{$i=0$}\\
             S_i - \mathbb{E}[S_i\mid Y_0,\cdots, Y_{i-1}] + Y_{i-1} &\text{$i>0$}
         \end{cases}
    \end{align*}

    \noindent We showed in the proof of Theorem~\ref{thm:ShiftedMartingaleUpperBound} that the sequence $Y_1,\cdots, Y_t$ is a martingale. Next, we would like to use Theorem~\ref{thm:BoundedVarianceMartingaleLowerBound} to get a concentration result for $Y_t$, which would imply a concentration result for $S_t$. Recall that in the proof of Theorem~\ref{thm:ShiftedMartingaleUpperBound} we showed that $Var[Y_i\mid Y_0\cdots, Y_{i-1}]\leq 4MP^h_i$ (which was shown in the third step in the proof of Theorem~\ref{thm:ShiftedMartingaleUpperBound}), and that $Y_t = S_t - \sum_{i=1}^t p_i$ (which was shown at the end of the proof of Theorem~\ref{thm:ShiftedMartingaleUpperBound}). Furthermore, recall that $\mathbb{E}[Y_t] = 0$. Hence, we can plug these bounds into Theorem~\ref{thm:BoundedVarianceMartingaleLowerBound} to get that:

     \begin{align*}
         \mathbb{P}[S_t \leq \lambda] \leq \mathbb{P}[Y_t +  P^{\ell}\leq\lambda] = \mathbb{P}[Y_t \leq   -(P^{\ell}-\lambda)]
         \leq \exp\left({-\frac{(P^{\ell}-\lambda)^2}{8\cdot M\cdot P^h + 2M\cdot (P^{\ell}-\lambda)/3}}\right)
     \end{align*}
     \noindent as desired.
 \end{proof}

\section{General Graphs: Deferred Proofs from Section~\ref{sec:warmup}}\label{app:general}

\begin{proof}[\textbf{Proof of Lemma~\ref{lem:regular-matching-size}}]
Let $X$ be the set of vertices in $G$ with degrees in $((1 - \kappa)D, (1 + \kappa)D)$. Restricting the graph to these vertices yields $H = G[X]$. Then by definition, we have $|V(H)| \ge (1 - \tau_v)|V(G)|$ and $|E(H)| \ge (1 - \tau_e)|E(G)|$. Let $\Dup = (1 + \kappa)D$. Our goal is to find a large matching in $H$, which we plan to do by providing a fractional point in the matching polytope and deducing that some integral point (matching) is at least as good.

Our fractional point is $x \in \mathbb{R}^{E(H)}$ where $x_e = \frac{1}{\Dup + 1},\ \forall e \in E(H)$. We first claim that $x$ belongs to the matching polytope of $H$ by verifying that it satisfies the conditions of \Cref{thm:polytope}. The first set of inequalities is true as $x_e = \frac{1}{\Dup + 1} > 0$. Since the degree of any vertex in $H$ is at most $\Dup$ by definition, the second set of inequalities is satisfied as well. For the third set of inequalities, let $S \subseteq X$ be an arbitrary odd sized set of vertices. Note that $|E(H[S])| \leq \frac{|S|(|S|-1)}{2}$ and further since each vertex in $H$ has degree at most $\Dup$, $|E(H[S])| \leq \frac{1}{2}\sum_{v \in S} \Dup = \frac{|S| \Dup}{2}$. Thus we have,

\begin{align*}
    \sum_{e\in E(H[S])} x_e &= |E(H[S])| \cdot \frac{1}{\Dup + 1} \le \frac{|S|\min(|S|-1,\Dup)}{2(\Dup+1)}
    = \frac{|S|\min(|S|-1,\Dup) + (\Dup + 1)}{2(\Dup+1)} - \frac{1}{2}\\
    &\le \frac{\max(|S|,\Dup+1)\min(|S|-1,\Dup) + \max(|S|,\Dup + 1)}{2(\Dup+1)} - \frac{1}{2}\\
    &= \frac{\max(|S|,\Dup+1)\min(|S|,\Dup+1)}{2(D_0+1)} - \frac{1}{2}\\
    &= \frac{|S|(\Dup+1)}{2(\Dup+1)} - \frac{1}{2}
    = \frac{|S|-1}{2}
\end{align*}
and thus the third set of inequalities is also satisfied. Thus, by \Cref{thm:polytope}, $x$ is in the convex hull of indicator vectors of matchings in $H$, which implies that there exists an integral matching $M$ in $H$ with size at least $\sum_{e\in E(H)} x_e$. In particular,

\begin{align*}
    |M| &\ge \sum_{e\in E(H)} x_e = \frac{|E(H)|}{\Dup+1}
    \ge \frac{(1 - \tau_e)|E(G)|}{\Dup+1}\\
    \intertext{But we have $|E(G)| \geq \frac{1}{2} \cdot(\sum_{v \in V(H)}(1-\kappa)D) \geq \frac{1}{2} \cdot((1-\tau_v)|V(G)|(1-\kappa)D)$}
    &\ge \frac{(1 - \tau_e)(1 - \tau_v)(1-\kappa)|V(G)|D}{2(\Dup+1)}
    \ge \frac{(1 - \tau_e)(1 - \tau_v)(1 - \kappa)|V(G)|D}{2(1+\kappa)(D+1)}\\
    &\geq (1 - \tau_e)(1 - \tau_v)(1 - 2\kappa)(1 - 1/(D+1))\frac{|V(G)|}{2}\\
    &\ge (1 - \tau_e - \tau_v - 2\kappa - 1/(D+1))\frac{|V(G)|}{2}
\end{align*}
as desired.
\end{proof}

\section{Dense Graphs: Deferred Proofs from Section~\ref{sec:reg}}\label{app:dense}

\begin{proof}[\textbf{Proof of ~\Cref{claim:BoundingParameters}}]
        We start with bounding $\alpha_i$, observe that:
        \setcounter{equation}{0}
        \begin{align}
            \alpha_i &= 10\alpha_{i-1} + \Delta^{-1/600}\\
            & = 10(10\alpha_{i-2} +\Delta^{-1/600})+\Delta^{-1/600}\\
            &\vdots\\
            &\leq \sum_{j=0}^i 10^{j+1}\cdot \Delta^{-1/600}\\
            &\leq i\cdot 10^{i+1} \cdot \Delta^{-1/600}\\
            &\leq 100\log(1/\epsilon)\cdot  10^{10\log(1/\epsilon)}\cdot \Delta^{-1/600}\\
            &\leq \frac{100}{c}\log\Delta\cdot \Delta^{40/c}\cdot \Delta^{-1/600}\\
            &\leq \Delta^{41/c-1/600}\\
            &\leq \Delta^{41/98400-1/600}\\
            &\leq \Delta^{1/2400-1/600}\\
            &= \Delta^{-1/800} \leq 1/10
        \end{align}

        \noindent where (4) follows since $\alpha_0=\Delta^{-1/600}$, (6) follows since $i\leq 10\log(1/\epsilon)$, and (7) follows since $\epsilon\geq 1/\Delta^{1/c}$. Next, we bound $\delta_i$. Let $\Delta_i=\Delta/2^i$. Since $\epsilon\geq 1/(\Delta^{1/c})$ and $i\leq 10\log(1/\epsilon)$, we have that $\Delta_i = \Delta/2^i\geq \Delta/2^{10\log1/\epsilon}\geq \Delta^{0.9}$. Therefore, $2\exp(-\Delta_i^{1/100})\leq 2\exp(-\Delta^{9/1000})\leq \exp(-\Delta^{1/200})$. Hence, we get that:

    \setcounter{equation}{0}
    \begin{align}
        \delta_i &\leq \Delta^2(\delta_{i-1} + \exp(-\Delta^{1/200}))\\
        &\leq \Delta^2(\Delta^2(\delta_{i-2} +\exp(-\Delta^{1/200}))+\exp(-\Delta^{1/200}))\\
        &\vdots\\
        &=\sum_{j=0}^i \Delta^{2(j+1)}\cdot \exp(-\Delta^{1/200})\\
        &\leq 2\exp(-\Delta^{1/200}) \Delta^{2(i+1)} \\
        &\leq 2\exp(-\Delta^{1/200})\cdot  \exp({\log^2\Delta})\\
        &\leq \exp(-\Delta^{1/300})
    \end{align}

    \noindent where (5)  holds since $1/\Delta^2\leq 1/2$, and (7) holds since $i\leq 10\log(1/\epsilon)\leq \log(\Delta)/10$. 
    \end{proof}

\begin{proof}[\textbf{Proof of ~\Cref{cor:ConstantMatchingRegularGraphs}}]
    To apply \Cref{thm:MatchConstantFraction}, we need to show that at least $1/2$ of the edges in $G'$ are incident with nodes of degree smaller than $2\bar{d}$, where $\bar{d}$ is the average degree in $G'$. We start with bounding $\bar{d}$. Let $R$ be the set of nodes that are $(\alpha',\Delta')$-regular in $G'$. 

    \paragraph{Upper bounding $\bar{d}$:} Observe that:

    \begin{align*}
        \bar{d} &= \frac{\sum_{v\in V(G')} deg_{G'}(v)}{|V(G')|}\\
        & \leq \frac{|R|\cdot \Delta'(1+\alpha') + \exp(-\Delta^{1/300})|V(G')|\cdot \Delta}{n'}\\
        &\leq \Delta'(1+\alpha') + \exp(-\Delta^{1/300})\cdot \Delta\\
        &\leq \Delta'(1+2/10)\\
    \end{align*}
    
    \paragraph{Lower bounding $\bar{d}$:} Observe that:

    \begin{align*}
        \bar{d} &= \frac{\sum_{v\in V(G')} deg(v)}{|V(G')|}
         \geq \frac{|R|\cdot \Delta'(1-\alpha')}{|V(G')|}\\
        &\geq (1-\exp(-\Delta^{1/300}))\Delta'(1-\alpha') \geq (1-1/10) \Delta' (1-1/10)\\
        &\geq \Delta'(1-2/10)
    \end{align*}
    where we used that $\exp(-\Delta^{1/300}) < 1/10$ for $\Delta$ larger than $C$ (which is a sufficiently large constant).

    Hence, all the nodes that are $(\alpha',\Delta')$-regular in $G'$ have degree within a factor of $2$ from the average degree. Furthermore, the only edges that are incident to nodes with degree larger than $2\bar{d}$ are the edges incident to the nodes that are not $(\alpha',\Delta')$-regular, and these are only a $\Delta\cdot \exp{(-\Delta^{1/300}))}<<1/2$-fraction of the edges. The claim follows.
\end{proof}

\section{Proof of Theorem \ref{thm:LowerDense}}\label{app:LowerDense}

We now prove Theorem \ref{thm:LowerDense}:

\thmmainoptimal*

We use the probabilistic round elimination approach used in \cite{KS25} to prove this result. We are able to use their key results as black-boxes. First, we use the definitions of $r$-neighborhoods, $r$-flowers, matching-certified algorithms, and vertex survival probability from their paper. We refer to the reader to the preliminaries of that paper for complete formal definitions:

\begin{definition}[Definitions 10 and 11 of \cite{KS25}]
Consider any $\Delta\ge 2$ and $r\ge 0$. (Informally) an \emph{$r$-neighborhood} is a labeling of the $r$-hop neighborhood of a vertex in an infinite $\Delta$-ary tree with numbers from $[0,1]$. An \emph{$r$-flower} is a labeling of the $r$-hop neighborhood of an \emph{edge} in an infinite $\Delta$-ary tree with numbers from $[0,1]$. Let $\mc R_r$ and $\mc F_r$ denote the sets of $r$-neighborhoods and $r$-flowers respectively.

Two $r$-flowers $w$ and $w'$ are said to be \emph{incident} iff there exists an $r$-neighborhood for which $w$ and $w'$ are extensions of that $r$-neighborhood in different directions.

A function $f:\mc F_r\rightarrow \{0,1\}$ (an $r$-round randomized algorithm) is called an $r$-\emph{round matching-certified algorithm} if it has the property that for any incident $r$-flowers $w,w'$, $f(w)f(w') = 0$ (i.e. the algorithm always outputs a matching). The \emph{vertex survival probability} $P_f$ of $f$ is the probability that a given vertex\footnote{Equal for all vertices by symmetry/lack of vertex IDs.} is not matched (incident with a matched edge) by the algorithm $f$.
\end{definition}

Specifically, we use the following results from their paper:

\begin{theorem}[Theorem 3 of \cite{KS25}]\label{thm:vsp}
For any $r\ge 0$ and $r$-round matching-certified algorithm $f$, $P_f \ge C_1^{-r}$, where $C_1 = 10^{80}$.
\end{theorem}

\begin{lemma}[Lemma 1 of \cite{KS25}; adapted from prior work with the last guarantee removed]\label{lem:densehardgraph}
For any two even positive integers $n,\Delta$, there exists a graph $G_{n,\Delta}$ with the following properties:

\begin{enumerate}
    \item (Regularity) $G_{n,\Delta}$ is $\Delta$-regular.
    \item (Girth) $G_{n,\Delta}$ has girth at least $(\log_{\Delta} n)/1000$.
\end{enumerate}
\end{lemma}

For the second result, there was also a \emph{Subgraph size} property, but it turns out that this is not required for proving Theorem \ref{thm:LowerDense}. Now, we outline the proof of Theorem \ref{thm:LowerDense}. The proof is a simplified version of the proof given for the hardness of maximal matching in \cite{KS25}. Suppose, for the sake of contradiction, that there exists an $r = \log(1/\epsilon)/c$-round randomized LOCAL algorithm that outputs 1 or 0 for each edge that happens to output a $(1+\epsilon)$-approximate matching with probability at least $p$, where $p = \exp(-\sqrt{n})$. We obtain a contradiction via the following steps:

\begin{enumerate}
    \item This algorithm can be used to obtain an $(r+1)$-round matching-certified algorithm that produces a $(1+\epsilon)$-approximate matching with probability at least $p$, simply by deleting any pair of edges in the output that are incident.
    \item Supply the graph given by Lemma \ref{lem:densehardgraph} as input to the algorithm. $r+1 < (\log_{\Delta} n)/1000$ by the second part of the maximum constraining $\epsilon$, so the view supplied to this $(r+1)$-round algorithm is a $\Delta$-ary tree. Theorem \ref{thm:vsp} thus shows that each vertex is unmatched by the algorithm with probability at least $\epsilon^{1/100}$.
    \item By Chernoff with bounded dependence (Theorem \ref{thm:CorrelatedChernoff}), the number of unmatched nodes is less than $n\epsilon^{1/100}$ with probability at most $\exp(-\sqrt{n})$. But Lemma \ref{lem:regular-matching-size} implies that the algorithm has not found a $(1+\epsilon)$-approximate matching in any other case, a contradiction.
\end{enumerate}

We now formalize these steps:

\begin{proof}[Proof of Theorem \ref{thm:LowerDense}]
Suppose for the sake of contradiction that there is an $r = \log(1/\epsilon)/c$-round randomized LOCAL algorithm $\mc A$ that takes as input a $\Delta$-regular graph $H$ and outputs 1 or 0 for each edge and that happens to output a $(1+\epsilon)$-approximate matching with probability at least $p = \exp(-\sqrt{n})$. The algorithm $\mc A$ is given a tape of randomness (represented by a number chosen i.i.d. from $[0,1]$) at each vertex, along with a unique ID at each vertex, and output 1 or 0 for each edge depending on whether or not the algorithm adds that edge to the output edge set. We now obtain an $(r+1)$-round matching-certified algorithm $f_{\mc A}$ as follows: given a number $r_e\in [0,1]$ for each edge $e\in E(H)$, use $\mc A$ as follows:

\begin{enumerate}
    \item Split $r_e$ into 5 numbers $r_e^{(0)}$, $r_e^{(1)}$, $r_e^{(2)}$, $r_e^{(3)}$, and $r_e^{(4)}$. $r_e^{(i)}\in [0,1]$ is the number obtained by concatenating the bits in position $5j+i$ in $r_e$ for all integers $j\ge 0$.
    \item For each vertex $v\in V(H)$, let $r_v^{(0)} = \sum_{e\sim v} r_e^{(0)}$. Since $\Delta\ge 2$, for each edge $e = \{u,v\}$, $r_u^{(0)} \ne r_v^{(0)}$ almost surely. Thus, we may orient each edge $e = \{u,v\}$ from the smaller $r^{(0)}$ value to the larger value.
    \item For each vertex $v\in V(H)$, port number the edges $e_1,e_2,\hdots,e_{\Delta}$ incident with $v$ in increasing order by value of $r_e^{(0)}$. This is a valid order because the numbers are different almost surely.
    \item For each $v\in V(H)$, join\footnote{Every $i$th bit out of $d$ coming from $r_{e_i}^{(1)}$ or $r_{e_i}^{(2)}$} the incoming $r_e^{(1)}$s and outgoing $r_e^{(2)}$s in port order to obtain the number $s_v$. Do the same for the $r_{e_i}^{(3)}$s and $r_{e_i}^{(4)}$s to obtain the number $t_v$.
    \item Use the first $100\log n$ bits of $s_v$ as the ID for $v$ (all IDs unique with probability at least $(1 - 1/\text{poly}(n))$), and use the number $t_v$ as a randomness tape for $v$ in the algorithm $\mc A$.
    \item Apply $\mc A$ to get numbers 0,1 on all edges. For any two incident edges that are labeled with 1, change both of their labels to 0 and return the result.
\end{enumerate}

This is an $(r+1)$-round matching-certified algorithm, where the matching-certified property comes from the last step. Lastly, it outputs a $(1+\epsilon)$-approximate matching in $H$ with probability at least $(1-1/\text{poly}(n))p$, as conditioned on fixed values of $r_e^{(0)}$, $r_e^{(1)}$, and $r_e^{(2)}$ for each edge $e$, the $t_v$ values are independent random variables for different $v$, and thus success probability at least $p$ is achieved for each fixing of the first variables which leads to unique IDs (a $ > (1 - 1/\text{poly}(n))$-fraction of the time).

Now, supply the graph $G$ given by Lemma \ref{lem:densehardgraph} as input to $f_{\mc A}$. Notice that

$$r+1 \le 2\log(1/\epsilon)/c \le 2\log(n^{1/(c\log\Delta)})/c = 2(\log_{\Delta} n)/c^2 < (\log_\Delta n)/1000$$

Thus, the algorithm is being given an $(r+1)$-flower as input, which means that Theorem \ref{thm:vsp} applies to show that

$$P_{f_{\mc A}} \ge C_1^{-r-1} > \epsilon^{1/100}$$

Apply Theorem \ref{thm:CorrelatedChernoff} to show concentration of the number of unmatched vertices as follows. Let $X_v$ be the indicator variable of $v$ being matched. The dependency graph of these random variables has degree (and thus chromatic number) at most $\Delta^{2r+2} \le n^{1/250}$. The expected sum of these random variables is at most $n(1 - P_{f_{\mc A}}) < n(1 - \epsilon^{1/100})$, so plugging in $\lambda = \epsilon^{1/100}n/2$ shows that the sum of these random variables is at most $n(1 - \epsilon^{1/100}/2)$ with probability at least $1 - \exp(-2\sqrt{n})$.

By Lemma \ref{lem:regular-matching-size} applied with $\tau_e = \tau_v = \kappa = 0$, $G$ has a matching consisting of at least $(1 - 1/(\Delta+1))n/2$ edges. By the first constraint in the maximum on $\epsilon$, this is more than $(1 - \epsilon)n/2$. Since $3\epsilon < \epsilon^{1/100}$, a $(1+\epsilon)$-approximate matching is found with probability at most $\exp(-2\sqrt{n})$. Thus, $\exp(-2\sqrt{n}) \ge (1 - 1/\text{poly}(n))p$, a contradiction to the definition of $p$. Thus, the algorithm $\mc A$ cannot exist.
\end{proof}

\end{document}